\documentclass[11pt,onecolumn]{IEEEtran}
\usepackage[utf8]{inputenc}
\usepackage[version=4]{mhchem}
\usepackage{siunitx,booktabs}
\usepackage{longtable,tabularx}
\usepackage{mathrsfs}
\usepackage{commath}
\usepackage{cite}
\usepackage{cuted} 

\makeatletter
\g@addto@macro\normalsize{%
  \setlength\abovedisplayskip{4pt}
  \setlength\belowdisplayskip{4pt}
  \setlength\abovedisplayshortskip{4pt}
  \setlength\belowdisplayshortskip{4pt}
}
\makeatother

\let\norm\relax

\usepackage{xcolor}
\let\proof\relax

\usepackage{amsmath,amssymb,amsfonts,amsthm} 
\usepackage{graphicx,url} 
\usepackage{bm}
\usepackage{float}
\usepackage{makecell}
\usepackage{enumitem}
\usepackage{subcaption}
\usepackage{wrapfig}
\usepackage{algorithm,algorithmicx} 
\usepackage{algpseudocode} 
\usepackage{multirow}
\usepackage{diagbox}

\usepackage{soul}  
\newcommand\rurl[1]{%
  \href{http://#1}{\nolinkurl{#1}}%
}

\usepackage{comment}

\def\red#1{{\color{red}#1}}

\def\cyan#1{{\color{cyan}#1}}

\let\abs\relax
\newcommand{\abs}[1]{\left\lvert#1\right\rvert}
\let \norm \relax 
\newcommand{\norm}[1]{\left\lVert#1\right\rVert}

\newcommand{\infnorm}[1]{\left\lVert#1\right\rVert_\infty}
\newcommand{\linfnorm}[1]{\left\lVert#1\right\rVert_{\mathcal{L}_\infty}}

\newcommand{\linfnormtruc}[2]{\left\lVert#1\right\rVert_{\mathcal{L}_\infty^{[0,#2]}}}

\newcommand{\lonenorm}[1]{\left\lVert#1\right\rVert_{\mathcal{L}_1}}
\newcommand{\lonenormbar}[1]{\left\lVert#1\right\rVert_{\bar{\mathcal{L}}_1}}
\def\lone{{\mathcal{L}_1}}
\def\lonebar{{\bar{\mathcal{L}}_1}}
\def\lonew{${\mathcal{L}_1}$ }
\def\linf{{\mathcal{L}_\infty}}

\def\rhor{{\rho_r}}
\def \loneAC {$\lone$AC}

\def\reft{{{\textup{ref}}}}

\def \idt{{\textup{id}}}

\def\tilx{\tilde{x}}

\def \hatx{\hat{x}}

\def \hsigma{\hat{\sigma}}
\def\xin{x_\textup{in}}
\def\xref{x_\reft}

\def \rt {\textup{r}}

\def\mbR{\mathbb{R}}

\def\mbI{\mathbb{I}}
\def\mbZ{\mathbb{Z}}

\def\mbZ{\mathbb{Z}}

\def\hsigma{\hat{\sigma}}

\def\mcC{\mathcal{C}}

\def\mcF{\mathcal{F}}

\def\mcH{\mathcal{H}}

\def\mcG{\mathcal{G}}
\def\mcW{\mathcal{W}}
\def\mrx{\mathrm{x}}
\def\mru{\mathrm{u}}
\def\mry{\mathrm{y}}
\def\mrA{\mathrm{A}}
\def\mrB{\mathrm{B}}
\def\mrC{\mathrm{C}}
\def\mrD{\mathrm{D}}

\def\mrP{\mathrm{P}}

\def\trieq{\triangleq}

\newtheorem{theorem}{Theorem}
\newtheorem{lemma}{Lemma}

\theoremstyle{definition}  \newtheorem{definition}{Definition}
\theoremstyle{definition} 
\newtheorem{assumption}{Assumption}
\theoremstyle{remark}  
\newtheorem{remark}{Remark}

\usepackage{accents}

\makeatletter
\def\cl@part {\@elt {chapter}}
\makeatother

\def \lrangle#1{\left\langle#1\right\rangle}

\makeatletter
\renewcommand*\env@matrix[1][\arraystretch]{%
  \edef\arraystretch{#1}%
  \hskip -\arraycolsep
  \let\@ifnextchar\new@ifnextchar
  \array{*\c@MaxMatrixCols c}}
\makeatother

\def \proof{\noindent{\it Proof}. }

\def\lammax#1{{\bar{\lambda}(#1)}}
\def\lammin#1{{\underline{\lambda}(#1)}}
\def\lammaxtheta#1{{\max_{\theta\in \Theta}\lammax{#1(\theta)}}}
\def\lammintheta#1{{\min_{\theta\in \Theta}\lammin{#1(\theta)}}}

\newcommand{\linfnormtrucinterval}[2]{\left\lVert#1\right\rVert_{\mathcal{L}_\infty^{[#2]}}}

\newcommand{\ssreal}[4]{ \arraycolsep=1.8pt \renewcommand{\arraystretch}{0.95}
\left[\begin{array}{c|c}
      #1 & #2 \\
     \hline
     #3 & #4
\end{array}\right]}

\def\te{\textup{e}}
\def\xin{x_\textup{in}}
\def\tileta{\tilde{\eta}}
\def\tilalpha{\tilde{\alpha}}
\def\tilsigma{{\tilde{\sigma}}}
\def\tilq{\tilde{q}}
\def\tildeltae{\tilde{\delta}_e}
\def\bbracket#1{\bm{[}#1\bm{]}}
\def\hsigmaum{\hsigma_{um}}
\def\uum{u_{um}}
\def \rt{\textup{r}}

\def\barH{{\bar{\mathcal{H}}}}
\def\Bu{B_{u}}

\def \total{\textup{total}}
\def \bl{\textup{bl}}

\def\l1rho{L_{1\rho}}
\def\l2rho{L_{2\rho}}

\def\b10{b_{10}}
\def\b20{b_{20}}
\def\bi0{b_{i0}}
\def\barqs{\bar{q}_s}

\usepackage{hyperref}
\usepackage{xcolor}
\hypersetup{
    colorlinks= true,
    pdfborderstyle={/S/U/W 1},
    linkcolor= green,
    filecolor= magenta,      
    urlcolor= cyan,
    linkbordercolor=red,
    citecolor  = black
}

\usepackage{cleveref}

\crefname{equation}{}{} 
\crefname{lemma}{Lemma}{Lemmas}
\crefname{theorem}{Theorem}{Theorems}
\crefname{table}{Table}{Tables}
\crefname{figure}{Fig.}{Figs.}
\crefname{remark}{Remark}{Remarks}
\crefname{assumption}{Assumption}{Assumptions}
\crefname{section}{Section}{Sections}
\crefname{definition}{Definition}{Definitions}
\crefname{algorithm}{Algorithm}{Algorithms}
\crefname{proposition}{Proposition}{Propositions}
\crefname{appendix}{Appendix}{Appendices}

\newcolumntype{L}[1]{>{\raggedright\arraybackslash}p{#1}}
 
\newcolumntype{C}[1]{>{\centering\arraybackslash}p{#1}}
 
\newcolumntype{R}[1]{>{\raggedleft\arraybackslash}p{#1}}

\def\qedclosed{\hfill$\blacksquare$}

\def\red#1{#1} 

\def\arraystretch{1.2}

    \makeatletter
\def\@fnsymbol#1{\ensuremath{\ifcase#1\or *\or \ddagger\or
  \mathsection\or \mathparagraph\or \|\or **\or \dagger\dagger
  \or \ddagger\ddagger \else\@ctrerr\fi}}
    \makeatother

\begin{document}

\title{\LARGE Robust Adaptive Control of Linear Parameter-Varying Systems with Unmatched Uncertainties}
\vspace{-.7cm}

\author{Pan Zhao$^1$, 
Steven Snyder$^2$, 
Naira Hovakimyan$^3$, 
and
Chengyu Cao$^4$
\thanks{This work is in part by the Air Force Office of Scientific Research through Grant FA9550-18-1-0269, and in part by the NASA Langley Research Center through Grant 80NSSC17M0051. 
(Corresponding author: Pan Zhao)}
\thanks{$^1$P.~Zhao is with the Department of Aerospace Engineering and Mechanics, University of Alabama, Tuscaloosa, AL, USA 35487 (e-mail:~pan.zhao@ua.edu)}
\thanks{$^2$S.~Snyder is with the Dynamics Systems and Control Branch, NASA Langley Research Center, Hampton, VA, USA 23681 (e-mail:~steven.m.snyder@nasa.gov)}
\thanks{$^3$N.~Hovakimyan is with the Department of Mechanical Science and Engineering, University of Illinois at Urbana-Champaign, Urbana, IL, USA 61801 (e-mail:~nhovakim@illinois.edu)}
\thanks{$^4$C.~Cao is with the Department of Mechanical Engineering, University of Connecticut, Storrs, CT, USA 06269 (e-mail:~chengyu.cao@uconn.edu)} \vspace{-.7cm}}

\maketitle
\begin{abstract}
In controlling systems with large operating envelopes, it is often necessary to adjust the desired dynamics according to operating conditions. This paper presents a robust adaptive control architecture for linear parameter-varying (LPV) systems that allows for the
desired dynamics to be systematically scheduled, while being able to handle a broad class of
uncertainties, both matched and unmatched, which can depend on both time and states. The proposed controller adopts an \lonew adaptive control architecture for designing the adaptive
control law and peak-to-peak gain (PPG) minimization for designing the robust control law
to mitigate the effect of unmatched uncertainties. 
{Leveraging the PPG bound of an LPV system,} we derive transient and steady-state performance bounds in terms of the input and output signals of the actual closed-loop system as compared to the same signals of a nominal system. 
The efficacy of the proposed method is validated by extensive simulations using the short-period dynamics of an F-16 aircraft operating in a large envelope. 
\end{abstract}


\section{Introduction}\label{sec:introduction}

Gain scheduling (GS) is a common practice for controlling nonlinear and time-varying systems with large operating envelopes \cite{Rugh00gs_survey}.  
As a systematic way for GS design and analysis, the {\em linear parameter-varying} (LPV) systems approach has been extensively studied in the last three decades \cite{Shamma91AM,Wu96Induced,Apk98,Sato11AM},
with numerous applications
, e.g.,~in aerospace \cite{biannic1997parameter,Lu06F16,he2018lpv-blended,he2019smooth-slpv}
and automotive systems \cite{wei2007lpv-engine,white2011hardware}.
On the other hand, model uncertainties and exogenous disturbances are inherent in real-world systems, which  lead to degraded stability and performance if not addressed properly.
{\it Adaptive control methods} have presented a number of promising solutions for controlling uncertain systems under certain assumptions with different guarantees.
In the last several decades, the field of adaptive control has witnessed tremendous developments, capturing different classes of linear and nonlinear systems, subject to various parametric and state-dependent uncertainties, and unmodeled dynamics \cite{ioannou2012robust,narendra1989adaptive-book,lavretsky2013robust-adaptive-book}. 
  Despite the significant progress made in the past several decades in model reference adaptive control (MRAC), the solutions are limited to parametric or structured nonlinear uncertainties \cite{ioannou2012robust,lavretsky2013robust-adaptive-book}.
{\it Disturbance-observer (DOB) based control} and related methods such as active disturbance rejection control  (ADRC) \cite{han2009adrc} 
are a group of promising approaches for controlling linear and nonlinear systems subject to uncertainties and disturbances \cite{chen2015dobc}. Hereafter, we use {\it DOBC} to refer to both DOB based control and the related methods including ADRC.  
The main idea of DOBC is to lump model uncertainties, unmodeled dynamics, and external disturbances together as a ``disturbance'', estimate it via an observation mechanism, 
 and then compute control actions to compensate for the estimated disturbance. {However, DOB design often needs inversion of the nominal dynamics (for frequency-domain design methods) which are generally only applicable to minimum-phase LTI systems, or an exogenous system with known parameters that is assumed to generate the disturbances (for time-domain design methods). Moreover, 
 state-dependence of uncertainties are generally ignored in theoretic analysis of DOBC \cite{chen2015dobc}. }  
As an alternative and promising solution for controlling uncertain systems, the {\lonew adaptive control} (\loneAC) \cite{naira2010l1book} can handle a broad class of uncertainties without needing a parametric structure, and 
allows for the use of fast adaptation without sacrificing robustness thanks to the decoupling of the control loop from the estimation loop. 
Further, \loneAC~provides guaranteed tracking performance, both in {\it transient} and steady-state, and guaranteed robustness margins (e.g., nonzero time-delay margin), in the presence of various uncertainties. \loneAC~has been verified on many real-world systems 
 including
NASA's subscale commercial jet \cite{naira2011L1_Safety}, 
and~Calspan's  Learjet (a piloted aircraft) \cite{ackerman2017evaluation,ackerman2019Learjet}. {It has also been adopted within the Learn-to-Fly framework of NASA and validated by flight tests \cite{snyder2022l2f}}. 

Despite the significant progress made in the past several decades, existing adaptive control and DOBC schemes generally employ a linear time-invariant (LTI) system to describe the desired or nominal dynamics. 
However, using an LTI model for a system with a large operating envelope
could lead to large uncertainties in certain operating regions. Large uncertainties will cause large
compensation actions, which could be energy inefficient or even cause actuator saturation and consequently loss of stability or degraded performance. Additionally, in certain applications, it is preferable to change
the desired dynamics according to the operating condition. For instance, for the longitudinal dynamics of manned aircraft, pilots usually prefer faster response at higher speeds and slower response at lower speeds during up-and-away flight \cite{biannic1997parameter}.
Under these scenarios, using an {LPV} model to characterize the {desired dynamics} provides the desired perspective.  {Nonlinear models have also been employed in adaptive and DOB controllers to describe the desired dynamics \cite{wang2012L1-nonlinear,chen2015dobc,lopez2020adaptive-ccm,lakshmanan2020safe}, but are probably not as appropriate for characterizing the variation of desired dynamics according to operating conditions as an LPV formulation does. Furthermore, for complex systems such as aerospace systems, it may be challenging to get a nonlinear analytic model for the whole operating envelope, e.g., due to the use of look-up tables to model sophisticated aerodynamics. In contrast, it is much easier to obtain an LPV model through interpolation of a series of LTI models resulting from linearization at different operating conditions. 
However, adaptive or DOB-based control of LPV systems has rarely been explored. \cite{xie2016MRAC-LPV} studied an MRAC scheme for switched LPV systems,  but considered only {matched parametric} uncertainties with only asymptotic performance guarantees.   \cite{snyder2020l1lpv} presented an \loneAC~controller for LPV systems in the presence of matched parametric uncertainties and external disturbances.


It remains a challenging problem in adaptive control or DOBC to handle unmatched uncertainties (UUs), i.e., those that enter the system through different channels from the control inputs. In general, it is impossible to remove the influence of the UUs on system states; therefore, almost all of the existing efforts focus on eliminating or attenuating the effect of UUs on interested variables, e.g., system outputs. 
For nominal LTI systems that are minimum-phase and square (i.e., having the same number of input and output channels), \cite{xargay2010unmatched} used the nominal system inversion to design a {\it dynamic} feedforward compensation term to remove the effect of the UUs on the output channels in \loneAC. Similar ideas were explored independently in DOBC, in which {\it static} feedforward terms were designed to eliminate the effect of UUs on the outputs in {\it steady state} only for
LTI systems \cite{li2012extendedESO} and nonlinear systems \cite{yang2012nonlinear-unmatched}. 
$H_\infty$ control theory was utilized to design state-feedback gains to attenuate the effect of UUs  
in DOBC \cite{wei2010composite-unmatched}.
For systems transformable to strict or semi-strict feedback forms, backstepping techniques have been leveraged to handle UUs in adaptive control \cite{yao1997adaptive,koshkouei2004backstepping}, and in DOBC \cite{sun2015global-unmatched}. 
{Alternative approaches based on learning the UUs and incorporating the learned model to adapt the reference model were investigated in linear MRAC settings \cite{yayla2016adaptive-unmatched,joshi2019hybrid}. 
 However, the tracking performance during the learning transients can be quite poor, due to a lack of active compensation for the UUs.  }

\underline{\it Statement of Contributions.} 
The contributions of this paper are threefold. To facilitate systematic scheduling of desired dynamics in controlling uncertain aerospace systems with large operating envelopes, this paper presents a robust adaptive control scheme for LPV systems that can
 handle a broad class of uncertainties that can depend on both time and states and can be unmatched. 
The proposed control scheme leverages the $\mathcal{L}_1$AC~architecture to design the adaptive control law. 
A dynamic feedforward LPV mapping, computed using linear matrix inequality (LMI) techniques, is introduced to attenuate the effect of unmatched uncertainties.
This approach is more general and does not need restrictive assumptions such as strict-feedback form, system squareness and non-minimum phase, needed in \cite{yao1997adaptive,koshkouei2004backstepping,sun2015global-unmatched,xargay2010unmatched}. 
In addition, we derive uniform transient performance bounds in terms of the inputs and states of the robust adaptive system as compared to the same signals of a nominal system. We furthermore show that the transient performance bounds can be {\it systematically improved} by reducing the estimation sampling time and increasing the \red{bandwidth of the low-pass filter (an important element of an \lonew adaptive controller, to be introduced in Section~\ref{sec:l1_architecture})}.  
{Finally, we apply the proposed method to control the longitudinal motion of an F-16 aircraft operating in a large envelope, and conduct extensive simulations using both LPV  and {fully nonlinear} models to validate the control performance}.

The paper is organized as follows. Section~\ref{sec:problem_formulation} introduces the main problem and key assumptions. Section~\ref{sec:prelinimaries} includes basic definitions and preliminary lemmas. Section~\ref{sec:l1_architecture} explains the proposed robust adaptive control architecture, while the performance analysis is given in Section~\ref{sec:analysis_l1}. 
Section~\ref{sec:simulation-example} includes the simulation results on the short-period dynamics of an F-16 aircraft. 

We denote by $\mathbb{R}^n$  the set of $n$-dimensional real vectors, $\mathbb{R}^+$ the set of non-negative real numbers, $\mathbb{R}^{m\times n}$ the set of real $m$ by $n$ matrices.
$\mbI_n$ denotes an $n$ by $n$ identity matrix, and {${0}$ denotes a zero matrix of compatible dimension}. $\mbZ_i$ and $\mbZ_1^n$ denote the integer sets $\{i, i+1, \cdots\}$ and $\{1, 2,\cdots,n\}$, respectively.
$\norm{\cdot}$, $\norm{\cdot}_1$ and $\norm{\cdot}_\infty$ to  denote the $2$-norm, $1$-norm and $\infty$-norm of a vector or a matrix, respectively.  $\bar \lambda (P)$ and $\underline \lambda(P)$ denote the maximal and minimal eigenvalues of the matrix $P$, respectively. For symmetric matrices $P$ and $Q$, $P>Q$ means $P-Q$ is positive definite. $\langle X\rangle$ is the shorthand notation of $X +X^{\top\!}$. 
In symmetric matrices, $\star$ denotes an off-diagonal block induced by symmetry. For a matrix $A\in \mbR^{m\times n}$, $\mathrm{vec}(A)\in R^{mn}$ represents the vectorization of $A$.
The Laplace transform of a function $x(t)$ is denoted by $x(s)\triangleq\mathfrak{L}[x(t)]$.
The mapping corresponding to a state-space realization $(A,B,C,D)$ is denoted by \scalebox{0.8}{${ \ssreal{A}{B}{C}{D}}$}.


\section{Problem Statement}\label{sec:problem_formulation}
Consider the class of multiple-input multiple-output LPV systems subject to uncertainties
  {\begin{subequations}\label{eq:plant-dynamics}
\begin{align}
  \dot{x}(t) = &~ A(\theta(t)) x(t)+ B(\theta(t))\omega u_\textup{total}(t)+ f(t,x(t)) \label{eq:plant-dynamics-1}\\
    y(t) = &~ C_m(\theta(t)) x(t), \ x(0) = x_0 ,
\end{align}
\end{subequations}}where $x(t)\in\mbR^n$ is the state vector, $u_\textup{total}(t)\in\mbR^m$ is the control input vector, $y(t)\in \mbR^p$ is the output vector,  $\theta(t)\triangleq[\theta_1(t),\cdots,\theta_s(t)]\in\mathbb{R}^s$ is a vector of time-varying parameters that can be measured online, {$A(\theta(t))\in\mbR^{n\times n}$}, $B(\theta(t))\in\mbR^{n\times m}$, and $C_m(\theta(t))\in \mbR^{m\times n}$  are known matrices dependent on $\theta(t)$. {The nominal model $\dot x(t) = A(\theta(t))x(t) + B(\theta(t))u_\total(t), y(t) = C_m(\theta(t))x(t)$ is termed as an LPV model and can be used to characterize a nonlinear system in a large operating envelope \cite{Rugh00gs_survey}.} 
Additionally, $\omega\in \mbR^{m\times m}$ is an unknown matrix {introduced to model actuator failures or modeling errors, which could lead to uncertain control gains or incorrectly estimated control effectiveness \cite[Section 9.5]{lavretsky2013robust-adaptive-book}}. Moreover, $f: \mbR\times \mbR^n\rightarrow \mbR^n$ is an unknown nonlinear function {that can represent model uncertainties, unknown parameters and external disturbances \red{\cite[Section 2.4]{naira2010l1book}}}. The initial state vector $x_0$ is assumed to be inside an {arbitrarily large} known set, i.e., for some known $\rho_0>0$, $
    \norm{x_0}\leq \rho_0 <\infty.
$ {Hereafter, for the brevity of notation, we often omit the dependence of $\theta(t)$, $x(t)$ and $u(t)$ on $t$. Given \cref{eq:plant-dynamics}, we would like to design a controller to ensure stability of the closed-loop system and to track a given bounded piecewise-continuous reference signal $r(t)\in \mbR^{p}$ (which satisfies $\linfnorm{r}\leq \bar r$ for a known constant $\bar r$) with prescribed performance.}  


\begin{assumption}\label{assump:unknown-input-gain}
{(Partial knowledge of input gain)} The input gain matrix $\omega$ is assumed to be an unknown strictly row-diagonally dominant matrix with $\textup{sgn}(\omega_{ii})$ known. Furthermore, we assume that $\omega\in \Omega\trieq  \{\omega\in \mbR^{m\times m}: \mathrm{vec}(w) \in \hat{\Omega}\}$ with $\hat\Omega\subset \mbR^{m^2}$ being a convex polytope. Without loss of generality, we further assume that $\mbI_m\in \Omega$.
\end{assumption}
{
\begin{remark}
The first statement in \cref{assump:unknown-input-gain} indicates that $\omega$ is always non-singular with known sign for the diagonal elements, and is essentially the same as the assumptions made in existing relevant work in adaptive control, e.g, \cite{lavretsky2010predictor-mrac} \cite[Sections 6 and 7]{ioannou2012robust}. {The choice of the form for $\omega$ is motivated by aerospace applications where control directions are known but their magnitudes and coupling are uncertain \cite{lavretsky2010predictor-mrac}}. 
\end{remark}}
{\begin{assumption}\label{assump:ft0-bound}
{(Uniform boundedness of $f(t,0)$)} 
There exist constants $b_f^0$ such that $\norm{f(t,0)}\leq b_f^0$, for any $ t\geq 0$.
\end{assumption}
\begin{assumption}
{(Semiglobal Lipschitz continuity of $f(t,x)$)} \label{assump:ftx-semiglobal_lipschitz}
For arbitrary $\delta>0$, there exist positive constants $L_f^\delta$ such that $
    \norm{f(t,x_1)-f(t,x_2)} \leq L_f^\delta \norm{x_1-x_2},
$
for any $t\geq 0$ and any $x_j$ satisfying $\norm{x_j}\leq \delta,\ j=1,2$.
\end{assumption}
\begin{assumption}\label{assump:ftx-rate-bounded}
{(Uniform boundedness of rate of variation of $f(t,x)$ with respect to $t$)}  There exists positive constant $l_f$ such that for any $x$ and $t_1,t_2\geq 0$,
$
    \left\| {f(t_1,x) - f(t_2,x)} \right\| \le {l_f}\abs{t_1-t_2}.
$
\vspace{-2mm}
\end{assumption}}
\begin{remark}
Roughly speaking, 
\cref{assump:ftx-semiglobal_lipschitz,assump:ftx-rate-bounded} require the {\it rate of variation} of $f(t,x)$ with respect to both $x$ and $t$ to be bounded, which are less restrictive than requiring the uncertainties to be slowly varying (i.e.,  $\dot f \simeq 0$) \cite{yang2012nonlinear-unmatched} or have constant values in steady state (i.e., $\lim_{t\rightarrow \infty}\dot f = 0$)
\cite{li2012extendedESO,sun2015global-unmatched} often needed in \red{disturbance observer-based control}. 
\end{remark}

\begin{assumption}\label{assum:theta-thetadot-bounded}
The parameter vector $\theta(t)$ and its derivative are in compact sets, i.e., 
\begin{equation}\label{eq:theta-thetadot-constraints}
    \begin{split}
    \theta(t)\in\Theta,\quad 
       \dot{\theta}(t)\in\Theta_d, \quad 
    \max_{\dot{\theta}\in\Theta_d} \norm{\dot{\theta}(t)}\leq b_{\dot{\theta}},
  \end{split}
\end{equation}
where $\Theta\trieq[\underline\theta_1,\bar\theta_1]\times \cdots \times [\underline\theta_s,\bar\theta_s] $ and $\Theta_d \trieq[\underline\nu_1,\bar\nu_1]\times \cdots \times [\underline\nu_s,\bar\nu_s]$ are known, and $b_{\dot{\theta}}$ is a known constant.
\end{assumption}

{
Under \cref{assum:theta-thetadot-bounded}, for the nominal or uncertainty-free plant (corresponding to the plant \cref{eq:plant-dynamics} with $\omega=\mbI_m$ and $f(t,x(t))=0$), we can design 
a baseline controller 
\begin{equation}\label{eq:control-baseline-general}
    u_\bl(t) = K_x(\theta)x(t) ,
\end{equation}
where $K_x(\theta) \in \mbR^{m\times n}$ is a parameter-dependent (PD) feedback gain matrix to ensure closed-loop stability. 
\begin{remark}
\cref{assum:theta-thetadot-bounded} is standard in analysis and control synthesis for LPV systems \cite{Wu96Induced,Apk98,Sato11AM,Apk95Self}, and therefore is not restrictive. In the absence of uncertainties, under \cref{assum:theta-thetadot-bounded}, it is straightforward to design the baseline controller \cref{eq:control-baseline-general} using existing methods, e.g., via the PD Lyapunov functions and linear matrix inequalities (LMIs) \cite{Apk98,Wu96Induced,Sato11AM} and \red{\cite[Section~1.3]{Moh12LPVBook}}. 
\end{remark}

With the nominal plant and the baseline controller \cref{eq:control-baseline-general},  we can obtain the following {\it ideal} closed-loop system
\begin{subequations}\label{eq:ideal-dynamics}
\begin{align}
\dot{x}_\textup{id}(t) & =   A_m(\theta)x_\textup{id}(t) + B(\theta)u_\textup{id}(t),\ x_\idt(0) = x_0, \label{eq:ideal-dynamics-a} \\
u_\textup{id}(t) & = K_r(\theta) x_\textup{id}(t),\quad 
y_\textup{id}(t)  = C_m(\theta) x_\textup{id}(t), \label{eq:ideal-dynamics-b}
\end{align}
\end{subequations}
where $A_m(\theta) \trieq A(\theta) + B(\theta) K_x(\theta)$ and $K_r(\theta)\in \mbR^{m\times p}$ is a  parameter-dependent (PD) feedforward gain matrix for achieving desired performance in tracking $r(t)$ in the absence of uncertainties.

\begin{remark}
For square systems, i.e., $m=p$,  we can set $K_r(\theta)= -\left(C_m(\theta) A_m^{-1}(\theta)B(\theta)\right)^{-1}$ to ensure zero steady-state error when tracking a step reference command for any fixed $\theta$ value.  For non-square systems, $K_r(\theta)$ can be designed using other techniques such as $H_\infty$ control. 
\end{remark}
\red{\begin{remark}
   The state feedback $K_x(\theta)x$ plus the feedforward $K_r(\theta)r$ is just one way to design the baseline controller that achieves the desired tracking performance in the absence of uncertainties. One can employ other methods to design the baseline controller, e.g., introducing the integration of tracking error as additional states and designing a state-feedback controller for the augmented system (as done \cite[Section~20]{lavretsky2013robust-adaptive-book}), which could lead to improved robustness.
\end{remark}}
\begin{assumption}\label{assump:desired-dynamics-stable-Lyapunov} {(Lyapunov stability of the ideal CL system)} 
{There exists a PD symmetric matrix $P(\theta)$ and a positive constant $\mu_P$ such that}
{
\begin{equation}\label{eq:laypunov_stability_Am}
\begin{split}
        \left\langle \! A_m^{\top\!} (\theta)P(\theta)\!\right\rangle +\dot{P}(\theta) & \!\leq\! -\mu_P P(\theta), \ \forall (\theta,\dot{\theta}) \in \Theta \!\times \!\Theta_d,\\
      P(\theta)&\!>\!0,\ \forall \theta \in \Theta.
\end{split}
\end{equation}}
\end{assumption}
\begin{remark}\label{rem:desired-dynamics-stable-Lyapunov}
{\cref{assump:desired-dynamics-stable-Lyapunov} ensures  stability of \cref{eq:ideal-dynamics} \cite{Wu96Induced}}.
Indeed, given the closed-loop system \cref{eq:ideal-dynamics-a}, the matrix $P(\theta)$ that satisfies \cref{eq:laypunov_stability_Am} (which defines a PD Lyapunov function $V(t,\theta) \trieq x^{\top\!}(t)P(\theta) x(t)$) can be easily searched for via solving a parameter-dependent LMI (PD-LMI) problem \cite[Section~1.3]{Moh12LPVBook,Wu96Induced}. Moreover, the matrix $P(\theta)$ that satisfies \cref{eq:laypunov_stability_Am} is automatically generated if the baseline controller is designed using PD Lyapunov functions \cite{Wu96Induced,Apk98} (see also the example in \cref{sec:sub-lpv-modeling-f16}).
Therefore, \cref{assump:desired-dynamics-stable-Lyapunov} is not restrictive \red{and can be verified}. 
\end{remark}
Given the baseline controller in \cref{eq:control-baseline-general}, we adopt a compositional control law in the form of 
\begin{equation}\label{eq:control-composition}
   u_\textup{total}(t) = u_\bl(t) + u(t), 
\end{equation}
where $u(t)$ is a robust adaptive control law to be designed to compensate for the uncertainties and track the reference signal $r(t)$.}  The goal of this paper is to {\it design a state-feedback robust adaptive control law to ensure that  $y(t)$\red{, the output of the uncertain system (1),}  {stays close to }  $y_\textup{id}(t)$, 
the output of the ideal CL system \cref{eq:ideal-dynamics},
with quantifiable bounds on the transient and steady-state performance}. 

\begin{assumption}\label{assump:B-lipschitz-bounded}
$B(\theta)$ has full column rank for all $\theta\in\Theta$, 
and there exist positive constants $L_{B}$, {$L_{B^\dag}$} and {$L_{K_x}$} such that 
\begin{align}
    \norm{B(\theta^1)-B(\theta^2)} &\leq L_{B}\norm{\theta^1-\theta^2}, \label{eq:B-lipschitz-const} \\
    {\norm{B^\dagger(\theta^1)-B^\dagger(\theta^2)}}& {\leq L_{B^\dag}\norm{\theta^1-\theta^2} \label{eq:Bdagger-lipschitz-const}}, \\
  {\norm{K_x(\theta^1)-K_x(\theta^2)}}& {\leq  L_{K_x}\norm{\theta^1-\theta^2}, \label{eq:K-lipschitz-const}}
\end{align}
for any $\theta^i\in\Theta$, $i=1,2$, where $B^\dagger (\theta)$ denotes the pseudo-inverse of $B(\theta)$, and {$K_x(\theta)$ is the gain of the baseline controller \cref{eq:control-baseline-general}.}
\end{assumption}
\vspace{-1mm}

\section{Preliminaries and Basic Lemmas}\label{sec:prelinimaries}
In this section, we introduce some definitions, notations, and lemmas for deriving the main results. Hereafter, for notation compactness, we often omit the time dependence of $\theta$ and $x$.

Under \cref{assump:B-lipschitz-bounded}, let us define
\begin{equation}\label{eq:Bbar-theta-defn}
    \bar B(\theta) \triangleq \left[B(\theta), \Bu(\theta)\right],
\end{equation}
where $\Bu(\theta)\in \mbR^{n\times(n-m)}$ is a 
 matrix dependent on $\theta$ such that 
 $\textup{rank}[\bar B(\theta)] = n$ and $B^{\top\!} (\theta)\Bu(\theta) = 0$ for any $\theta \in \Theta$.  
On the other hand, since $\textup{rank}[\bar B(\theta)] = n$ and $B^{\top\!} (\theta)\Bu(\theta)) = 0$ for any $\theta \in \Theta$, we can always find  {vector-valued functions $f_i(t,x(t))$ ($i=1,2$) such that 
\begin{equation}\label{eq:f1-f2-f-relation}
\bar B(\theta)[f_1^{\top\!}(t,x), f_2^{\top\!}(t,x)]^{\top\!} = B(\theta)f_1(t,x)+\Bu(\theta)f_2(t,x)=\! {\bar f}(t,x)\!\trieq\!  f(t,x) \!+\! B(\theta) (\omega\!-\!\mbI_m)K_x(\theta)x\!
\end{equation}
holds for any $t$ and any $\theta \in \Theta$, where the dependence of $x(t)$ and $r(t)$ on $t$ has been omitted for brevity. 
With \cref{eq:control-composition,eq:control-baseline-general,eq:f1-f2-f-relation}, the  system \eqref{eq:plant-dynamics} can be rewritten as 
\begin{subequations}\label{eq:plant-dynamics-f12} 
\begin{align}
 \hspace{-2.5mm}  \dot{x}(t)  \!= & ~ A_m(\theta)x(t)\!+\! B(\theta)\omega  u(t) \!+\! {\bar f}(t,x(t)) \label{eq:plant-dynamics-f12-a-fhat} \\  
   \! = &~A_m(\theta)x(t)\!+\! B(\theta)\left(\omega  u(t) \!+\! f_1(t,x(t))\right) 
      +\!\Bu (\theta)f_2(t,x(t)),  \label{eq:plant-dynamics-f12-b}   \\
 \hspace{-2.5mm}  y(t)
\!= & ~ C_m(\theta) x(t),\ x(0) = x_0. \label{eq:plant-dynamics-f12-c} 
\end{align}
\end{subequations}}

Since  the set $\Theta$ (introduced in \cref{assum:theta-thetadot-bounded}) and the parameter-dependent matrices are all known, we can define and compute the following constants:
\begin{align}
 &  b_{A_m}\!\!\trieq\! \max_{\theta\in\Theta}\!\norm{A_m(\theta)},  \  b_{B} \!\trieq\! \max_{\theta\in\Theta}\!\norm{B(\theta)},\    b_{B^\dagger}\!\trieq\! \max_{\theta\in\Theta}\!\norm{B^\dagger(\theta)},\nonumber\\ 
  &b_C \!\trieq\! \max_{\theta\in\Theta}\!\norm{C_m(\theta)},\   b_{\Bu}\!\trieq\!\max_{\theta\in\Theta}\! \norm{\Bu(\theta)}, \   b_{\Bu^\dagger}\!\trieq\!\max_{\theta\in\Theta}\! \norm{\Bu^\dag(\theta)}, \nonumber \\ 
  & { b_{K_x} \!\trieq\! \max_{\theta\in\Theta}\!\norm{K_x(\theta)} , \ b_{K_r}\!\trieq\! \max_{\theta\in\Theta}\! \norm{K_r(\theta)}}, \label{eq:Am-B-Bu-Kx-Kr-bounds}
\end{align}
where $\Bu^\dag(\theta)$ denotes the pseudo-inverse of $\Bu(\theta)$ introduced in \cref{eq:Bbar-theta-defn}.
Next, we present a lemma stating some properties of $f_i(t,x)$ under the assumptions introduced in \cref{sec:problem_formulation}. 
The proof is given in Appendix~\ref{sec:prof-lem:f-f1-f2-constants-relation}.
{\begin{lemma}\label{lem:f-f1-f2-constants-relation}
Suppose \cref{assump:unknown-input-gain,assum:theta-thetadot-bounded,assump:B-lipschitz-bounded,assump:ft0-bound,assump:ftx-semiglobal_lipschitz,assump:ftx-rate-bounded} hold, and consider the functions ${\bar f}(t,x)$ and $f_i(t,x)$ ($i=1,2$)  that satisfy \cref{eq:f1-f2-f-relation}. Then, 
given an arbitrary $\delta>0$,  for any $x_i$ satisfying $\norm{x_i}\leq \delta$ ($i=1,2$) and for any $t_i\geq 0$ ($i=1,2$), we have
\begin{align}
\norm{{\bar f}(t,0)} & \leq b_f^0,\ \forall t\geq 0, \label{eq:fbar0-bound}\\
\left\| {f_i(t,0)} \right\| &\leq  {b_{f_i}^0},\ \forall t\geq 0  \label{eq:fit0-bound},  \\ 
\left\| {{\bar f}(t_1,x_1)\! - \!{\bar f}(t_2,x_2)} \right\| &\le { L_{{\bar f}}^\delta}\left\| {x_1 \!- \!x_2} \right\| + l_{\bar  f}\abs{t_1\!-\!t_2}, \label{eq:f-bar-lipschitz-cond} \\
\left\| {f_i(t_1,x_1) \!- \!f_i(t_2,x_2)} \right\| &\le {L_{f_i}^\delta}\left\| {x_1 \!- \!x_2} \right\| + l_{f_i}\abs{t_1\!-\!t_2}, \hspace{-2mm} \label{eq:fi-lipschitz-cond}
\end{align}
where $b_f^0$ is introduced in \cref{assump:ft0-bound}, and
\begin{subequations}
\begin{align}
  \hspace{-3mm}   b_{f_1}^0 \!& \trieq  b_{B^\dagger}b_f^0 
, b_{f_2}^0\trieq b_{B^\dagger_u} b_f^0, \label{eq:b10-b20-defn}\\
\hspace{-3mm} l_{\bar f} & \trieq {l_f} +{\max\limits_{\omega  \in \Omega}} \left\| {\omega  - {\mathbb{I}_m}} \right\|\left( {L_B  b_{K_x} + b_{B}{L_{K_x}}} \right)  \delta{b_{\dot \theta }}  , 
\label{eq:l-barf-defn}\\
\hspace{-3mm} L_{\bar f} & \trieq  {L_f^\delta }+ \mathop {\max }\limits_{\theta  \in \Theta } \left\| {B(\theta )K_x(\theta )} \right\|{\max\limits_{\omega  \in \Omega }} \left\| {\omega  - {\mathbb{I}_m}} \right\| \label{eq:L-barf-defn},  \\
\hspace{-3mm} l_{f_1} \!& \trieq  {l_f}b_{B^\dagger} \! +\! L_{B^\dag }b_{\dot \theta}({b_f^0} \!+ \!{L_f^\delta }\delta )\! + \! {L_{K_x}} b_{\dot \theta}\delta {\max_{\omega  \in \Omega }}\left\| {\omega \! -\! {\mathbb{I}_m}} \right\|, \hspace{-1mm} \label{eq:l1t-defn}\\
\hspace{-3mm} l_{f_2} \!&\trieq    {l_f}b_{\Bu^\dagger} + L_{\Bu^\dag }b_{\dot \theta}({b_f^0} + {L_f^\delta }\delta ), \label{eq:l2t-defn}\\
\hspace{-3mm} L_{f_1}^\delta \! & \trieq  {L_f^\delta }b_{B^\dagger}\!+  \!\! {\max\limits_{\omega  \in \Omega ,\:\theta  \in \Theta }}\left\| {(\omega \! -\! {\mathbb{I}_m})K_x(\theta )} \right\|, \label{eq:L1delta-defn}\\
\hspace{-3mm} L_{f_2}^\delta \!& \trieq  {L_f^\delta }b_{\Bu^\dagger}.\label{eq:L2delta-defn}
\end{align}
\end{subequations}

\end{lemma}}

\begin{definition}\label{defn:linf-signal-norm}
The $\mathcal{L}_\infty$ and truncated $\mathcal{L}_\infty$ norm of a function $x:\mathbb{R}^+ \rightarrow\mathbb{R}^n$ are defined as $\norm{x}_{\mathcal{L}_\infty}\triangleq \sup_{t\geq 0}\norm{x(t)}$, and $\linfnormtruc{x}{\tau}\triangleq \sup_{0\leq t\leq\tau}\norm{x(t)}$, respectively.
\end{definition}

\begin{definition}\label{defn:MIMO-L1} \cite[Section~III.F]{Scherer97Multi}
For a stable proper MIMO system $H(s)$ with input $\mru(t)\in \mbR^m$ and output $\mry(t)\in \mbR^p$, its \lonew norm is defined as
\begin{equation}\label{eq:l1-norm-defn}
\lonenorm{H(s)} \trieq \sup_{\mrx(0)=0,\linfnorm{\mru}\leq 1} {\linfnorm{\mry}}.
\end{equation}
\end{definition}
The \lonew norm of an LTI system in Definition~\ref{defn:MIMO-L1} is also known as the peak-to-peak gain (PPG) \cite[Section~III.F]{Scherer97Multi}. The following lemma is directly from Definition~\ref{defn:MIMO-L1}. 
\begin{lemma}\label{lem:L1-Linf-relation}
For a stable proper MIMO system $H(s)$ with input $\mru(t)\in \mbR^m$ and output $\mry(t)\in \mbR^p$, we have 
$\linfnormtruc{\mry}{\tau}\leq \lonenorm{H(s)}\linfnormtruc{\mru}{\tau}$, for any $\tau\geq 0$.
\end{lemma}
\begin{remark}
We use the vector 2-norm in definition of the $\linf$ function norm (Definition~\ref{defn:linf-signal-norm}) to further define the \lonew system norm (Definition~\ref{defn:MIMO-L1}), which is different from the common practice that uses the vector $\infty$-norm (see e.g., \cite[Section A.7.1]{naira2010l1book}). 
\end{remark}

Let $x_\textup{in}$ be the state of the system 
\begin{equation}\label{eq:xin-dynamics}
    \dot{x}_\textup{in}(t) = A_m(\theta)x_\textup{in}(t), \quad x_\textup{in}(0) = x_0.
\end{equation}
Then $\xin$ is bounded according to the following lemma.
\begin{lemma}\label{lem:xin_bound}
Given the autonomous system in \eqref{eq:xin-dynamics}, under Assumption~\ref{assump:desired-dynamics-stable-Lyapunov}, the following inequality holds:
\begin{equation}\label{eq:xin-linfnorm-rho-in-relation}
 \linfnorm{x_\textup{in}}\leq \rho_\textup{in}, 
\end{equation}
where 
\begin{equation}\label{eq:rhoin-defn}
   \rho_\textup{in}\triangleq\rho_0  \sqrt{\frac{\lammaxtheta{P}}{\lammintheta{P}}}
\end{equation}
with $\rho_0$ defined below \eqref{eq:plant-dynamics}.
\end{lemma}
\begin{proof}
Define $V(t,\theta) \triangleq x_\textup{in}^{\top\!}(t) P(\theta) x_\textup{in}(t) $,
where $P(\theta)$ is introduced in \eqref{eq:laypunov_stability_Am}. From the dynamics in \eqref{eq:xin-dynamics} and the inequality in \eqref{eq:laypunov_stability_Am}, we have 
\begin{align*}
    \dot{V}(t,\theta)& = \xin^{\top\!}(t)\left(A_m^{\top\!}(\theta) P(\theta) + P(\theta)A_m(\theta) + \dot{P}(\theta)\right)\xin(t)   \leq -\mu_P \xin^{\top\!}(t) P(\theta)\xin(t) 
    = -\mu_P V(t,\theta). 
\end{align*}
Therefore, 
$ V(t,\theta) \leq V(0,\theta(0))\te^{-\mu_P t}\leq \xin^{\top\!}(0)P(\theta(0))\xin(0) \te^{-\mu_P t}  \leq \lammaxtheta{P}\norm{x_0}^2 \te^{-\mu_P t}.  \label{eq:Vt-x0-bound}
$
The preceding inequality, together with the fact that
$\label{eq:Vt-xt-relation}
    V(t,\theta)\geq \lammintheta{P} \norm{\xin(t)}^2,
$
 leads to 
\[ 
\begin{split}
 \norm{\xin(t)} &\leq \norm{x_0}\sqrt{\frac{\lammaxtheta{P}}{\lammintheta{P}}}\te^{-\frac{\mu_P}{2} t} 
 \leq \rho_0\sqrt{\frac{\lammaxtheta{P}}{\lammintheta{P}}} ,\quad \forall t\geq 0,   
\end{split}
 \] 
which proves \eqref{eq:xin-linfnorm-rho-in-relation}. \qedclosed
\end{proof}

We next consider a general LPV system:
\begin{equation}\label{eq:lpv-system-general}
\mcG(\theta):
\left\{
\begin{aligned}
    \dot{\mrx}(t) &= \mrA(\theta(t))\mrx(t) + \mrB(\theta(t))\mru(t),\\
    \mry(t) & = \mrC(\theta(t))\mrx(t) + \mrD(\theta(t))\mru(t),
\end{aligned}\right.
\end{equation}
where $\mrx(t)\in\mbR^{n},\ \mru(t)\in\mbR^{m},\ \mry(t)\in\mbR^{p}$, and  $\mrA(\theta),\mrB(\theta),\mrC(\theta),\mrD(\theta)$ are state-space matrices of compatible dimensions depending on the varying parameter vector $\theta(t)\in\mbR^s$. Furthermore, the trajectory of $\theta$ is constrained by \eqref{eq:theta-thetadot-constraints}. We use $\mcG(\theta)$ to denote the input-output mapping of the system in \eqref{eq:lpv-system-general}.

\begin{definition} 
Given {a stable} LPV system in \eqref{eq:lpv-system-general} subject to Assumption~\ref{assum:theta-thetadot-bounded}, 
the PPG of $\mcG(\theta)$ is defined as 
{\begin{equation}\label{eq:ppg-defn}
\lonenorm{\mcG(\theta)} \trieq \sup_{\mrx(0)=0,\ \linfnorm{\mru}\leq 1}\sup_{\ \theta(t)\in \Theta,\ \dot \theta(t)\in \Theta_d} {\linfnorm{\mry}}.
\end{equation}}
Moreover, a positive scalar $\gamma$ is a bound on the PPG of $\mcG(\theta)$ if for $\mrx(0)=0$ and all trajectories of $\theta$ satisfying \eqref{eq:theta-thetadot-constraints}, we have
\begin{equation}\label{eq:ppg-bound-defn}
\linfnormtruc{\mry}{\tau} \leq \gamma \linfnormtruc{\mru}{\tau},\quad \forall \tau\geq 0.
\end{equation}
\end{definition}

%
The PPG is equivalent to the $\lone$ norm for an LTI system, and thus we still use $\lonenorm{\cdot}$ to denote the PPG of an LPV system. In practice, it is very challenging to compute the PPG of an LPV system exactly. However, a PPG bound can be computed using LMI techniques (see Lemma~\ref{lem:ppg-bound-computaion-lmi}). Hereafter, we use $\norm{\mcG(\theta)}_\lonebar$ to denote an arbitrary PPG bound of $\mcG(\theta)$.
Next lemma follows straightforwardly from the inequalities in \eqref{eq:ppg-defn} and \eqref{eq:ppg-bound-defn}, and thus the proof is omitted.
\begin{lemma} \label{lem:input-output-ppg-bound-relation}
Given {a stable} LPV system in \eqref{eq:lpv-system-general} with  $\mrx(0)=0$, we have 
\begin{equation} \label{eq:input-output-ppg-bound-relation}
  \linfnormtruc{\mry}{\tau} \leq \norm{\mcG(\theta)}_\lonebar  \linfnormtruc{\mru}{\tau}, \quad \forall \tau\geq 0.  
\end{equation}
\end{lemma}
The LMI conditions for computing a PPG bound for an LTI system are given in \cite[Section~III.F]{Scherer97Multi}. The following lemma is a straightforward extension of this result to LPV systems using the PD Lyapunov function $V(t,\theta)= \mrx^{\top\!}(t) \mrP(\theta)\mrx(t)$. Therefore, the detailed proof is omitted for brevity.  
\begin{lemma}\label{lem:ppg-bound-computaion-lmi}
Given the LPV system in \eqref{eq:lpv-system-general} {subject to \cref{assum:theta-thetadot-bounded}},  it is bounded-input-bounded-output (BIBO) stable, and $\gamma>0$ is a PPG bound of this system, if there exist a PD and symmetric matrix $\mrP(\theta)$ and positive constants $\mu$ and $\upsilon$, such that the following inequalities hold:
\begin{equation}\label{eq:ppg-bound-computation-lmi1}
    \begin{bmatrix}[1.2]
        \lrangle{\mrA^{\top\!}(\theta)\mrP(\theta)}+ \mu\mrP(\theta) + \dot{\mrP}(\theta) & \star \\
        \mrB^{\top\!}(\theta)\mrP(\theta) &-\upsilon \mbI_m
    \end{bmatrix}
    < 0,  \quad  \forall (\theta,\dot{\theta})\in\Theta\times\Theta_d,   
  \end{equation}
  \begin{equation}\label{eq:ppg-bound-computation-lmi2}
      \begin{bmatrix}[1]
        \mu\mrP(\theta) & \star & \star\\
        0 & (\gamma-\upsilon)\mbI_m & \star \\
        \mrC(\theta) & \mrD(\theta) & \gamma \mbI_p
    \end{bmatrix}
    >0, \quad \forall \theta\in\Theta. 
  \end{equation}
\end{lemma}

\begin{remark}\label{rem:linesearch-mu}
Due to the term $\mu \mrP(\theta)$, the inequalities \eqref{eq:ppg-bound-computation-lmi1} and \eqref{eq:ppg-bound-computation-lmi2} are only LMIs if $\mu$ is fixed. Thus, finding the best bound requires performing a line search over $\mu>0$. 
{The PD-LMIs \cref{eq:ppg-bound-computation-lmi1,eq:ppg-bound-computation-lmi2}, despite involving the derivatives of decision matrices, are standard in the LPV literature \cite{Apk98}, and can be readily solved by available techniques (see Section~\ref{sec:sub-computation-perspective}}). 
\vspace{0mm} 
\end{remark}

The following lemma characterizes the state evolution of a BIBO stable LPV system, given a bounded input. 
\begin{lemma}\label{lem:lpv-time-domain-response}
Consider an LPV system in \eqref{eq:lpv-system-general} with $\mrx(0)=0$, {subject to \cref{assum:theta-thetadot-bounded}}. If there exists a  PD symmetric matrix $\mrP(\theta)$ and a constant $\mu>0$ such that \begin{equation}\label{eq:appendix-Lyapunov-stability-cond}
   \left\langle\mrA^{\top\!}(\theta)\mrP(\theta)\right\rangle+ \dot{\mrP}(\theta) \leq -\mu \mrP (\theta),\  \forall (\theta,\dot{\theta})\in\Theta\times\Theta_d, 
\end{equation}
then 
\begin{equation}\label{eq:x-bound-u-sharp-exp}
\norm{\mrx(t)}\leq \beta_\Theta(t,\mu,\mrP(\theta),\mrB(\theta))\linfnormtruc{\mru}{t}, \quad \forall t\geq 0,
\end{equation}
where
{\begin{align}
\!\!\!\beta_\Theta&(t,\mu,\mrP,\mrB) \!\trieq \! \sqrt{\!\sqrt{\frac{\lambda_2}{\lambda_1}}(1\!-\!e^{-\mu t})}\frac{2}{\mu \lambda_1}\!\max_{\theta\in\Theta}\!\norm{\mrP(\theta)\mrB(\theta)}\!,\label{eq:beta-defn-in-LPV-input-output-response} \\
   \!\!\! \lambda_1 & = \lammintheta{\mrP}, \  \lambda_2= \lammaxtheta{\mrP}. \label{eq:lambda12-defn-in-LPV-input-output-response}  
\end{align}}
\vspace{-5mm}\end{lemma}
\begin{proof}
Define
\begin{equation}\label{eq:Vttheta-defn}
  V(t,\theta)=\mrx^{\top\!}(t)\mrP(\theta)\mrx(t).  
\end{equation}
Then, 
\begin{align}
    \dot V(t,\theta) &= \mrx^{\top\!}\left( \langle\mrA^{\top\!}(\theta)\mrP(\theta) \rangle+  \dot{\mrP}(\theta)\right)\mrx + 2\mrx^{\top\!}\mrP(\theta)\mrB(\theta)u \leq -\mu\mrx^{\top\!}\mrP(\theta)\mrx + 2\mrx^{\top\!}\mrP(\theta)\mrB(\theta)\mru \label{eq:Vdot-V-u}
    \\
    & \leq -\mu \lambda_1 \norm{\mrx(t)}^2 + 2 \norm{\mrx(t)}\norm{\mrP(\theta)\mrB(\theta)}\norm{\mru(t)}\leq 0,\quad \forall \norm{\mrx(t)}\geq \frac{2}{\mu\lambda_1}\max_{\theta\in\Theta}\norm{\mrP(\theta)\mrB(\theta)}\linfnormtruc{\mru}{t}.  \nonumber
\end{align}
The preceding inequality \red{and the fact that $V(t,\theta) \leq \lambda_2\norm{\mrx(t)}^2$ indicate} \begin{equation}
    V(t,\theta)\leq \lambda_2 \left(\frac{2}{\mu\lambda_1}\max_{\theta\in\Theta}\norm{\mrP(\theta)\mrB(\theta)}\linfnormtruc{\mru}{t}\right)^2,
\end{equation}
which, together with $ V(t,\theta)\geq \lambda_1 \norm{\mrx(t)}^2$, implies that
\begin{equation}\label{eq:x-bound-u-conservative}
    \linfnormtruc{\mrx}{t}\leq \sqrt{\frac{\lambda_2}{\lambda_1}}\frac{2}{\mu\lambda_1}\max_{\theta\in\Theta}\norm{\mrP(\theta)\mrB(\theta)}\linfnormtruc{\mru}{t}.
\end{equation}
Next, we show that the preceding bound on $\mrx$  can be further improved. Equations \red{\cref{eq:Vttheta-defn,eq:Vdot-V-u}}  indicate that
\begin{flalign*}
    \dot V(t,\theta)  \leq -\mu V(t,\theta) +  2\norm{\mrx(t)}\norm{\mrP(\theta)\mrB(\theta)}\norm{\mru(t)},
\end{flalign*}which, together with $V(0,\theta)=0$, leads to
{\begin{flalign}
    V(t,\theta)\leq &e^{-\mu t} V(0,\theta)+  \int_0^{t} e^{-\mu (t-\tau)} 2\norm{\mrx(\tau)}\norm{\mrP(\theta(\tau))\mrB(\theta(\tau))}\norm{\mru(\tau)}d\tau \nonumber \\
    \leq &2  \linfnormtruc{\mrx}{t}\max_{\theta\in\Theta}\norm{\mrP(\theta)\mrB(\theta)}\linfnormtruc{\mru}{t} \int_0^{t} e^{-\mu (t-\tau)} d\tau \nonumber  \\
    \leq  &2\linfnormtruc{\mrx}{t}\max_{\theta\in\Theta}\norm{\mrP(\theta)\mrB(\theta)}\linfnormtruc{\mru}{t}  \frac{1}{\mu} (1-e^{-\mu t}). \label{eq:Vtheta-mu-relation}
\end{flalign}}
By plugging \eqref{eq:x-bound-u-conservative} into \red{\eqref{eq:Vtheta-mu-relation}} and further considering $V(t,\theta)\geq \lambda_1 \norm{\mrx(t)}^2$, we arrive at \eqref{eq:x-bound-u-sharp-exp}. \qedclosed
\end{proof}

The following lemma refines the result in \cref{lem:lpv-time-domain-response} and establishes {\it individual}  bounds for the states of \eqref{eq:lpv-system-general} before and after an arbitrary time instant $T$. 
\begin{lemma}\label{lem:lpv-time-domain-response-split@T}
Given an LPV system in \eqref{eq:lpv-system-general} with $\mrx(0)=0$, {subject to \cref{assum:theta-thetadot-bounded}}, assuming that there exists a  PD and symmetric matrix $\mrP(\theta)$ and a constant $\mu>0$ such that 
\eqref{eq:appendix-Lyapunov-stability-cond} holds, we have
\begin{equation}\label{eq:x-bound-u-sharp-exp-0-T-infty}
\hspace{-1.2mm}
\norm{\mrx(t)}\leq \!
\left\{
  \!\!\!  \begin{array}{lr}
   \beta_\Theta(T,\mu,\mrP,\mrB)\linfnormtrucinterval{\mru}{0,T},\  \forall 0\leq t\leq T, \\
      \!\!\! \begin{array}{l}
 \sqrt{\frac{\lambda_2}{\lambda_1}}\beta_\Theta(T,\mu,\mrP,\mrB) \linfnormtrucinterval{\mru}{0,T} 
 \, \negmedspace {} + \beta_\Theta(\infty,\mu,\mrP,\mrB)\linfnormtrucinterval{\mru}{T,t},\ \forall t>T,
    \end{array} 
    \end{array}
    \right.
\end{equation}
where $\lambda_i$ ($i=1,2$) are defined in \eqref{eq:lambda12-defn-in-LPV-input-output-response}.
\vspace{-1mm}\end{lemma}
\begin{proof}
The first part follows directly from Lemma~\ref{lem:lpv-time-domain-response}, along with the fact that for $0\leq t\leq T$, $\beta_\Theta(t,\mu,\mrP,\mrB)$ $\leq \beta_\Theta(T,\mu,\mrP,\mrB)$. The second part results from application of Lemmas~\ref{lem:xin_bound} and \ref{lem:lpv-time-domain-response} considering the state evolution starting from $T$ with the initial condition $x(T)$ , as well as the fact of $\beta_\Theta(t-T,\mu,\mrP,\mrB)\leq \beta_\Theta(\infty,\mu,\mrP,\mrB)$ for any $t>T$. \qedclosed
\end{proof}

\section{Robust Adaptive Control of LPV Systems With Unmatched Uncertainties}\label{sec:l1_architecture}
{A unique feature of an $\mathcal{L}_1$AC architecture is a  low-pass filter that decouples the estimation loop from the control loop, thereby allowing for arbitrarily fast adaptation without sacrificing the robustness \red{\cite[Section~1.3]{naira2010l1book}}. In the presence of an unknown input gain as considered in this paper, the filter is constructed through a feedback gain matrix $K\in \mbR^{m \times m}$ and an $m\times m$ strictly proper transfer function matrix $D(s)$, which lead, for all $\omega\in \Omega$, to a strictly proper stable low-pass filter
$ \mcC(s) \triangleq \omega KD(s)(\mbI_m + \omega KD(s))^{-1},$
with DC gain $C(0)= \mathbb{I}_m$. The gain  $K$ and the filter $D(s)$ will be used in the control law in \eqref{eq:l1-control-law}}. See also \cref{rem:filter-in-control-law} for more insights.  For simplifying the derivations, we set $D(s) =\frac{1}{s}\mbI_m$, which leads to 
\begin{equation}\label{eq:filter-defn}
\mcC(s) = \omega K(s\mbI_m+\omega K)^{-1}.
\end{equation}
{Hereafter, we use the notation $\mcC$ (i.e., without argument) to denote the input-output mapping corresponding to the transfer function matrix $\mcC(s)$.}  We now introduce a few notations to  be used later in the controller design and performance analysis. Let $\mcH_{xm}(\theta)$ and $\mcH_m(\theta)$ respectively denote the input-to-state and input-to-output mappings{\footnote{{We use mappings instead of transfer functions since time-varying systems are considered.}}}of the system 
\begin{equation}\label{eq:lpv-system-no-nonlinearity}
   \dot{x}(t) = A_m(\theta)x(t)+B(\theta)u(t), \  
   y(t) = C_m(\theta)x(t),
\end{equation}
and 
$\mcH_{xum}(\theta)$ and $\mcH_{um}(\theta)$ respectively denote the input-to-state and input-to-output mappings of the system
\begin{equation}\label{eq:lpv-system-no-nonlinearity-Bum}
\dot{x}(t) = A_m(\theta)x(t)+\Bu(\theta)u(t),\ y(t) = C_m(\theta)x(t). 
\end{equation}
We further define
\begin{subequations}\label{eq:Gxm-Gxum-defn}
\begin{align}
    \mcG_{xm}(\theta)&\triangleq\mcH_{xm}(\theta)(\mbI_m-\mcC),\label{eq:Gxm-defn}\\
      \mcG_{xum}(\theta) &\triangleq
      \mcH_{xum}(\theta)-\mcH_{xm}(\theta)\mcC \barH(\theta).\label{eq:Gxum-defn}
\end{align}
\end{subequations}
In \eqref{eq:Gxum-defn}, $\barH(\theta)$ is a $\theta$-dependent mapping designed to mitigate the effect of the unmatched uncertainty, detailed in \cref{sec:sub-unmatched-dist-rejection}. 
For presentation of \cref{sec:sub-unmatched-dist-rejection}, we also define 
{
\begin{equation}\label{eq:Gm-Gum-defn}
\mcG_{m}(\theta) \triangleq C_m(\theta)\mcG_{xm}(\theta),\    \mcG_{um}(\theta) \triangleq C_m(\theta)\mcG_{xum}(\theta),
\end{equation}}where $C_m(\theta)$ is the output matrix in \eqref{eq:plant-dynamics}.
\subsection{Proposed control architecture}
The dynamics in \cref{eq:plant-dynamics-f12-a-fhat} is equivalent to
{\begin{equation}\label{eq:plant-dynamics-sigma-txu}
\dot{x}(t) =  A_m(\theta)x(t)+ B(\theta) u(t) + \sigma(t,x(t),u(t)), \hspace{-1.6mm}
\end{equation}
where  
\begin{equation}\label{eq:sigma-txu-defn}
\sigma(t,x(t),u(t)) \triangleq B(\theta)(\omega-\mbI_m)u(t) +{\bar f}(t,x(t)),
\end{equation}
with ${\bar f}(t,x(t))$ defined in \cref{eq:f1-f2-f-relation}}.
We now describe the proposed robust adaptive control architecture (see Fig.~\ref{fig:control-architecture}) consisting of a state predictor, an adaptive law and a low-pass filtered control law. 
\begin{figure}[htb]
\vspace{-3mm}
    \centering
    \includegraphics[width=0.4\columnwidth]{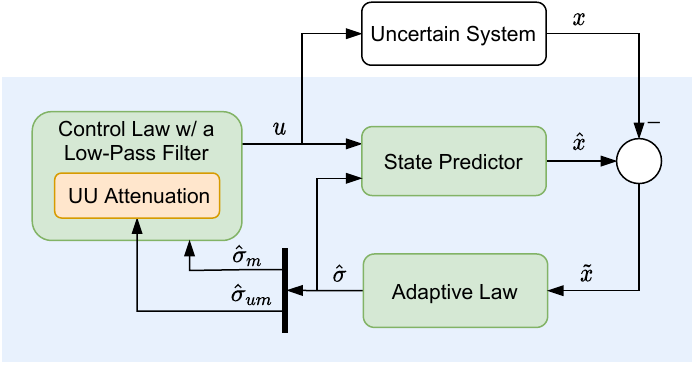}
        \vspace{-2mm}
    \caption{Proposed robust adaptive control architecture}
    \label{fig:control-architecture}
    \vspace{-3mm}
\end{figure}

{ \bf State Predictor:} The state predictor is selected as:
{
\begin{equation}\label{eq:state-predictor}
 \begin{split}
\dot{\hat{x}}(t)&=  A_m(\theta)x(t)+  B(\theta) u(t) +\hsigma(t) -a\tilde{x}(t),
  \hspace{-3mm}  
\end{split} 
\end{equation}}where $\hat{x}(0) = {x}_0$, $\hat{x}(t)\in\mbR^n$ is the predicted state vector,  
$\tilde{x}(t) \triangleq \hat{x}(t)-x(t)$ is the prediction error, 
$a$ is a positive scalar, and $\hsigma(t)$ is the adaptive estimation of the {\it lumped} uncertainty, $\sigma(t,x,u)$, defined in \cref{eq:sigma-txu-defn}.


{\bf Adaptive Law:} The estimation $\hsigma(t)$ is updated by the following piecewise-constant estimation law (similar to that in \cite[Section~3.3]{naira2010l1book}):
\begin{equation}\label{eq:adaptive_law}
\begin{split}
   \hsigma(t) 
    &=  
    \hsigma(iT)
    , \quad t\in [iT, (i+1)T), \\
 \hsigma(iT) 
&= -\Upsilon(T)\tilde{x}(iT),
\end{split}
\end{equation}
for $i=0,1,\cdots$, where $T$ is the estimation sampling time, and $\Upsilon(T) \trieq a/({e^{a T}-1})$.
We further let 
\begin{equation}\label{eq:hsigma_m_um_hsigma_relation}
    \begin{bmatrix}
    \hsigma_m(t) \\
    \hsigma_{um}(t) 
\end{bmatrix} = 
\begin{bmatrix}[1]
    \mbI_m & 0 \\
    0 & \mbI_{n-m}
\end{bmatrix}
\bar B^{-1}(\theta(t))\hsigma(t),
\end{equation}
Note that $\hsigma_m(t)$ and $\hsigma_{um}(t)$ represent the \textit{matched} and \textit{unmatched} components of $\hsigma(t)$, respectively.
\begin{remark}
The parameter $a$ in \eqref{eq:state-predictor} must be positive to ensure the stability of the prediction error dynamics (defined later in \eqref{eq:prediction-error}) and accurate estimation of the uncertainty (defined in \eqref{eq:sigma-txu-defn}). Based the estimation error bound defined later in \eqref{eq:tilsigma-bound}, the smaller $a$ is, the 
more accurate the uncertainty estimation is. On the other hand, setting $a$ too small, together with a small sampling time ($T$ near 0), will make the estimation gain $\Upsilon(T)$ in \eqref{eq:adaptive_law} approach $\infty$, potentially causing numerical issues.
\vspace{0mm} \end{remark}
{\bf Control Law:}
The control signal $u(t)$ is generated as the output of the following dynamic system: 
{\begin{equation}\label{eq:l1-control-law}
    \dot u(t) = -K\left(u(t)+\hsigma_m(t) +\hat{\eta}_2(t)- K_r(\theta)r(t)\right),
\end{equation}}where 
$\hat{\eta}_2 \triangleq \barH(\theta)\hsigma_{um}$. Here, $\barH(\theta)$ denotes a dynamic feedforward mapping for attenuating the effect of UUs, which will be detailed in \cref{sec:sub-unmatched-dist-rejection}.
{\begin{remark}\label{rem:filter-in-control-law}
Equation \cref{eq:l1-control-law} 
can be rewritten as $u=-\omega^{-1}\mcC\left( (\mbI_m-\omega)u+\hat\sigma_m +\hat\eta_2 - K_r(\theta)r\right)$, which clearly shows the  filter $\mcC(s)$ embedded in the control law. \red{As shown later in \cref{lem:ref-id-bound} and \cref{rem:ref-id-bound}, a higher bandwidth of $\mcC(s)$ is desired for satisfying the stability condition \cref{eq:l1-stability-condition} and achieving better performance (i.e., smaller error bound) but may lead to reduced robustness against time delay, as illustrated in \cite[Sections 1.3, and 2.6]{naira2010l1book} for the case of a simple LTI nominal system.} 
\end{remark}}

\begin{remark}\label{rem:not-filer-r}
To improve tracking performance when $K$ is small (indicating the filter in \eqref{eq:filter-defn} has a low (nominal) bandwidth), one can choose to use a different filter of a higher bandwidth for the feedforward signal, $K_r(\theta)r(t)$, or not to filter it. In this case, the control law in \eqref{eq:l1-control-law} should be adapted to $u = u_{ad} + \mcC_r \left(K_r(\theta)r)\right)$, where $\mcC_r$ represents a filter (with $\mcC_r=\mbI_m$ corresponding to the no-filtering case) and $u_{ad}(t)$ is generated by {$\dot u_{ad}(t) = -K\left( u_{ad}(t)+\hsigma_m(t) +\hat{\eta}_2(t)\right)$}
and $\hsigma_m$ and $\hat \eta_2$ having the same meaning as in \eqref{eq:l1-control-law}. The analysis in \cref{sec:analysis_l1} based on the control law in \eqref{eq:l1-control-law} can be straightforwardly adjusted accordingly, and is omitted.
\vspace{-3mm} \end{remark}
\subsection{Unmatched uncertainty attenuation via PPG minimization}\label{sec:sub-unmatched-dist-rejection}
\begin{figure}[h]
\vspace{-4mm}
    \centering
    \includegraphics[width=.6\textwidth]{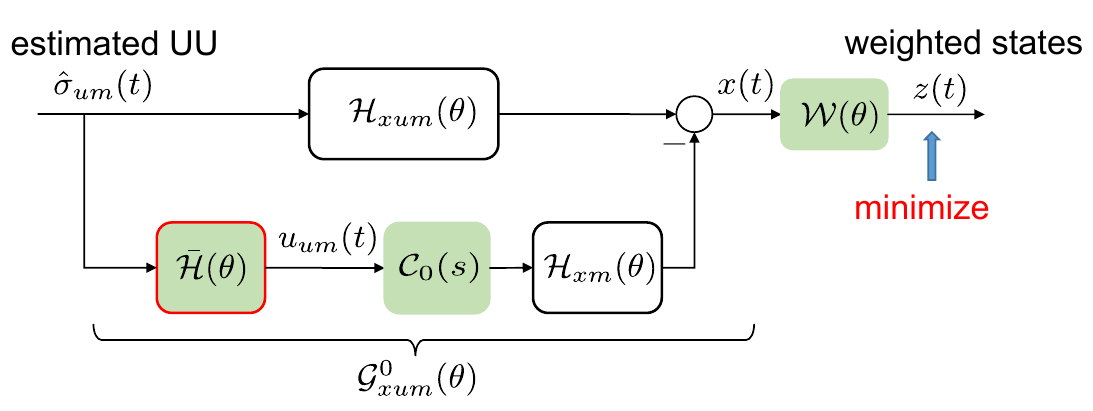}
    \vspace{-2mm}
    \caption{Scheme for unmatched uncertainty attenuation}
    \label{fig:unmatched_compensation}
    \vspace{-2mm}
\end{figure}
Existing \loneAC~architectures often rely on the direct inversion of the nominal LTI systems to cancel the effect of unmatched uncertainties (UUs) on the outputs under the assumption that the nominal systems are square and minimum phase \cite{xargay2010unmatched}.
The idea cannot be applied here since it is very challenging, if not impossible, to compute the inverse of an LPV system, or even define it. Here, we propose a new method for unmatched uncertainty attenuation (UUA) based on PPG minimization. The main idea is as follows. With the UU estimation, $\hsigma_{um}(t)$, we design a feedforward LPV mapping $\barH(\theta)$ to produce a control input $\uum(t)$ that tries to minimize the PPG gain from $\hsigmaum(t)$ to the weighted output $z(t)\in \mbR^n$. The idea is shown in Fig.~\ref{fig:unmatched_compensation}, where the weighting function $\mcW(\theta)$ can be used to tune the frequency range where we want to suppress the effect of the UUs, and emphasize different output channels. $\mcC_0(s)$ is a low-pass filter defined by $\mcC_0(s) \triangleq \omega_0 K(s\mbI_m+\omega_0 K)^{-1},$
which is obtained by plugging a nominal input gain $\omega_0$ (selected to be $\mbI_m$ in \eqref{eq:state-predictor}) into \eqref{eq:filter-defn}.   Letting $\mcC_0$ denote the input-output mapping of $\mcC_0(s)$,  with the architecture in \cref{fig:unmatched_compensation}, the mapping from $\hsigmaum(t)$ to $x(t)$ can be represented by $\mcG_{xum}^0(\theta) \triangleq
      \mcH_{xum}(\theta)-\mcH_{xm}(\theta)\mcC_0\barH(\theta)$, which will be exactly $\mcG_{xum}(\theta)$ in the stability condition \eqref{eq:l1-stability-condition} if $\omega_0 = \omega$. 
\begin{remark}\label{rem:barH-tradeoff}
If we set $\mathcal{W}(\theta)=\textup{diag}(C_m(\theta),0)$, then $z=[y,0]^{\top\!}$, which means that $\barH(\theta)$ is for attenuating the effect of the UUs on the outputs. This is desired for improving the tracking performance. 
On the other hand, if we set $\mathcal{W}(\theta) = \mbI_n$, then we have $z=x,$ which indicates that $\barH(\theta)$ is for mitigating the effect of the UUs on all the states. By setting  $\mathcal{W}(\theta) = \mbI_n$, we can minimize $\lonenormbar{\mcG_{xum}^0(\theta)}$ and potentially $\lonenormbar{\mcG_{xum}(\theta)}$ as well, which allows for the stability condition \eqref{eq:l1-stability-condition} to hold under more severe uncertainties. Thus, $\mathcal{W}(\theta)$ can  be tuned to balance the tracking performance and the robustness margin for satisfying the stability condition \eqref{eq:l1-stability-condition}. 
\vspace{0mm} \end{remark}
\begin{remark}
We design $\barH(\theta)$ to minimize the PPG as the performance index from the UUs to the weighted states, because we aim to provide a transient performance guarantee. If the transient performance is not a major concern, then other performance indices can be utilized, such as $\mathcal{L}_2$ gain.
\vspace{0mm} \end{remark}

\begin{remark}
The proposed approach for UUA can be easily adapted for LTI systems. It does not need strict-feedback forms required in backstepping based approaches \cite{yao1997adaptive,koshkouei2004backstepping,sun2015global-unmatched,li2014continuous}, and relaxes square system and non-minimum phase assumptions needed in dynamic inversion based approaches \cite{xargay2010unmatched}. Combined with the small-gain theorem (used to derive \eqref{eq:l1-stability-condition}), our approach deals with state- and time-dependent uncertainties, unlike the $H_\infty$ \cite{wei2010composite-unmatched} and coordinate transformation \cite{che2012eigenvalue} approaches that  handle only bounded disturbances. 
\vspace{-1mm} \end{remark}
Next, we will show how to design the mapping $\barH(\theta)$ leveraging LMI optimization-based control design \cite{Scherer97Multi,Apk98,sato2008inverse}. Let us denote the state-space realization of $\mcW(\theta)$ as $\left(A_\mcW(\theta), B_\mcW(\theta),C_\mcW(\theta),D_\mcW(\theta)\right)$, and the order of $\mcW(\theta)$, i.e., the number of state variables, as $n_\mcW$. The generalized plant  for designing $\barH(\theta)$ is then given by 
{\begin{flalign}
\dot{\bar{x}}(t) &= \bar A(\theta )\bar x(t) + {{\bar B}_1}(\theta ){\hsigmaum}(t) + {{\bar B}_2}(\theta ){\uum}(t),\nonumber\\
z(t) &= {{\bar C}_1}(\theta )\bar x(t) + {{\bar D}_{11}}(\theta ){\hsigmaum}(t) + {{\bar D}_{12}}(\theta ){\uum}(t),  \label{eq:lpv-generalized-plant} \hspace{-2mm} \\
\bar y(t) &= {{\bar C}_2}(\theta )\bar x(t) + {{\bar D}_{21}}(\theta ){\hsigmaum}(t) \nonumber,
\vspace{-3mm}
\end{flalign}}
where
$\bar x \trieq \left[\begin{array}{*{20}{c}}
{{x_{um}^{\top\!}}} &
{{x_m^{\top\!}}} & 
{{x_\mcW^{\top\!}}}
\end{array}\right]^{\top\!}$,
{\setlength{\arraycolsep}{1.4pt}
$$ \left[{\begin{array}{*{20}{c}}
{\bar A }&\vline& {{{\bar B}_1} }&\vline& {{{\bar B}_2} }\\
\hline
{{{\bar C}_1} }&\vline& {{{\bar D}_{11}} }&\vline& {{{\bar D}_{12}} }\\
\hline
{{{\bar C}_2} }&\vline& {{{\bar D}_{21}} }&\vline& 
\end{array}} \right]\trieq 
{\renewcommand*{\arraystretch}{1.1}
\left[ {\begin{array}{*{20}{c}}
{{A_m} }&0&0&\vline& {{\Bu} }&\vline& 0\\
0&{{A_m} }&0&\vline& 0&\vline& {{B} }\\
{{B_\mcW }{C} }&{{B_\mcW }{C} }&{{A_\mcW }}&\vline& 0&\vline& 0\\
\hline
{{D_\mcW }{C} }&{{D_\mcW }{C} }&{{C_\mcW }}&\vline& 0&\vline& 0\\
\hline
0&0&0&\vline& {{\mbI_{n - m}}}&\vline& 
\end{array}}  \right]},$$}\vspace{-1mm} \\
with the dependence of the matrices on $\theta$ omitted. Here, $x_{um}(t)
\in \mbR^n$, $x_{m}(t)
\in \mbR^n$ and $x_{W}(t)
\in \mbR^{n_\mcW}$ are the state variables associated with $\mcH_{um}(\theta)$, $\mcH_{m}(\theta)$ and  $\mcW(\theta)$, respectively. Obviously, the dimension of the generalized plant in \eqref{eq:lpv-generalized-plant} is $\bar n = 2n+n_\mcW$. 
We consider $\barH(\theta)$ of the same order as the generalized plant \eqref{eq:lpv-generalized-plant} with the state-space realization:
\begin{equation}\label{eq:barH-ss-form}
\left\{
    \begin{aligned}
        \dot x_\barH(t) & = A_\barH(\theta) x_\barH(t) + B_\barH(\theta) \bar y(t)\\
        \uum(t) & = C_\barH(\theta) x_\barH(t) + D_\barH(\theta) \bar y(t)
    \end{aligned}\right.,
\end{equation}
where $x_\barH(t)\in\mbR^{\bar n}$.
The following lemma presents synthesis conditions for designing $\barH(\theta)$ in order to minimize the PPG from $\hsigmaum(t)$ to $z(t)$. 

\begin{lemma}\label{lem:p2p-synthesis-final}
Consider the LPV generalized plant in \eqref{eq:lpv-generalized-plant} {subject to \cref{assum:theta-thetadot-bounded}}. There exists an LPV controller \eqref{eq:barH-ss-form} 
enforcing internal stability and a bound $\gamma$ on the PPG of the closed-loop system, whenever there exist PD symmetric matrices $X(\theta)$ and $Y(\theta)$, a PD quadruple of state-space data ($\hat{A}(\theta), \hat{B}(\theta), \hat C(\theta), \hat D(\theta)$) and positive constants $\mu$ and $\upsilon$ such that the conditions \eqref{eq:p2p-lim-final-1} and \eqref{eq:p2p-lim-final-2} hold. In such case, the controller in \eqref{eq:barH-ss-form} is given by
\begin{equation}\label{eq:lpv-constroller-construction-final}
    \left\{ \begin{array}{ll}
{A_{\bar {\mathcal H}}}(\theta ) = \hat A{(X - Y)^{ - 1}}, &{B_{\bar {\mathcal H}}}(\theta ) = \hat B,\\
{C_{\bar {\mathcal H}}}(\theta ) = \hat C{(X - Y)^{ - 1}}, &
{D_{\bar {\mathcal H}}}(\theta ) = \hat D.
\end{array} \right.
\end{equation}
\vspace{-3mm}
{
\setlength{\arraycolsep}{3.0pt}
\begin{flalign}
\left[ {\begin{array}{*{20}{c}}
M_{11}& \star & \star \\
M_{21} &{ - \dot Y +  \left\langle \bar AY \right \rangle + \mu Y}& \star \\
M_{31} &{{\bar B}_1^ \top \!\! +\! \bar D_{21}^{\top\!}({{\hat B}^ \top }\! \!+ \!\hat D{\bar B}_2^ \top )}&{ - \upsilon {\mbI_{n - m}}}
\end{array}} \right] & < 0, \label{eq:p2p-lim-final-1} 
\end{flalign}
 \vspace{-1mm}
\begin{flalign}
\left[ {\begin{array}{*{20}{c}}
{\mu X}& \star & \star & \star \\
{\mu Y}&{\mu Y}& \star & \star \\
0&0&{(\gamma  - \upsilon ){\mbI_{n - m}}}& \star \\
{{{\bar C}_1}X \!-\! {{\bar D}_{12}}\hat C}&{{{\bar C}_1}Y}&{{{\bar D}_{11}}\! + \!{{\bar D}_{12}}\hat D{{\bar D}_{21}}}&{\gamma {\mbI_n}}
\end{array}} \right] &> 0, \label{eq:p2p-lim-final-2} 
\end{flalign}}where $M_{11} \trieq { - \dot X + \left\langle {\bar AX - {{\bar B}_2}\hat C} \right\rangle  + \mu X}$,  $M_{21}  \trieq  - \hat A - {{\bar B}_2}\hat C + \bar AX - \dot Y + {{\bar A}^ \top } + \mu {\mbI_{\bar n}}$, $M_{31} \trieq  {{{({{\bar B}_1} + {{\bar B}_2}\hat D{{\bar D}_{21}})}^ \top }}$.
\end{lemma}
\begin{proof}
The proof is built upon a preliminary lemma, i.e., \cref{lem:p2p-synthesis}, given in Appendix~\ref{sec:p2p-synthesis}. 
The equations \cref{eq:lpv-constroller-construction-final,eq:p2p-lim-final-1,eq:p2p-lim-final-2} can be derived from \cref{eq:lpv-controller-construction,eq:p2p-lim-1,eq:p2p-lim-2} in \cref{lem:p2p-synthesis} in two steps. The first step is to do a congruence transformation, which includes (i) multiplying the left-hand side of \eqref{eq:p2p-lim-1} by  $\rm{diag}(\mbI_{\bar{n}},\tilde Y^{-1}(\theta), \mbI_{n-m})$ and its transpose from the left and right, and (ii)  multiplying the left-hand side of \eqref{eq:p2p-lim-2} by  $\rm{diag}(\mbI_{\bar{n}},\tilde Y^{-1}(\theta), \mbI_{n-m},\mbI_n)$ and its transpose from the left and right. The second step is to introduce new variables $X(\theta)$, $Y(\theta)$ such that $X(\theta)=\tilde X(\theta)$, $Y(\theta)=\tilde Y^{-1}(\theta)$. Since both of these two steps induce equivalent transformations, the result directly follows from  Lemma~\ref{lem:p2p-synthesis} in Appendix~\ref{sec:p2p-synthesis}. \qedclosed
\end{proof}

\begin{remark}
Similar to the computation of a PPG bound of a given LPV system according to \cref{lem:ppg-bound-computaion-lmi}, a line search for $\mu$ is needed to obtain the tightest bound $\gamma$. Also, since $\gamma$ appears linearly in the synthesis conditions in Lemma~\ref{lem:p2p-synthesis-final}, we can design the mapping $\barH(\theta)$ to minimize $\gamma$. 
\vspace{0mm} \end{remark}

For performance analysis in \cref{sec:analysis_l1} {(more specifically in Lemma~\ref{lem:F-theta-tilsigma-bound})}, we define
\begin{equation}\label{eq:Ftheta-defn}
    \mcF(\theta ) \trieq  {\mcC}\left( {B^\dag (\theta ) + \bar {\mathcal H} (\theta )\Bu^\dag (\theta )} \right),
\end{equation}
where $B^\dag (\theta )$ and $\Bu^\dag (\theta )$ are the pseudo-inverse of $B (\theta )$ and $\Bu(\theta )$, respectively.  
{Then, under zero initial states, given the signal $\tilsigma(t)\in \mbR^{n}$ (defined later in \eqref{eq:tilsigma-defn}) to $\mcF(\theta )$ as the input,   the output $y_\mcF(t)\in \mbR^{m}$ can be computed as} 
\begin{equation}\label{eq:F-theta-ss-form}
\begin{split}
  {{\dot x_\mcF}}(t) &= {A_\mcF}(\theta ){x_\mcF}(t) + {B_\mcF}(\theta ){\tilsigma}(t),\quad x_\mcF(0)=0,\\
y_\mcF(t) &= {C_\mcF}(\theta ){x}_\mcF(t),
\end{split}
\end{equation}
where $A_\mcF(\theta)$, $B_\mcF(\theta)$ and $C_\mcF(\theta)$ are the state-space matrices. 
Since both $\barH(\theta)$ and $\mcC$ are stable, $\mcF(\theta)$ is also stable. In order to establish the bound on $y_\mcF(t)$ given a bounded $\tilsigma(t)$ to $\mcF(\theta)$ later in Lemma~\ref{lem:F-theta-tilsigma-bound}, we make the following assumption. 
\begin{assumption}\label{assump:F-theta-stability}
There exists a PD symmetric matrix $P_\mcF(\theta)$ and a constant $\mu_\mcF$ such that 
\begin{align}
      \left\langle\! A_\mcF^{\top\!} (\theta)P_\mcF(\theta)\!\right\rangle\!+\!\dot{P}_\mcF(\theta)  & \!\leq\! -\mu_\mcF P_\mcF(\theta), \ \forall (\theta,\dot{\theta}) \in \Theta \!\times \!\Theta_d,\nonumber\\
       P_\mcF(\theta) & >0,\ \forall \theta \in \Theta. \label{eq:laypunov_stability_mcF}
\end{align}
\vspace{-5mm} \end{assumption}
\begin{remark}
\cref{assump:F-theta-stability} essentially means that there exists a quadratic PD Lyapunov function $V(t,\theta) \trieq x_\mcF^{\top\!}(t) P_\mcF(\theta) x_\mcF(t)$ guaranteeing the stability of the system defined in \eqref{eq:F-theta-ss-form}.
\vspace{0mm} \end{remark}
{\begin{remark}
The only decision variables in \eqref{eq:laypunov_stability_mcF} are $ P_\mcF(\theta)$ and $\mu_\mcF$. Thus, \cref{assump:F-theta-stability} can be verified by performing a line search for $\mu_\mcF$ and solving  a PD-LMI problem for each $\mu_\mcF$. 
\vspace{0mm} \end{remark}}

\subsection{Computational perspectives}\label{sec:sub-computation-perspective}
\ul{\it Solving PD-LMI  conditions:} Up to now, LMI technique is still the most common and effective technique for analysis and control synthesis for LPV systems \cite[Section~1.3]{Moh12LPVBook}. Therefore, it is no surprise that (PD-)LMIs are adopted in this work for (i) \red{verifying the Lyapunov stability condition \cref{eq:laypunov_stability_Am} in \cref{assump:desired-dynamics-stable-Lyapunov}}; (ii)
determining the state bounds for autonomous LPV systems (\cref{lem:xin_bound}), (iii)
computing the PPG bounds for LPV systems (\cref{lem:ppg-bound-computaion-lmi}), and (iv) designing the
feedforward mapping for unmatched uncertainty attention (\cref{lem:p2p-synthesis-final}). 
The PD-LMIs can be solved by applying the gridding technique for general parameter dependence  \cite{Apk98}. Under special circumstances like polynomial parameter dependence, techniques 
employing \red{multiconvexity concepts \cite{Apk00LMI}, homogeneous polynomially parameter-dependent solution \cite{oliveira2007pd-lmis} and 
sum-of-square relaxations \cite{wu2005sos-lpv}, and tools such as Robust LMI Parser} \cite{agulhari2019robustLMI-parser} can be leveraged to obtain a finite number of sufficient LMI conditions. 

\underline{\it Incorporation of unknown input gain matrix:}
In the presence of an uncertain input gain, to verify the stability condition in \eqref{eq:l1-stability-condition}, we need to compute the PPG bounds of several mappings (i.e., $\mcG_{xm}(\theta),\ \mcG_{xum}(\theta)$ and $\mcH_{xm}(\theta)\mcC K_r(\theta)$) involving the uncertain input gain  $\omega$. Therefore, $\omega$ will appear in the LMI conditions \cref{eq:ppg-bound-computation-lmi1,eq:ppg-bound-computation-lmi2} when we apply Lemma~\ref{lem:ppg-bound-computaion-lmi}. Fortunately, under proper state-space realizations, $\omega$ will appear linearly in the state-space matrices of those mappings (as demonstrated in Appendix~\ref{sec:ss-realization}), and thus appear linearly in the LMI conditions in  Lemma~\ref{lem:ppg-bound-computaion-lmi}. Considering that $\mathrm{vec}(\omega)\in \hat\Omega$ with $\hat\Omega$ being a polytope (Assumption~\ref{assump:unknown-input-gain}), to account for the uncertainty in $\omega$, we can formulate robust LMI conditions using the vertices of $\hat \Omega$; the resulting LMIs will be independent of $w$.

\section{Stability and Performance Analysis}\label{sec:analysis_l1}
In this section, we will analyze the performance of the adaptive closed-loop system with the controller defined in \cref{eq:state-predictor,eq:adaptive_law,eq:l1-control-law}, and derive performance bounds for both transient and steady-state phases under certain conditions. The characterization of the performance of the adaptive system uses a non-adaptive reference system, defined as 
\begin{align}
     \dot{x}_\rt(t) &=  A_m\left(\theta\right)x_\rt(t)+ B\left(\theta\right)\left(\omega u_\rt(t) + f_1(t,x_\rt(t)\right) 
     +\Bu(\theta)f_2(t,x_\rt(t)), \ x_\rt(0) = x_0, \nonumber \\
    u_\rt  &= -\omega^{-1}\mcC\left(\eta_{1\rt}+\barH(\theta)\eta_{2\rt}  -K_r(\theta)r\right), \label{eq:reference-system}
    \\
y_\rt(t) &= C_m(\theta) x_\rt(t), \nonumber
\end{align}
where 
\begin{equation}\label{eq:eta_ir-def}
  \eta_{i\rt}(t)\trieq f_i(t,x_{\rt}(t)), \quad i=1,2.  
\end{equation}
 It can be seen that the control law in the reference system aims to {\it partially} cancel the matched uncertainty within the bandwidth of the filter $\mcC(s)$ and reject the unmatched uncertainty based on the idea introduced in \cref{sec:sub-unmatched-dist-rejection}. Moreover, the control law utilizes the true uncertainties and is thus {\it not implementable}. Therefore, the reference system is only used to help characterize the performance of the adaptive system. With the reference system, performance analysis of the adaptive system will be performed in four steps: (i) establish the boundedness of the reference system (\cref{lem:xref-bounds}); (ii) quantify the difference between the state, input and output  signals of the adaptive system and those of the reference system (\cref{them:x-xref-bounds}); (iii) quantify the difference between the state, input and output signals of the reference system and those of the ideal system (\cref{lem:ref-id-bound}); (iv) based on the results from (ii) and (iii),  quantify the difference between the state, input and output signals of the adaptive system and those of the ideal system (\cref{them:x-xid-bounds}).
{\subsection{Stability condition}
Given the uncertain system dynamics in \cref{eq:plant-dynamics-f12}, a constant $\rho_0$ (the bound on $\norm{x_0}$), a baseline controller \cref{eq:control-baseline-general}  
that guarantees the stability of the nominal CL system (\cref{assump:desired-dynamics-stable-Lyapunov}), and the prior knowledge of uncertainties represented in \cref{assump:unknown-input-gain,assump:ftx-semiglobal_lipschitz,assump:ftx-rate-bounded,assump:ft0-bound}, for guaranteeing the stability and tracking performance, the design of $K$ in \cref{eq:filter-defn} needs to ensure that there exists a positive constant $\rho_r$ such that the following condition holds: \begin{equation}\label{eq:l1-stability-condition}
\lonenormbar{\mcG_{xm}(\theta)}({L_{f_1}^\rho\rho_r+b_{f_1}^0}) + \lonenormbar{\mcG_{xum}(\theta)} ({L_{f_2}^\rho\rho_r+b_{f_2}^0}) < {\rho_r-\lonenormbar{\mcH_{xm}(\theta)\mcC K_r(\theta)}\linfnorm{r} - \rho_\textup{in}},
\end{equation}
where \red{$\rho_\textup{in}$ is defined in \eqref{eq:rhoin-defn}}, $\mcC$ denotes the input-output mapping of $\mcC(s)$, and 
\begin{align}\label{eq:rho-defn}
    \rho & \triangleq \rho_r + \gamma_1,
\end{align}
with ${\gamma}_1$ being a positive constant that is selected by users and therefore can be arbitrarily small. Additionally, $L_{f_i}^\rho$ $(i=1,2)$ are defined in \cref{eq:L2delta-defn,eq:L1delta-defn}, $K_r(\theta)$ is the feedforward gain introduced in \eqref{eq:ideal-dynamics},  and $\mcG_{xm}(\theta)$ and $\mcG_{xum}(\theta)$ are defined in \cref{eq:Gxm-Gxum-defn}.

\begin{remark}
The condition in \eqref{eq:l1-stability-condition} will be used to guarantee the stability of a non-adaptive reference system (defined later in \eqref{eq:reference-system}) and derive the performance bounds of the closed-loop adaptive system with respect to the reference system (see Theorem~\ref{them:x-xref-bounds}) in \cref{sec:analysis_l1}. More specifically, we will show that the states of the reference system and the adaptive systems are bounded by $\rho_r$ and $\rho$, respectively, while the bound on their difference, i.e., $\gamma_1$, can be made arbitrarily small. \red{Additionally, the first (second) term on the left-hand side of  \eqref{eq:l1-stability-condition} denote a bound on $\norm{\xref}$ associated with partially compensated matched (unmatched) uncertainties under zero initial condition, while the term $\lonenormbar{\mcH_{xm}(\theta)\mcC K_r(\theta)}\linfnorm{r}$ on the right-hand side denotes a bound associated with the reference signal $r$ under zero initial condition.} 
\end{remark}
Also, the stability condition \eqref{eq:l1-stability-condition} uses PPG bounds instead of $\lone$ norms, generally used in existing \loneAC~work with desired dynamics represented by LTI models \red{\cite[Chapters~2--4]{naira2010l1book}}. This is because for LPV systems it is difficult to compute the $\lone$ norm exactly, while an upper bound for the PPG can be computed using LMIs, as shown in Lemma~\ref{lem:ppg-bound-computaion-lmi}. 
{\begin{remark}
In the absence of UUs, condition \eqref{eq:l1-stability-condition} degrades to $\lonenormbar{\mcG_{xm}(\theta)}({L_{f_1}^\rho\rho_r+b_{f_1}^0})
< \rho_r-\lonenormbar{\mcH_{xm}(\theta)\mcC K_r(\theta)}\cdot \\ \linfnorm{r} - \rho_\textup{in},$ which can always be satisfied by increasing the filter bandwidth since $\mcG_{xm}(\theta)$ tends towards $0$ when the filter bandwidth becomes infinite. Additionally, if all the terms in \eqref{eq:l1-stability-condition} except $\rho_r$ are finite and $L_{f_i}^\rho$ ($i=1,2$) are sufficiently small regardless of $\rho_r$ (meaning that the uncertainties $f_i(t,x)$ ($i=1,2$) change sufficiently slowly with respect to $x$), such that $\lonenormbar{\mcG_{xm}(\theta)}L_{f_1}^\rho + \lonenormbar{\mcG_{xum}(\theta)}L_{f_2}^\rho
< 1$, then we can always find a large enough $\rho_r$ such that \eqref{eq:l1-stability-condition} holds. 
\vspace{0mm} \end{remark}}
\begin{remark}\label{rem:stability-condition-conservative}
The condition in \eqref{eq:l1-stability-condition} could be rather conservative because it is derived using small gain theorem (see the proof of \cref{lem:xref-bounds}), and furthermore, the PPG bound computed using LMI techniques could be quite loose. Nevertheless, it provides a verifiable way for theoretic guarantees, when addressing the challenging problem of adaptive control of LPV systems subject to unmatched nonlinear uncertainties. 
\end{remark}}


\vspace{-3mm}
\subsection{Boundedness of the reference system}
We now present a lemma to  establish the boundedness of the reference system.
\begin{lemma}\label{lem:xref-bounds}
Given the closed-loop reference system in \eqref{eq:reference-system}, if   {\cref{assump:desired-dynamics-stable-Lyapunov}}, {conditions \cref{eq:fi-lipschitz-cond,eq:fit0-bound}}, and the stability condition \eqref{eq:l1-stability-condition} hold, we have
\vspace{-4mm}\\
\begin{align}
 \linfnorm{x_\rt} &<\rho_r \label{eq:xref-bound}, \quad \linfnorm{u_\rt} <\rho_{ur}, 
\end{align}\vspace{-4mm}\\
where 
\vspace{-3mm}
\begin{align}
   \rho_{ur}  \triangleq& \lonenorm{\omega^{-1}\mcC(s)}\left(\max_{\theta\in\Theta}\norm{K_r(\theta)}\linfnorm{r}+L_{f_1}^{\rho_r}\rhor +b_{f_1}^0+\lonenormbar{\barH(\theta)}(L_{f_2}^{\rho_r}\rhor+b_{f_2}^0)\right) \label{eq:rho-ur-defn}. 
\end{align}\vspace{-4mm}\\
\vspace{-4mm}\end{lemma}
\begin{proof}
It follows from \eqref{eq:reference-system},  the definition of $\mcH_{xm}(\theta)$  above \eqref{eq:lpv-system-no-nonlinearity}, and the definitions of $\mcG_{xm}(\theta)$ and $\mcG_{xum}(\theta)$ in \eqref{eq:Gxm-Gxum-defn} that 
\begin{equation*}
    x_\rt = \mcG_{xm}(\theta) \eta_{1\rt}+\mcG_{xum}(\theta) \eta_{2\rt}+  \mcH_{xm}(\theta)\mcC K_r(\theta)r + x_\textup{in}.
\end{equation*}
Due to \cref{assump:desired-dynamics-stable-Lyapunov}, according to Lemma~\ref{lem:input-output-ppg-bound-relation}, we have 
{
\begin{flalign}
 &\linfnormtruc{x_{\rt}}{\tau} \leq   \linfnorm{x_\textup{in}} +  \lonenormbar{\mcG_{xm}(\theta)}\linfnormtruc{\eta_{1\rt}}{\tau} +   \lonenormbar{\mcG_{xum}(\theta)}\!\linfnormtruc{\eta_{2\rt}}{\tau}  
   +  \lonenormbar{\mcH_{xm}(\theta)\mcC K_r(\theta)}\!\linfnorm{r}\!.  \label{eq:xref-tau-bound-eta-ref}
\end{flalign}}
If the bound on $x_\rt$ in \eqref{eq:xref-bound} is not true, since $\norm{x_\rt(0)}=\norm{x_0}<\rho_r$ and $x_\rt(t)$ is continuous, then there exists $\tau>0$ such that
$\norm{x_\rt(t)}<\rho_r$, for any $t\in[0,\tau)$, and $\norm{x_\rt(\tau)}=\rho_r,$ which implies that
\begin{equation}\label{eq:xref-tau-linfnorm-equal-to-rhor}
    \linfnormtruc{x_{\rt}}{\tau} = \rho_r.
\end{equation}
Further considering conditions \cref{eq:fi-lipschitz-cond,eq:fit0-bound}, we have
\begin{equation}\label{eq:eta_iref-tau-bound}
    \linfnormtruc{\eta_{i\rt}}{\tau} \leq  L_{f_i}^{\rho_r}\rho_r+b_{f_i}^0, \quad i=1,2.
\end{equation}
Plugging the preceding inequality into \eqref{eq:xref-tau-bound-eta-ref} results in
$
  \linfnormtruc{x_{\rt}}{\tau}  \leq  \linfnorm{x_\textup{in}} +  \lonenormbar{\mcG_{xm}(\theta)} (L_{f_1}^{\rho_r}\rho_r+b_{f_1}^0) 
    + \lonenormbar{\mcG_{xum}(\theta)}\cdot \\ (L_{f_2}^{\rho_r}\rho_r+b_{f_2}^0)  
   +  \lonenormbar{\mcH_{xm}(\theta)\mcC K_r(\theta)}\linfnorm{r} 
   <  \rho_r-\rho_\textup{in} +\linfnorm{x_\textup{in}} \leq \rho_r, \label{eq:xref-tau-bound-rhor}
$
where 
the third and last inequalities are due to \eqref{eq:l1-stability-condition} (together with the fact that $L_{f_i}^{\rho_r}\leq L_{f_i}^\rho$ for $i=1,2$) and \eqref{eq:xin-linfnorm-rho-in-relation}, respectively.
As a result, we arrive at $\linfnormtruc{x_{\rt}}{\tau} <\rho_r$, contradicting \eqref{eq:xref-tau-linfnorm-equal-to-rhor}.  This proves the bound on $x_\rt$ in \eqref{eq:xref-bound}. 
Using \cref{eq:xref-bound,eq:eta_iref-tau-bound} and the definition of $u_\rt$ in \eqref{eq:reference-system}, we obtain 
the bound on $u_\rt$ in \eqref{eq:xref-bound}.  \qedclosed
\end{proof}

Let $\bar{\gamma}_0$ be a small positive constant such that 
{
\begin{equation}\label{eq:underline-gamma_1-defn-constr}
\frac{\lonenormbar{\mcH_{xm}(\theta)}}{1-\norm{\mcG_{xm}(\theta)}_\lonebar   L_{f_1}^\rho-\lonenormbar{\mcG_{xum}(\theta)}L_{f_2}^\rho} \bar{\gamma}_0  < \gamma_1,
\end{equation}}where $\gamma_1$ is introduced in \cref{eq:rho-defn}. 
We further define
\begin{equation}\label{eq:rho_u-defn}
\rho_u \triangleq \rho_{ur} +\gamma_2,
\end{equation}
\vspace{-5mm}
\begin{align}
    \gamma_2  \triangleq  \left( {{{\| {{\omega ^{ - 1}}{\mcC}}\|}_{{{\mathcal L}_1}}}{L_{f_1}^\rho} \!+\! {{\| {{\omega ^{ - 1}}{\mcC}\bar {\mathcal H} (\theta )} \|}_{{{\bar {\mathcal L} }_1}}}{L_{f_2}^\rho}} \right){\gamma _1}
  \!  + \!\| {{\omega ^{ - 1}}}\|{{\bar \gamma }_0}. \label{eq:gamma2-defn}
\end{align}
We will show (in Theorem~\ref{them:x-xref-bounds}) that $\rho$ and $\rho_u$ in \eqref{eq:rho-defn} and \eqref{eq:rho_u-defn} are the bounds for the states and inputs of the adaptive closed-loop system with the controller defined by \cref{eq:state-predictor,eq:adaptive_law,eq:l1-control-law}. 
\vspace{-3mm} 
\subsection{Transient and steady-state performance} \label{sec:transient-ss-perf}
For the derivations in this section, let us first define 
\begin{subequations}\label{eq:eta-i-tileta-i-defn} 
\begin{align}
\eta_i(t) & \triangleq f_i(t,x(t)),\quad i=1,2, \label{eq:etai-defn}  \\
     \tileta_m(t) &\triangleq \hsigma_m(t)-\left((\omega-\mbI_m)u(t)+\eta_1(t)\right),\quad   \\
     \tileta_{um}(t) &\triangleq \hsigma_{um}(t)-\eta_2(t).
\end{align}
\end{subequations}
From \cref{eq:plant-dynamics-f12-a-fhat}and \cref{eq:state-predictor}, we obtain the prediction error dynamics:
\begin{align}
   &  \dot{\tilde{x}}(t) = -a \tilx(t) + \hsigma(t) - B(\theta)(\omega-\mbI_m)u(t) -{{\bar f}(t,x(t))} = -a \tilx(t) + \hsigma(t)-\sigma(t,x(t),u(t)), \  \tilx(0) = 0,\hspace{-3mm} \label{eq:prediction-error}
\end{align}
where $\sigma(t,x(t),u(t))$ is defined in \cref{eq:sigma-txu-defn}. 
Recalling from \eqref{eq:hsigma_m_um_hsigma_relation} that
$
    \hsigma(t) = B(\theta)\hsigma_m(t) + \Bu(\theta)\hsigma_{um}(t),
$
and considering \eqref{eq:f1-f2-f-relation}, \eqref{eq:sigma-txu-defn} and \eqref{eq:eta-i-tileta-i-defn}, we have 
\begin{subequations}
\begin{align}
  \tilsigma(t) & \trieq  \hsigma(t)-\sigma(t,x(t),u(t)) \label{eq:tilsigma-defn}\\
  & = B(\theta)\tileta_m(t) +  \Bu(\theta) \tileta_{um}(t), \label{eq:sigma-error-m-um-split-sigmatilde-defn}
\end{align}
\end{subequations}
which further implies that 
{\renewcommand*{\arraystretch}{1.2}
\begin{equation}\label{eq:tileta-m-um-sigma-error}
   \begin{bmatrix}
       \tileta_m(t) \\
       \tileta_{um}(t)
   \end{bmatrix}  = 
   \begin{bmatrix}
       B^\dagger(\theta) \\
        \Bu^\dagger(\theta)
   \end{bmatrix}
   \tilde\sigma(t).
\end{equation}}
\vspace{-1mm}
For following derivations, let us first define
{
\begin{align}
\hspace{-2mm}b_\sigma(\rho,\rho_u)  \triangleq &\  b_{B}\max_{\omega\in\Omega}\norm{\omega\!-\!\mbI_m} (\rho_u{+b_{K_x}\rho}) 
+ L_f^\rho\rho+b_f^0, \label{eq:b-sigma-defn} \\
\hspace{-2mm} b_{\dot{x}}(\rho,\rho_u) \triangleq &\ b_{A_m}\rho + b_{B}\rho_u \max_{\omega\in\Omega}\norm{\omega} +   L_f^\rho\rho+b_f^0  {+ b_B\max\limits_{\omega\in\Omega} \norm{(\omega\!-\!\mbI_m)}b_{K_x}\rho}, \label{eq:b-xdot-defn}\\
\hspace{-2mm} b_{\hsigma}(\rho,\rho_u) \triangleq & \  e^{-aT}\sqrt{n}~b_\sigma(\rho,\rho_u), \label{eq:b-hsigma-defn}\\
b_{\dot u}(\rho,\rho_u) \triangleq &\ \left\| K \right\|\left[{{\rho _u}} + ( 1+ \left\| {\bar {\mathcal H} (\theta )} \right\|_{{\bar {\mathcal L} }_1} )b_{B^\dag}{b_{\hat \sigma}}(\rho ,{\rho _u})
+ b_{K_r} {\left\| r \right\|}_{{\mathcal L}_\infty}  \right], \label{eq:b-dotu-defn} \\
\hspace{-2mm} b_{\tilsigma}(\rho,\rho_u) \triangleq&\   2\sqrt{n}~T\left[\left\| {\omega  - {\mbI_m}} \right\|( {{b_{{B}}}{b_{\dot u}}(\rho ,{\rho _u}) + L_B{b_{\dot \theta }}{\rho _u}} ) + {l_f} + {L_f^\rho }{b_{\dot x}}(\rho ,{\rho _u}) \right] 
    + (1-e^{-aT}) \sqrt{n}~b_\sigma(\rho,\rho_u). \label{eq:b-tilsigma-defn} 
\vspace{-8mm}
\end{align}}
\begin{lemma}\label{lem:sigmatilde-xtilde-bound}
Given the uncertain system in \eqref{eq:plant-dynamics-f12} satisfying \cref{assump:B-lipschitz-bounded,assump:ftx-rate-bounded,assump:ftx-semiglobal_lipschitz,assump:ft0-bound,assump:unknown-input-gain}, the state predictor in \eqref{eq:state-predictor}, and the adaptive law in \eqref{eq:adaptive_law}, if
\begin{equation}\label{eq:x-u-tau-bound-assump-in-lemma}
    \linfnormtruc{x}{\tau}\leq \rho, \quad \linfnormtruc{u}{\tau}\leq \rho_u,
\end{equation}
then 
 \vspace{-4mm}
\begin{align}
    \norm{\tilsigma(t)}&\leq \left\{
    \begin{array}{l}
    b_{\sigma}(\rho,\rho_u),\quad \forall t\in[0,T), \\
     b_{\tilsigma}(\rho,\rho_u),\quad \forall t\in [T,\tau].
    \end{array}
    \right. \label{eq:tilsigma-bound}
\end{align}\vspace{-4mm}\\
\vspace{-4mm}\end{lemma}
\begin{proof}
In the following proof,  we often omit the dependence of $\sigma$ on $x$ and $u$, i.e., simply writing $\sigma(t)$ for $\sigma(t,x(t),u(t))$, for notation brevity.  {As explained in the proof of \cref{lem:f-f1-f2-constants-relation},  \cref{assump:ftx-semiglobal_lipschitz,assump:ft0-bound} imply $\norm{f(t,x)}\leq  L_f^\rho \rho + b_f^0$, for all  $\norm{x}\leq \rho$ and for all $t\geq 0$, 
which, together with the definition of ${\bar f}(t,x)$ in \cref{eq:f1-f2-f-relation} and  \eqref{eq:x-u-tau-bound-assump-in-lemma}, indicate 
\begin{equation}\label{eq:f-bar-bound-in-rho}
 \norm{{\bar f}(t,x)}\leq L_f^\rho \rho + b_f^0 + b_B\max\limits_{\omega\in\Omega} \norm{(\omega\!-\!\mbI_m)}b_{K_x}\rho,\quad
 \forall t\in[0,\tau].  
\end{equation}
From \cref{eq:f-bar-bound-in-rho} and the definition of $\sigma(t,x,u)$ in \eqref{eq:sigma-txu-defn}, we have  
\begin{flalign}
  &   \norm{\sigma(t,x,u)} \leq \norm{B(\theta)}\norm{\omega-\mbI_m}\rho_u +\norm{{{\bar f}(t,x)}} \leq  L_f^\rho \rho+b_f^0+ b_{B}\max_{\omega\in\Omega}\norm{\omega\!-\!\mbI_m} (\rho_u\!+\!b_{K_x}\rho) = b_\sigma(\rho,\rho_u), \label{eq:sigma-bound-in-tau}
\end{flalign}
for any $t$ in $[0,\tau]$}. Equations \cref{eq:f-bar-bound-in-rho,eq:plant-dynamics-f12-a-fhat,eq:x-u-tau-bound-assump-in-lemma} imply 
\begin{equation}\label{eq:xdot-b-xdot-relation}
    \norm{\dot{x}(t)} \leq \norm{A_m(\theta)}\norm{x(t)}+\norm{B(\theta)}\norm{\omega}\norm{u(t)}
    +\norm{{{\bar f}(t,x)}} \leq b_{\dot{x}}(\rho,\rho_u), \quad  \forall t\in[0,\tau],  
\end{equation}
where $b_{\dot{x}}(\rho,\rho_u)$ is defined in \eqref{eq:b-xdot-defn}. From the prediction error dynamics in \eqref{eq:prediction-error}, we have for $0\leq t <T$, $i\in\mbZ_0$,
{\begin{align}
   & \tilx(iT+t) = e^{-at}\tilx(iT)\!+\!\int_{iT}^{iT+t}\! e^{-a(iT+t-\xi)}(\hsigma(\xi)\!-\!\sigma(\xi))d\xi  =  e^{-at}\tilx(iT)+\int_{iT}^{iT+t}\! e^{-a(iT+t-\xi)}(\hsigma(iT)\!-\!\sigma(\xi))d\xi,\label{eq:tilx-iTplust}
    \vspace{-3mm}
\end{align}}which implies that for $i\in\mbZ_0$,
$
     \tilx((i\!+\!1)T) = e^{-aT}\tilx(iT)\!+\!\int_{iT}^{(i\!+\!1)T} e^{-a((i\!+\!1)T-\xi)}\hsigma(iT)d\xi 
    -\int_{iT}^{(i\!+\!1)T} e^{-a((i\!+\!1)T-\xi)} \sigma(\xi)d\xi
$.
Plugging the adaptive law in \eqref{eq:adaptive_law} into the above equation gives 
{\begin{equation}\label{eq:xtilde-i1T-sigma}
     \tilx((i+1)T) = 
  -\int_{iT}^{(i+1)T} e^{-a((i+1)T-\xi)} \sigma(\xi)d\xi.
\end{equation}}
Furthermore, since $e^{-a((i+1)T-\xi)}$ is always positive and $\sigma(\xi)$ is continuous, we can apply the first mean value theorem in an element-wise manner
to \eqref{eq:xtilde-i1T-sigma}. Doing this gives $\tilx((i+1)T) = -\int_{iT}^{(i+1)T} e^{-a((i+1)T-\xi)} d\xi \bm{[}\sigma_j(\tau_j^*)\bm{]} $, further implying
\begin{equation}\label{eq:xtilde-iPlus1T-tau-star}
     \tilx((i+1)T) =  -\frac{1}{a}(1-e^{-aT})\bbracket{\sigma_j(\tau_j^*)}, 
\end{equation}
for some $\tau_j^*\in(iT,(i+1)T)$, for $j\in\mbZ_1^n$ and $i\in\mbZ_0$, where $\sigma_j(t)$ is the $j$-th element of $\sigma(t)$, and $$\bbracket{\sigma_j(\tau_j^*)}\triangleq 
[\sigma_1(\tau_1^*), \sigma_2(\tau_2^*),\cdots, \sigma_n(\tau_n^*)
]^{\top\!}.$$ The adaptive law in \eqref{eq:adaptive_law} and the equality \eqref{eq:xtilde-iPlus1T-tau-star} indicate that for any $t\in\left[(i+1)T,  (i+2)T\right) \\ \cap [0,\tau]$ with $i\in\mbZ_0$, there exist $\tau_j^*\in(iT,(i+1)T)$ ($j\in\mbZ_1^n$) such that 
\begin{align}
    \hsigma(t) = -\frac{a}{e^{a T}-1}\tilx((i+1)T)= e^{-aT}\bbracket{\sigma_j(\tau_j^*)}. \label{eq:hsigma-sigma-relation}
\end{align}
Note that 
\begin{align}
    \norm{\sigma(t)-\bbracket{\sigma_j(\tau_j^*)}}\leq  \sqrt{n}\infnorm{\sigma(t)-\bbracket{\sigma_j(\tau_j^*)}} 
    =  \sqrt{n}\abs{\sigma_{j^\star_t}(t)-{\sigma_{j^\star_t}(\tau_{j^\star_t}^*)}} \leq \sqrt{n}\norm{\sigma(t)-{\sigma(\tau_{j^\star_t}^*)}},\label{eq:sigma-sigmaStar-index-conversion}
\end{align}where $j^\star_t=\arg\max_{j\in\mbZ_1^n} \abs{\sigma_j(t)-{\sigma_j(\tau_j^*)}}$. Similarly, we have 
{\begin{align}
    \norm{\bbracket{\sigma_j(\tau_j^*)}}  & \leq \sqrt{n}\abs{{\sigma_{{\bar j^\star}}(\tau_{{\bar j^\star}}^*)}} \leq \sqrt{n}\norm{{\sigma(\tau_{{\bar j^\star}}^*)}},\label{eq:sigmaStar-index-conversion}
\end{align}}
where ${\bar j^\star}=\arg\max_{j\in\mbZ_1^n} \abs{{\sigma_j(\tau_j^*)}}$. The inequality \eqref{eq:sigmaStar-index-conversion}, together with \eqref{eq:hsigma-sigma-relation}, implies that for any $t$ in $[T,\tau]$, we have
$
    \norm{\hsigma(t)}\leq e^{-aT}\sqrt{n}~b_\sigma(\rho,\rho_u) = b_{\hsigma}(\rho,\rho_u),
$ where $b_{\hsigma}(\rho,\rho_u)$ is defined in \eqref{eq:b-hsigma-defn}. 
Further considering that $\hsigma(t)=0$ in $[0,T)$,
we have
$
\linfnormtruc{\hsigma}{\tau}\leq b_{\hsigma}(\rho,\rho_u).
$
On the other hand, the equality in \eqref{eq:hsigma_m_um_hsigma_relation} indicates that
\begin{equation*}\label{eq:sigmahat_m_um-b-sigmac-relation}
 \linfnormtruc{\begin{bmatrix}[1]
{{{\hat \sigma }_m}}\\
{{{\hat \sigma }_{um}}}
\end{bmatrix}}{\tau}  \le {\max _{\theta  \in \Theta }}\left\| {{\bar B^{ - 1}}(\theta )} \right\|{ \linfnormtruc{\hat \sigma}{\tau}} \le {b_{B^\dag}} {b_{\hsigma}}(\rho ,{\rho _u}),
\end{equation*}
which, together with the control law in \eqref{eq:l1-control-law}, implies 
{\begin{align}
\linfnormtruc{\dot u}{\tau}  & \le  \left\| K \right\|\left[ {\rho _u} + \linfnormtruc{\hsigma_m}{\tau} + {{\left\| {\bar {\mathcal H} (\theta )} \right\|}_{{{\bar {\mathcal L} }_1}}}{\linfnormtruc{\hsigma_{um}}{\tau} }\!  +\! b_{K_r} \linfnorm{r} \right ] \nonumber \\
 &  \le\left\| K \right\|\left[ {{\rho _u}} + \left( 1+ \left\| {\bar {\mathcal H} (\theta )} \right\|_{{\bar {\mathcal L} }_1} \right) b_{B^\dag} {b_{\hat \sigma}}(\rho ,{\rho _u}) \right]   + \!\left\| K \right\| b_{K_r} {\left\| r \right\|}_{{\mathcal L}_\infty}  \!\!= b_{\dot u}(\rho,\rho_u). \hspace{-2mm} \label{eq:udot-b-udot-relation}
\end{align}}
Define $\tilde \omega \trieq \omega-\mbI_m$. With the bounds on $\norm{\dot x(t)}$ and $\norm{\dot u(t)}$ in \eqref{eq:xdot-b-xdot-relation} and \eqref{eq:udot-b-udot-relation}, a bound on $\left\| {\sigma (t) - \sigma (\tau _{j^\star_t}^*)} \right\| $ (used in \eqref{eq:sigma-sigmaStar-index-conversion}) is given by \begin{subequations}
\begin{align}
 \nonumber 
& \left\| {\sigma (t) - \sigma (\tau _{j^\star_t}^*)} \right\| = \left\| {{B}(\theta (t))\tilde \omega u(t) + {\bar f}(t,x(t)) - {B}(\theta (\tau _{j^\star_t}^*))\tilde \omega u(\tau _{j^\star_t}^*) - {\bar f}(\tau _{j^\star_t}^*,x(\tau _{j^\star_t}^*))} \right\| \nonumber \\
 &= \left\| { {{B}(\theta (t))\tilde \omega\left( {u(t) - u(\tau _{j^\star_t}^*)} \right) + \left( {{B}(\theta (t)) - {B}(\theta (\tau _{j^\star_t}^*))} \right)\tilde \omega u(\tau _{j^\star_t}^*)} + {\bar f}(t,x(t)) - {\bar f}(\tau _{j^\star_t}^*,x(\tau _{j^\star_t}^*))} \right\| \nonumber  \\
 &\le \left\| {\tilde \omega} \right\|\left( {\left\| {{B}(\theta (t))} \right\|\left\| {u(t) - u(\tau _{j^\star_t}^*)} \right\| + \left\| {{B}(\theta (t)) - {B}(\theta (\tau _{j^\star_t}^*))} \right\|\left\| {u(\tau _{j^\star_t}^*)} \right\|} \right) + \left\| {{\bar f}(t,x(t)) - {\bar f}(\tau _{j^\star_t}^*,x(\tau _{j^\star_t}^*))} \right\| \nonumber  \\
 &\le \left\| {\tilde \omega} \right\|\left( {\left\| {{B}(\theta (t))} \right\|\left\| {\mathop \smallint \nolimits_{\tau _{j^\star_t}^*}^t \dot u(\xi )d\xi } \right\| + \left\| {{B}(\theta (t)) - {B}(\theta (\tau _{j^\star_t}^*))} \right\|\left\| {u(\tau _{j^\star_t}^*)} \right\|} \right) + {l_f}(t - \tau _{j^\star_t}^*) + L_f^\rho \left\| {x(t) - x(\tau _{j^\star_t}^*)} \right\| \nonumber  \\
 &\le \left\| {\tilde \omega} \right\|\left( {{b_{{B}}}\mathop \smallint \nolimits_{\tau _{j^\star_t}^*}^t \left\| {\dot u(\xi )} \right\|d\xi  + L_B\left\| {\theta (t) - \theta (\tau _{j^\star_t}^*)} \right\|\left\| {u(\tau _{j^\star_t}^*)} \right\|} \right) + {l_f}(t - \tau _{j^\star_t}^*) + L_f^\rho \left\| {\mathop \smallint \nolimits_{\tau _{j^\star_t}^*}^t \dot x(\xi )d\xi } \right\| \label{eq:sigma-sigmaStar-bound-intermed-1} \\
 &\le \left\| {\tilde \omega} \right\|\left( {{b_{{B}}}{b_{\dot u}}(\rho ,{\rho _u})(t - \tau _{j^\star_t}^*) + L_B\left\| {\mathop \smallint \nolimits_{\tau _{j^\star_t}^*}^t \dot \theta (\xi )d\xi } \right\|{\rho _u}} \right) + {l_f}(t - \tau _{j^\star_t}^*) + L_f^\rho {b_{\dot x}}(\rho ,{\rho _u})(t - \tau _{j^\star_t}^*) \label{eq:sigma-sigmaStar-bound-intermed-2}\\
 &\le \! \left( {\left\| {\tilde \omega}  \right\|\! \left( {{b_{{B}}}{b_{\dot u}}(\rho ,{\rho _u}) \!  + \! L_B{b_{\dot \theta }}{\rho _u}} \! \right) \!  + \! {l_f} \!  + \! L_f^\rho {b_{\dot x}}(\rho ,{\rho _u})} \! \right)(t \!  -\! \tau _{j^\star_t}^*) \le 2T\! \left( {\left\| {\tilde \omega}  \right\|\! \left( {{b_{{B}}}{b_{\dot u}}(\rho ,{\rho _u}) \!  + \! L_B{b_{\dot \theta }}{\rho _u}} \! \right) \!  + \! {l_f} \!  + \! {L_f^\rho }{b_{\dot x}}(\rho ,{\rho _u})} \! \right) \label{eq:sigma-sigmaStar-bound}
\end{align}
\end{subequations}
where \cref{eq:sigma-sigmaStar-bound-intermed-1} is due to \cref{eq:B-lipschitz-const}, \cref{eq:sigma-sigmaStar-bound-intermed-2} is due to \cref{eq:udot-b-udot-relation}, and  \eqref{eq:sigma-sigmaStar-bound} 
is due to the fact of $t\in\left[(i+1)T,  (i+2)T\right)\cap [0,\tau]$ with $i\in\mbZ_0$ and $\tau _{j^\star_t}^*\in(iT,(i+1)T)$.
Therefore, for  any $t\in [T, \tau]$, we have
$
 \norm{\tilsigma(t)} \! =\! \norm{\sigma(t)-\hsigma(t)}  
\!  = \!\norm{\sigma(t)-e^{-aT}\bbracket{\sigma_j(\tau_j^*)}} 
\!\leq \!\norm{\sigma(t)-\bbracket{\sigma_j(\tau_j^*)}}+(1\!-\!e^{-aT})  \norm{\bbracket{\sigma_j(\tau_j^*)}} 
\!  \leq  \! \sqrt{n}\norm{\sigma(t)-{\sigma(\tau_{j^\star_t}^*)}}  \! +\! (1\!-\!e^{-aT})  \sqrt{n}\norm{{\sigma(\tau_{{\bar j^\star}}^*)}},
$
where the third equality is due to \cref{eq:hsigma-sigma-relation}, and the last inequality is due to \eqref{eq:sigma-sigmaStar-index-conversion} and \cref{eq:sigmaStar-index-conversion}. 
Plugging the bounds \eqref{eq:sigma-bound-in-tau} and \eqref{eq:sigma-sigmaStar-bound}  into the preceding inequality and considering the definition in \eqref{eq:b-tilsigma-defn} lead to $ \norm{\tilsigma(t)} \leq b_{\tilsigma}(\rho,
\rho_u)$ for any $t\in [T, \tau]$. 
On the other hand, the inequality \eqref{eq:sigma-bound-in-tau} and the fact of $\hsigma(0)=0$ imply 
$
 \norm{\tilsigma(t)} =  \norm{\sigma(t}\leq b_{\sigma}(\rho,\rho_u)
$ for any $t\in[0,T)$. 
As a result, we have proved \eqref{eq:tilsigma-bound}.  \qedclosed
\end{proof}

\begin{remark}
Given finite $\rho$ and $\rho_u$ and bounded reference signal $r$ in $[0,\tau]$, the definition in  \eqref{eq:b-tilsigma-defn} implies 
$
    \lim_{T\rightarrow 0}b_{\tilsigma}(\rho,\rho_u) =0. 
$
This means that the uncertainty estimation error in $[T,\tau]$ can be made arbitrarily small by decreasing $T$, even in the presence of initialization error, i.e., $\tilx(0)\neq 0$.
\vspace{0mm} \end{remark}
\begin{lemma}\label{lem:F-theta-tilsigma-bound}
Given the dynamic mapping $\mcF(\theta)$  in \eqref{eq:F-theta-ss-form}, if Assumption~\ref{assump:F-theta-stability} holds and the input $\tilsigma$ is bounded by \eqref{eq:tilsigma-bound}, we have 
\begin{equation}\label{eq:y_F-bound}
\linfnormtruc{y_\mcF}{\tau}\leq
\gamma_0(T,\rho,\rho_u),
\end{equation}
where 
{\begin{align}
\gamma_0(T,\rho,\rho_u) \trieq  & \max_{\theta\in \Theta}\norm{C_\mcF(\theta)}\left(
\tilde\lambda \beta_\Theta(T,\mu_\mcF,P_\mcF,B_\mcF)b_{\sigma}(\rho,\rho_u) + \beta_\Theta(\infty,\mu_\mcF,P_\mcF,B_\mcF) b_{\tilsigma}(\rho,\rho_u)
\right),\label{eq:by_F-def}
\vspace{-3mm}
\end{align}}
with $\tilde\lambda\trieq \sqrt{\frac{\lammaxtheta{P_\mcF}}{\lammintheta{P_\mcF}}}$ and $\beta_\Theta(\cdot,\cdot,\cdot,\cdot) $ defined in \eqref{eq:beta-defn-in-LPV-input-output-response}. 
Furthermore,
\begin{equation}\label{eq:gamma0-T-go-to-zero}
    \lim_{T\rightarrow 0}\gamma_0(T,\rho,\rho_u) =  0.
\end{equation}
\vspace{-5mm}\end{lemma}
\begin{proof}
Since Assumption~\ref{assump:F-theta-stability} holds, applying Lemma~\ref{lem:lpv-time-domain-response-split@T} while considering the bound on $\tilsigma$ in \eqref{eq:tilsigma-bound}
immediately lead to
\begin{equation*}
    \norm{x_\mcF(t)}\leq
\left\{
\!\!\!\begin{array}{l}
\beta_\Theta(T,\mu_\mcF,P_\mcF,B_\mcF)  b_{\sigma}(\rho,\rho_u),\ \ \forall t\in[0, T],\\ \sqrt{\frac{\lammaxtheta{P_\mcF}}{\lammintheta{P_\mcF}}}\beta_\Theta(T,\mu_\mcF,P_\mcF,B_\mcF)  b_{\sigma}(\rho,\rho_u) + \beta_\Theta(\infty,\mu_\mcF,P_\mcF,B_\mcF) b_{\tilsigma}(\rho,\rho_u), \ \ \forall t\in(T,\tau],
\end{array}
\right.
\end{equation*}
which further implies
$
    \linfnormtruc{x_\mcF}{\tau} \leq \sqrt{\frac{\lammaxtheta{P_\mcF}}{\lammintheta{P_\mcF}}}\beta_\Theta(T,\mu_\mcF,P_\mcF,B_\mcF)  b_{\sigma}(\rho,\rho_u) + \beta_\Theta(\infty,\mu_\mcF,P_\mcF,B_\mcF) b_{\tilsigma}(\rho,\rho_u). 
$
The above inequality together with the output equation in \eqref{eq:F-theta-ss-form} immediately gives \eqref{eq:y_F-bound}. On the other hand, the definitions in \eqref{eq:beta-defn-in-LPV-input-output-response} and \eqref{eq:b-tilsigma-defn} indicate 
$\lim_{T\rightarrow 0} \beta_\Theta(T,\mu_\mcF,P_\mcF,B_\mcF) = 0$ and $\lim_{T\rightarrow 0} b_{\tilsigma}(\rho,\rho_u) =0$, 
which, together with the fact that all other terms on the right-hand side of \eqref{eq:by_F-def} are finite, implies \eqref{eq:gamma0-T-go-to-zero}. \qedclosed
\end{proof}

Next, we show that if $T$ is chosen to satisfy
\begin{equation}\label{eq:T-constraints}
    \gamma_0(T,\rho,\rho_u) < \bar{\gamma}_0,
\end{equation}
then we can prove the stability and derive the performance bounds of the actual closed-loop adaptive system with respect to the reference system defined in \cref{eq:reference-system}.  

\begin{theorem}\label{them:x-xref-bounds}
Consider the uncertain system \eqref{eq:plant-dynamics} subject to \cref{assum:theta-thetadot-bounded}, {the control law \cref{eq:control-composition} consisting of the baseline controller \cref{eq:control-baseline-general} and the robust adaptive controller} defined via \cref{eq:state-predictor,eq:adaptive_law,eq:l1-control-law},  and  the reference system \eqref{eq:reference-system}. {Suppose \cref{assump:F-theta-stability,assump:B-lipschitz-bounded,assump:unknown-input-gain,assump:ft0-bound,assump:ftx-semiglobal_lipschitz,assump:ftx-rate-bounded,assump:desired-dynamics-stable-Lyapunov},
the stability condition  \eqref{eq:l1-stability-condition} holds with a user-selected constant $\gamma_1$ that can be arbitrarily small, and  the constraints  \cref{eq:underline-gamma_1-defn-constr,eq:T-constraints} hold. Then,} we have 
\begin{align}
    \linfnorm{x} & \leq \rho, \label{eq:x-bound}\\
    \linfnorm{u} & \leq \rho_u, \label{eq:u-bound}\\
    \linfnorm{x_\rt-x} &\leq \gamma_1, \label{eq:xref-x-bound}\\
    \linfnorm{u_\rt-u} &\leq \gamma_2. \label{eq:uref-u-bound}\\
    \linfnorm{y_\rt-y} &\leq b_C \gamma_1,\label{eq:yref-y-bound}
    \vspace{-3mm}
\end{align}
where $\rho$, $\rho_u$ and $\gamma_2$ are defined in \cref{eq:rho-defn}, \cref{eq:rho_u-defn} and \cref{eq:gamma2-defn}, respectively. 
\end{theorem}
\begin{proof}
{As shown in \cref{sec:problem_formulation}, under the baseline controller \cref{eq:control-baseline-general}, the compositional control law \cref{eq:control-composition} and \cref{assump:B-lipschitz-bounded}, we have that (i) the dynamics in \eqref{eq:plant-dynamics} is equivalent to \eqref{eq:plant-dynamics-f12}, and (ii) conditions \cref{eq:fi-lipschitz-cond,eq:fit0-bound} hold according to \cref{lem:f-f1-f2-constants-relation} since \cref{assump:unknown-input-gain,assum:theta-thetadot-bounded,assump:B-lipschitz-bounded,assump:ft0-bound,assump:ftx-semiglobal_lipschitz,assump:ftx-rate-bounded}  hold.} 
We next use contradiction for the proof. Assume that at least one of the bounds in \eqref{eq:xref-x-bound} and \eqref{eq:uref-u-bound} does not hold. Then, since $\norm{x_\rt(0)-x(0)} =0 < \gamma_1$, $\norm{u_\rt(0)-u(0)} =0 < \gamma_2$, and $x(t)$, $x_\rt(t)$, $u(t)$ and $u_\rt(t)$ are continuous, there exists $\tau>0$ such that 
$
\norm{x_\rt(\tau)-x(\tau)} = \gamma_1 \ \textup{or} \ 
\norm{u_\rt(\tau)-u(\tau)} = \gamma_2,
$
while $
   \norm{x_\rt(t)-x(t)} < \gamma_1$ and $  \norm{u_\rt(t)-u(t)} < \gamma_2$ for any $ t$ in $[0,\tau)$.
This implies that at least one of the following equalities holds:
\begin{equation}\label{eq:xref-x-tau-infnorm}
    \linfnormtruc{x_\rt-x}{\tau} = \gamma_1, \quad \linfnormtruc{u_\rt-u}{\tau} = \gamma_2.
\end{equation}
{Since \cref{assump:desired-dynamics-stable-Lyapunov}, and conditions \cref{eq:fit0-bound}, \cref{eq:fi-lipschitz-cond} and \cref{eq:l1-stability-condition} hold}, it follows from Lemma~\ref{lem:xref-bounds} that
\begin{equation}\label{eq:xref-uref-tau-bounds}
    \linfnormtruc{x_{\rt}}{\tau} \leq \rho_r, \quad \linfnormtruc{u_{\rt}}{\tau} \leq \rho_{ur}.
\end{equation}
The definitions of $\rho$ and $\rho_u$ in \eqref{eq:rho-defn} and \eqref{eq:rho_u-defn}, and the bounds in \cref{eq:xref-x-tau-infnorm} and \eqref{eq:xref-uref-tau-bounds} imply that 
\begin{align}
  \linfnormtruc{x}{\tau}  &\leq \rho_r + \gamma_1 = \rho, \label{eq:x-tau-rhor-rho-relation} \\
    \linfnormtruc{u}{\tau} & \leq \rho_{ur}+\gamma_2 = \rho_u. \label{eq:u-tau-rhour-rhou-relation}
\end{align}
It follows from \eqref{eq:l1-control-law}  that
\begin{equation}\label{eq:u-expression-by-mapping}
\hspace{-2mm} u \! = \!  - {\omega ^{ - 1}}{\mcC}\left( {{\eta _1} \! + \! {\tileta_m} \!+\! \bar {\mathcal H} (\theta ){\eta_2}\! +\! \bar {\mathcal H} (\theta ){\tileta_{um}} \!-\! {K_r}(\theta )r} \right), \hspace{-3mm}
\end{equation}where $\eta_i$, $\tilde \eta_{m}$  and $\tilde \eta_{um}$ are defined in \eqref{eq:eta-i-tileta-i-defn}. Equation \cref{eq:u-expression-by-mapping}  further implies 
$\omega u =  - {\mcC}\left( {\eta _1} + {\tileta_m} +  \right.  \left. \bar {\mathcal H} (\theta ){\eta_2} + \bar {\mathcal H} (\theta ){\tileta_{um}} - {K_r}(\theta )r \right).$
Therefore, the system in \eqref{eq:plant-dynamics-f12} (which is equivalent to \eqref{eq:plant-dynamics} as explained at the beginning of this proof) can be rewritten using input-output mapping as:
\begin{align}
x =\ &  {\mcG_{xm}}(\theta ){\eta _1} +  {\mcG_{xum}}(\theta ){\eta_2} - {{\mathcal H}_{xm}}(\theta ){\mcC}({\tileta_m} + \bar {\mathcal H} (\theta ){\tileta_{um}}) + {{\mathcal H}_{xm}}(\theta ){\mcC}{K_r}(\theta )r + {x_{{\rm{in}}}} \nonumber\\
  =\ & {\mcG_{xm}}(\theta ){\eta _1} +  {\mcG_{xum}}(\theta ){\eta_2} - {{\mathcal H}_{xm}}(\theta ){\mathcal F}(\theta )\tilde \sigma + {{\mathcal H}_{xm}}(\theta ){\mcC}{K_r}(\theta )r + {x_{{\rm{in}}}}, \label{eq:x-input-output-map}
\end{align}
where the second equality is due to ${\mcC}({\tileta_m} + \bar {\mathcal H} (\theta ){\tileta_{um}}) = {\mcC}(B^\dagger(\theta) + \bar {\mathcal H} (\theta ) \Bu^\dagger(\theta))\tilsigma = \mcF(\theta)\tilsigma$ according to \eqref{eq:tileta-m-um-sigma-error} and \eqref{eq:Ftheta-defn}, with $\mcH_{xm}(\theta )$, $\mcG_m(\theta )$ and $\mcG_{um}(\theta )$ introduced in \eqref{eq:lpv-system-no-nonlinearity} and \cref{eq:Gm-Gum-defn}. 
The definition of the reference system in \cref{eq:reference-system} implies
\begin{equation}\label{eq:xref-input-output-map}
   {x_\rt} = {\mcG_{xm}}(\theta ){\eta _{1\rt}} \!+\!  {\mcG_{xum}}(\theta ){\eta _{2\rt}} \!+\! {{\mathcal H}_{xm}}(\theta ){\mcC}{K_r}(\theta )r \!+\! x_{{\rm{in}}}, \end{equation}
which, together with \eqref{eq:x-input-output-map}, leads to 
\begin{align}
    {x_\rt} - x =~& {\mcG_{xm}}(\theta )({\eta _{1\rt}} - {\eta _1}) + {\mcG_{xum}}(\theta )({\eta _{{2\rt}}} - \eta_2)  + {{\mathcal H}_{xm}}(\theta ){\mathcal F}(\theta )\tilde \sigma . \label{eq:xref-x-tileta-relation}
\end{align}
The condition \cref{eq:fi-lipschitz-cond} (which holds as explained at the beginning of the proof), together with \cref{eq:xref-uref-tau-bounds} and \cref{eq:x-tau-rhor-rho-relation}, implies
\begin{equation}\label{eq:eta-ref-eta-bound}
 {\left\| {{\eta _{i\rt}} \!-\! {\eta _i}} \right\|_{{\mathcal L}_\infty ^{[0,\tau ]}}} \le\! L_{f_i}^\rho{\left\| {{x_\rt} \!-\! x} \right\|_{{\mathcal L}_\infty ^{[0,\tau ]}}},
\end{equation}
where  $\eta_{i\rt}(t,x_\rt)=f_{i}(t,x_\rt)$ and $\eta_i(t)=f_i(t,x(t))$  according to  
\cref{eq:eta_ir-def} and \cref{eq:etai-defn}. Under \cref{assump:desired-dynamics-stable-Lyapunov}, according to \cref{lem:L1-Linf-relation}, equations \cref{eq:eta-ref-eta-bound} and \eqref{eq:xref-x-tileta-relation} imply $$ {\left\| {{x_\rt} \!-\! x} \right\|_{{\mathcal L}_\infty ^{[0,\tau ]}}}
 \!\le\!   \left\| {\mcG_{xm}}(\theta )\right\|_{{{\bar {\mathcal L} }_1}}{L_{f_1}^\rho}\left\| x_\rt \!-\! x \right\|_{{\mathcal L}_\infty ^{[0,\tau ]}} \!+\! \left\| {\mcG_{xum}}(\theta ) \right\|_{{{\bar {\mathcal L} }_1}}\!\!L_{f_2}^\rho\!\left\| {{x_\rt} \!-\! x} \right\|_{{\mathcal L}_\infty ^{[0,\tau ]}} \!+\! \left\| {{{\mathcal H}_{xm}}(\theta )} \right\|_{{\bar {\mathcal L} }_1}\!\left\| {\mathcal F}(\theta )\tilde \sigma  \right\|_{{\mathcal L}_\infty ^{[0,\tau ]}},$$
which can be rewritten as $
    \zeta {\left\| {{x_\rt} \!-\! x} \right\|_{{\mathcal L}_\infty ^{[0,\tau ]}}} \le 
\left\| {{{\mathcal H}_{xm}}(\theta )} \right\|_{{{\bar {\mathcal L} }_1}}\left\| {\mathcal F}(\theta )\tilde \sigma  \right\|_{{\mathcal L}_\infty ^{[0,\tau ]}},$ where $\zeta \trieq 1 \!-\! {\left\| {{\mcG_{xm}}(\theta )} \right\|}_{{{\bar {\mathcal L} }_1}}{L_{f_1}^\rho} \!-\! {\left\| {{\mcG_{xum}}(\theta )} \right\|}_{{{\bar {\mathcal L} }_1}}{L_{f_2}^\rho}$. 
Since $\zeta$ is positive according to the stability condition in \eqref{eq:l1-stability-condition}, the preceding inequality can be further written as 
$$
 {\left\| {{x_\rt} - x}  \right\|_{{\mathcal L}_\infty ^{[0,\tau ]}}}
\le \frac{1}{\zeta}{{{{\left\| {{{\mathcal H}_{xm}}(\theta )} \right\|}_{{{\bar {\mathcal L} }_1}}}{{\left\| {{\mathcal F}(\theta )\tilde \sigma } \right\|}_{{\mathcal L}_\infty ^{[0,\tau ]}}}}} \le \frac{1}{\zeta}{{{{\left\| {{{\mathcal H}_{xm}}(\theta )} \right\|}_{{{\bar {\mathcal L} }_1}}}{\gamma _0}(T,\rho ,{\rho _u})}}{}  < \frac{1}{\zeta}{{{{\left\| {{{\mathcal H}_{xm}}(\theta )} \right\|}_{{{\bar {\mathcal L} }_1}}}{{\bar \gamma }_0}}}{\zeta},$$
where the second inequality results from Lemma~\ref{lem:F-theta-tilsigma-bound}  and \cref{eq:tilsigma-bound} (which holds according to \cref{lem:sigmatilde-xtilde-bound}), and the third inequality is due to constraint \eqref{eq:T-constraints}. Further considering \eqref{eq:underline-gamma_1-defn-constr} yields  
 \begin{equation}\label{eq:xref-x-tau-bound}
\left\| {{x_\rt} - x} \right\|_{{\mathcal L}_\infty ^{[0,\tau ]}} <\gamma_1.
 \end{equation}
On the other hand, \eqref{eq:reference-system} and \eqref{eq:u-expression-by-mapping} imply that
$
    {u_\rt} - u =  - {\omega ^{ - 1}}{\mcC}\left( {{\eta _{1\rt}} - {\eta _1}} \right) - {\omega ^{ - 1}}{\mcC}\bar {\mathcal H} (\theta )\left( {{\eta _{2\rt}} - \eta_2} \right) 
    + {\omega ^{ - 1}}{\mcC}\left( {{\tileta_m} + \bar {\mathcal H} (\theta ){\tileta_{um}}} \right).
$
Considering the bound in \eqref{eq:eta-ref-eta-bound}, we have 
\begin{align*}
 \left\| {{u_\rt} - u} \right\|_{{\mathcal L}_\infty ^{[0,\tau ]}} 
&\le {\left\| {{\omega ^{ - 1}}{\mcC}} \right\|_{{{\mathcal L}_1}}}{L_{f_1}^\rho}{\left\| {{x_\rt} - x} \right\|_{{\mathcal L}_\infty ^{[0,\tau ]}}} + {\left\| {{\omega ^{ - 1}}{\mcC}\bar {\mathcal H} (\theta )} \right\|_{{{\bar {\mathcal L} }_1}}}{L_{f_2}^\rho} {\left\| {{x_\rt} - x} \right\|_{{\mathcal L}_\infty ^{[0,\tau ]}}} +
\left\| {{\omega ^{ - 1}}} \right\|{\left\| {{\mathcal F}(\theta )\tilde \sigma } \right\|_{{\mathcal L}_\infty ^{[0,\tau ]}}} \\
& < \left( {{{\left\| {{\omega ^{ - 1}}{\mcC}} \right\|}_{{{\mathcal L}_1}}}{L_{f_1}^\rho} + {{\left\| {{\omega ^{ - 1}}{\mcC}\bar {\mathcal H} (\theta )} \right\|}_{{{\bar {\mathcal L} }_1}}}{L_{f_2}^\rho}} \right){\gamma _1} + \left\| {{\omega ^{ - 1}}} \right\|{{\bar \gamma }_0},   
\end{align*}
where the last inequality is due to \cref{eq:tilsigma-bound} (which holds according to \cref{lem:sigmatilde-xtilde-bound}), Lemma~\ref{lem:F-theta-tilsigma-bound}, and the constraint in \eqref{eq:T-constraints}. Further considering the definition of $\gamma_2$ in \eqref{eq:gamma2-defn} gives
\begin{equation}\label{eq:uref-u-tau-bound}
\left\| {{u_\rt} - u} \right\|_{{\mathcal L}_\infty ^{[0,\tau ]}} < \gamma_2.
\end{equation}
Finally, we notice that the upper bounds in \eqref{eq:xref-x-tau-bound} and \eqref{eq:uref-u-tau-bound} contradict both of the two equalities in \eqref{eq:xref-x-tau-infnorm}, indicating that {\it neither} of the two equalities can hold. This proves the bounds in \eqref{eq:xref-x-bound} and \eqref{eq:uref-u-bound}. The bounds in \cref{eq:x-bound,eq:u-bound} follow directly from \cref{eq:x-tau-rhor-rho-relation,eq:u-tau-rhour-rhou-relation}, with $\tau=+\infty$, while the bound in \eqref{eq:yref-y-bound} follows from  $y_\rt(t)-y(t) = C_m(\theta)(x_\rt(t)-x(t))$, and the bound on $C_m(\theta)$ in \eqref{eq:Am-B-Bu-Kx-Kr-bounds}. This completes the proof. \qedclosed
\end{proof}

\begin{remark}\label{rem:ref-ad-bound-discussion}
Theorem~\ref{them:x-xref-bounds} indicates that the differences between the states and inputs of the actual adaptive system and the same signals of the non-adaptive reference system in \eqref{eq:reference-system} are bounded by constants $\gamma_1$ and $\gamma_2$, respectively. As mentioned below \eqref{eq:rho-defn}, $\gamma_1$ satisfies the constraint \eqref{eq:underline-gamma_1-defn-constr} that depends on the constant  $\bar \gamma_0$, while $\gamma_2$ is dependent on $\gamma_1$ and  $\bar{\gamma}_0$ according to \eqref{eq:gamma2-defn}. Furthermore, $\bar{\gamma}_0$ is constrained by \eqref{eq:T-constraints} that involves the estimation sampling time $T$. 
Also, \eqref{eq:gamma0-T-go-to-zero} implies that  $\bar{\gamma}_0$ can be made arbitrarily small while still satisfying the constraint \eqref{eq:T-constraints} by reducing $T$. Moreover, \eqref{eq:underline-gamma_1-defn-constr} and \eqref{eq:gamma2-defn} imply that both $\gamma_1$ and $\gamma_2$ can be made arbitrarily small by reducing $\bar{\gamma}_0$.
Therefore, the difference between the inputs and states of the actual adaptive system and those of the reference system can be made arbitrarily small by reducing $T$, while the size of $T$ is limited by computational hardware and measurement noises.
\vspace{0mm} \end{remark}

After establishing the bounds between the actual closed-loop system and the reference system, we now consider the bounds between the reference system and the ideal system. The results are summarized in the following lemma. 
\begin{lemma}\label{lem:ref-id-bound}
Given the reference system in \eqref{eq:reference-system} and the ideal system in \eqref{eq:ideal-dynamics},  if \cref{assump:desired-dynamics-stable-Lyapunov}, {conditions \cref{eq:fit0-bound,eq:fi-lipschitz-cond}}, and the stability condition in  \cref{eq:l1-stability-condition} hold, we have 
\begin{align}
{\left\| {{x_\rt} - {x_{{\rm{id}}}}} \right\|_{{\mathcal L}{_\infty }}} &  \le \alpha_1(\mcC,\barH), \label{eq:xref-xid-bound}
\\
{\left\| {{u_\rt} - {u_{{\rm{id}}}}} \right\|_{{\mathcal L}{_\infty }}} &\le \alpha_2(\mcC,\barH), \label{eq:uref-uid-bound}\\
{\left\| {{y_\rt} - {y_{{\rm{id}}}}} \right\|_{{\mathcal L}{_\infty }}} & \le \alpha_3(\mcC,\barH),  \label{eq:yref-yid-bound}
\end{align}
where
{\begin{align}
    &\alpha_1(\mcC,\barH) \trieq {\left\| {{{\mathcal G}_{xm}}(\theta )} \right\|_{{{\bar {\mathcal L} }_1}}}(L_{f_1}^{\rho_r}{\rho _r} + {b_{f_1}^0}) + {\left\| {{{\mathcal G}_{xum}}(\theta )} \right\|_{{{\bar {\mathcal L} }_1}}}(L_{f_2}^{\rho_r}{\rho _r} \! + \! {b_{f_2}^0}) \!+ \!{\left\| {{{\mathcal H}_{xm}}(\theta )({\mcC} \!-\! {\mbI_m}){K_r}(\theta )} \right\|_{{{\bar {\mathcal L} }_1}}}{\left\| r \right\|_{{\mathcal L}{_\infty }}}, \label{eq:alpha_1-defn} \\
   &\alpha_2(\mcC,\barH)  \trieq {\left\| {{\omega ^{ - 1}}{\mcC}} \right\|_{{{\mathcal L}_1}}}(L_{f_1}^{\rho_r}{\rho _r} + {b_{f_1}^0}) + {\left\| {{\omega ^{ - 1}}{\mcC}\bar {\mathcal H} (\theta )} \right\|_{{{\bar {\mathcal L} }_1}}}(L_{f_2}^{\rho_r}{\rho _r} + {b_{f_2}^0})+ {\left\| {( {{\omega ^{ - 1}}{\mcC} - {\mbI_m}}){K_r}(\theta )} \right\|_{{{\bar {\mathcal L} }_1}}}{\left\| r \right\|_{{\mathcal L}{_\infty }}}, \label{eq:alpha_2-defn} \\
 & \alpha_3(\mcC,\barH)   \trieq {\left\| {{\mathcal G}{_m}(\theta )} \right\|_{{{\bar {\mathcal L} }_1}}}(L_{f_1}^{\rho_r}{\rho _r} + {b_{f_1}^0}) + {\left\| {{{\mathcal G}_{um}}(\theta )} \right\|_{{{\bar {\mathcal L} }_1}}}(L_{f_2}^{\rho_r}{\rho _r} \!+\! {b_{f_2}^0})\! + \!{\left\| {{{\mathcal H}_m}(\theta )({\mcC} \!-\! {\mbI_m}){K_r}(\theta )} \right\|_{{{\bar {\mathcal L} }_1}}}{\left\| r \right\|_{{\mathcal L}{_\infty }}}. \label{eq:alpha_3-defn} 
\end{align}}
\vspace{-5mm}\end{lemma}
\begin{proof}
The ideal dynamics in \eqref{eq:ideal-dynamics} implies 
${x_{{\rm{id}}}} = {{\mathcal H}_{xm}}(\theta ){K_r}(\theta )r + {x_{{\rm{in}}}},$
which, together with  \eqref{eq:xref-input-output-map}, leads to 
${x_\rt} - {x_{{\rm{id}}}} = {{\mathcal G}_{xm}}(\theta ){\eta _{1\rt}} + {{\mathcal G}_{xum}}(\theta ){\eta _{2\rt}} + {{\mathcal H}_{xm}}(\theta )({\mcC} - {\mbI_m}){K_r}(\theta )r.$ Therefore, 
$
{\left\| {{x_\rt} - {x_{{\rm{id}}}}} \right\|_{{\mathcal L}{_\infty }}} \le {\left\| {{\mathcal G}{_{xm}}(\theta )} \right\|_{{{\bar {\mathcal L} }_1}}}{\left\| {{\eta _{1\rt}}} \right\|_{{\mathcal L}{_\infty }}} + {\left\| {{{\mathcal G}_{xum}}(\theta )} \right\|_{{{\bar {\mathcal L} }_1}}}{\left\| {{\eta _{{2\rt}}}} \right\|_{{\mathcal L}{_\infty }}} + {\left\| {{{\mathcal H}_{xm}}(\theta )({\mcC} - {\mbI_m}){K_r}(\theta )} \right\|_{{{\bar {\mathcal L} }_1}}}{\left\| r \right\|_{{\mathcal L}{_\infty }}},$ which holds due to \cref{assump:desired-dynamics-stable-Lyapunov} and \cref{lem:L1-Linf-relation}. The preceding inequality implies \eqref{eq:xref-xid-bound}, since for $i=1,2,$
\begin{equation}\label{eq:eta_ref-bound}
    {\left\| {{\eta _{i\rt}}} \right\|_{{\mathcal L}{_\infty }}} = {\left\| {{f_{{i\rt}}}(t,{x_\rt})} \right\|_{{\mathcal L}{_\infty }}} \le L_{f_i}^{\rho_r}{\rho _r} + {b_{f_i}^0},
\end{equation}
which holds due to $\linfnorm{x_\rt}\leq \rho_r$ (Lemma~\ref{lem:xref-bounds}) and {conditions \cref{eq:fit0-bound,eq:fi-lipschitz-cond}}. On the other hand,
${y_\rt} - {y_{{\rm{id}}}} = {C}(\theta )\left( {{x_\rt} - {x_{{\rm{id}}}}} \right) = {{\mathcal G}_m}(\theta ){\eta _{1\rt}} + {{\mathcal G}_{um}}(\theta ){\eta _{2\rt}} + {{\mathcal H}_m}(\theta )({\mcC} - {\mbI_m}){K_r}(\theta )r.$
Thus,
${\left\| {{y_\rt} - {y_{{\rm{id}}}}} \right\|_{{\mathcal L}{_\infty }}} \le {\left\| {{\mathcal G}{_m}(\theta )} \right\|_{{{\bar {\mathcal L} }_1}}}{\left\| {{\eta _{1\rt}}} \right\|_{{\mathcal L}{_\infty }}} + {\left\| {{{\mathcal G}_{um}}(\theta )} \right\|_{{{\bar {\mathcal L} }_1}}}{\left\| {{\eta _{{2\rt}}}} \right\|_{{\mathcal L}{_\infty }}} + {\left\| {{{\mathcal H}_m}(\theta )({\mcC} - {\mbI_m}){K_r}(\theta )} \right\|_{{{\bar {\mathcal L} }_1}}}{\left\| r \right\|_{{\mathcal L}{_\infty }}},$
which, together with \cref{eq:eta_ref-bound}, implies \eqref{eq:yref-yid-bound}.

The input equations in \eqref{eq:reference-system} and \eqref{eq:ideal-dynamics} indicate 
${u_\rt} - {u_{{\rm{id}}}} =  - {\omega ^{ - 1}}{\mcC}{\eta _{1\rt}} - {\omega ^{ - 1}}{\mcC}\bar {\mathcal H} (\theta ){\eta _{{2\rt}}} + \left( {{\omega ^{ - 1}}{\mcC} - {\mbI_m}} \right){K_r}(\theta )r.$
Therefore, 
${\left\| {{u_\rt} - {u_{{\rm{id}}}}} \right\|_{{\mathcal L}{_\infty }}} \le {\left\| {{\omega ^{ - 1}}{\mcC}} \right\|_{{{\mathcal L}_1}}}{\left\| {{\eta _{1\rt}}} \right\|_{{\mathcal L}{_\infty }}} + {\left\| {{\omega ^{ - 1}}{\mcC}\bar {\mathcal H} (\theta )} \right\|_{{{\bar {\mathcal L} }_1}}}{\left\| {{\eta _{{2\rt}}}} \right\|_{{\mathcal L}{_\infty }}} + {\left\| {\left( {{\omega ^{ - 1}}{\mcC} - {\mbI_m}} \right){K_r}(\theta )} \right\|_{{{\bar {\mathcal L} }_1}}}{\left\| r \right\|_{{\mathcal L}{_\infty }}},$
which leads to \eqref{eq:uref-uid-bound} due to \eqref{eq:eta_ref-bound}.  \qedclosed
		  
\end{proof}

\begin{remark}\label{rem:ref-id-bound}
According to \cref{eq:alpha_1-defn,eq:alpha_2-defn,eq:alpha_3-defn}, $\alpha_i(\mcC,\barH)$ ($i=1,2,3$) depend on the filter $\mcC(s)$ and the mapping $\barH(\theta)$ (for UUA). The terms ${\left\| {{{\mathcal G}_{xm}}(\theta )} \right\|_{{{\bar {\mathcal L} }_1}}}$ and ${\left\| {{{\mathcal H}_{xm}}(\theta )({\mcC} - {\mbI_m}){K_r}(\theta )} \right\|_{{{\bar {\mathcal L} }_1}}}$ in \cref{eq:alpha_1-defn}, and ${\left\| {{{\mathcal G}_{m}}(\theta )} \right\|_{{{\bar {\mathcal L} }_1}}}$ and ${\left\| {{{\mathcal H}_{m}}(\theta )({\mcC} - {\mbI_m}){K_r}(\theta )} \right\|_{{{\bar {\mathcal L} }_1}}}$ in \cref{eq:alpha_3-defn} can be systematically reduced (to zero) by increasing the  filter bandwidth (to infinity). 
Despite being desired for better performance bounds,  
a high-bandwidth filter allows for high-frequency control signals to enter the system under fast adaptation (corresponding to small $T$), compromising the robustness (e.g., reducing stability margins). Thus, the filter presents a trade-off between robustness and performance.  More details about the role and design of the filter can be found in \red{\cite[Sections 1.3, 2.1.4, and 2.6]{naira2010l1book} for the case of using an LTI model to describe the desired dynamics}.  
 \end{remark}

With the results in \cref{them:x-xref-bounds} and \cref{lem:ref-id-bound}, we are now ready to quantify the difference between the adaptive closed-loop system and the ideal system in the following theorem. 
\begin{theorem}\label{them:x-xid-bounds}
Consider the uncertain system \eqref{eq:plant-dynamics} subject to \cref{assum:theta-thetadot-bounded}, {the control law \cref{eq:control-composition} consisting of the baseline controller \cref{eq:control-baseline-general} and the robust adaptive controller} defined via \cref{eq:state-predictor,eq:adaptive_law,eq:l1-control-law}, and  the ideal system in \eqref{eq:ideal-dynamics}. {Suppose \cref{assump:F-theta-stability,assump:B-lipschitz-bounded,assump:unknown-input-gain,assump:ft0-bound,assump:ftx-semiglobal_lipschitz,assump:ftx-rate-bounded,assump:desired-dynamics-stable-Lyapunov} hold, and 
the stability condition  \eqref{eq:l1-stability-condition} holds with a user-selected constant $\gamma_1$ that can be arbitrarily small, and the constraints  \cref{eq:underline-gamma_1-defn-constr,eq:T-constraints} hold. Then}, we have 
\begin{align}
{\left\| {{x} - {x_{{\rm{id}}}}} \right\|_{{\mathcal L}{_\infty }}} &  \le \alpha_1(\mcC,\barH) +\gamma_1, \label{eq:x-xid-bound}
\\
{\left\| {{u} - {u_{{\rm{id}}}}} \right\|_{{\mathcal L}{_\infty }}} &\le \alpha_2(\mcC,\barH) +\gamma_2, \label{eq:u-uid-bound}\\
{\left\| {{y} - {y_{{\rm{id}}}}} \right\|_{{\mathcal L}{_\infty }}} & \le \alpha_3(\mcC,\barH)+b_C \gamma_1.  \label{eq:y-yid-bound}
\end{align}
\vspace{-3mm}
\end{theorem} 
\begin{proof}
Since $x(t)-x_\idt(t) = x(t)-x_\rt(t) + x_\rt(t)-x_\idt(t),$
we have $ 
    \linfnorm{x-x_\idt} \le \linfnorm{x-x_\rt} + \linfnorm{x_\rt-x_\idt}.$
Then, the bound in \eqref{eq:x-xid-bound} immediately follows from  \eqref{eq:xref-x-bound} in Theorem~\ref{them:x-xref-bounds} and \eqref{eq:xref-xid-bound} in Lemma~\ref{lem:ref-id-bound} (since conditions \cref{eq:fi-lipschitz-cond,eq:fit0-bound} hold as explained in the proof of \cref{them:x-xref-bounds}). The bounds in \eqref{eq:u-uid-bound} and \eqref{eq:y-yid-bound} can be proved analogously. \qedclosed
\end{proof}

According to \eqref{eq:x-xid-bound} and \eqref{eq:y-yid-bound}, each of the bounds on  $\linfnorm{x-x_\idt}$ and $\linfnorm{y-y_\idt}$ contains two terms. The first term $\alpha_i(\mcC,\barH)\ (i=1,3)$ can be reduced by increasing the filter bandwidth, but  cannot be made arbitrarily small due to its dependence on $\barH(\theta)$ (see Remark~\ref{rem:ref-id-bound}).
The second term, however, can be arbitrarily reduced by decreasing the estimation sampling time $T$ in theory, although the size of T is limited by computation hardware and measurement noises in practice. (see Remark~\ref{rem:ref-ad-bound-discussion}). 
{\begin{remark}
    The bounds  in \cref{them:x-xid-bounds} are  characterized using $\linf$ norm and could be quite conservative for some individual states and/or inputs. It is possible to derive a tighter separate bound for each state or input, as done in \cite{zhao2024integrated-L1-RG}.
\end{remark}}
\section{Application to Flight Control}\label{sec:simulation-example}
We consider the short-period longitudinal dynamics of an F-16 aircraft, which involves the angle of attack (AoA) and pitch rate of the aircraft as the states and the elevator deflection as the control input.  The control objective is to track a reference AoA signal throughout a large operating envelope. 
\vspace{-3mm}
\subsection{LPV modeling and baseline controller design}\label{sec:sub-lpv-modeling-f16}
The operating envelope of the aircraft was selected as $h\in [5000,40000]\ \textup{ft}, V\in[350, 900]\ \textup{ft/s}$, where $h$ and $V$ are the altitude and the true airspeed, respectively. The dynamic pressure $\bar{q}= \frac{1}{2}\rho(h)V^2,$
where $\rho(h)$ is the air density at altitude $h$.
To derive an LPV model to capture the time-varying short-period dynamics (SPD), we used the software package in \cite{huo2006F16} that implements the nonlinear six-degree-of-freedom F-16 model, which has $13$ states. We first linearized this full nonlinear model at different trim points, then used the linearized models for the SPD to fit an LPV model. 
The scheduling parameters are chosen to be the (linearly) scaled dynamic pressure, $\barqs \in [-1,1]$, and scaled airspeed, $V_s\in [-1,1]$,  where
$
\bar{q}_s \trieq 2\frac{\bar{q}-\bar{q}_{\min} }{\bar{q}_{\max}-\bar{q}_{\min} }-1,\quad V_s \trieq  2\frac{V-V_{\min} }{V_{\max}-V_{\min} }-1,
$
with $\bar{q}_{\min}= 37.1$ psf, $\bar{q}_{\max}= 830.4$ psf, $V_{\min}= 350$ ft/s, $V_{\max}= 900$ ft/s. The LPV model for the SPD is obtained as 
{
\begin{equation*}
      \dot x(t)
        = A(\theta) 
       {x(t)} + B(\theta) u_{total}(t), \   x(0) 
    = [0,0]^{\top\!}, 
\end{equation*}
    where 
\begin{align*}
    x(t) \triangleq [
         \tilalpha(t), \tilq(t)
]^{\top\!}, \ u_{total}(t)\triangleq\tilde \delta_e(t),\ 
\tilalpha(t) \triangleq \alpha(t) - \alpha_0(\theta(t)),\\  \tilq(t) \triangleq q(t) - q_0(\theta(t)), \  \tildeltae(t) \triangleq \delta_e(t) - \delta_{e,0}(\theta(t)),\ \theta(t) \triangleq [\bar{q}_s(t), V_s(t)]. 
\end{align*}}
In the preceding equations,
 $\alpha(t)$ is the AoA, $q(t)$ is the pitch rate, $\delta_e$ is the elevator deflection, and
$\alpha_0(\theta(t))$,  $q_0(\theta(t))$ and $\delta_{e,0}(\theta(t))$ are the corresponding value at the trim point defined by $\theta(t)$. \red{We used an LPV model with polynomial parameter dependence to fit those LTI models obtained from linearization. 
Through experimentation, we found that an LPV model with affine dependence on $\barqs$ and $V_s$ yielded a good accuracy in fitting the LTI models, while increasing the order of the parameter dependence did not bring much improvement of the accuracy. Use of the LTI models obtained from linearization of the nonlinear model in \cite{huo2006F16} at different operating points to fit an LPV model gives the following state-space matrices}: 
 \begingroup
\setlength\arraycolsep{1.3pt}
\renewcommand*{\arraystretch}{1.2}
\begin{align}
    A(\theta)& = \begin{bmatrix}
        -0.97& 0.94 \\
        -3.44 & -1.30 
    \end{bmatrix} -
    \begin{bmatrix}
         0.70 &  0.02\\
   2.99  & 0.89
    \end{bmatrix}
    \bar{q}_s  -
    \begin{bmatrix}
          .004 &   0\\
    -.086  & .004
    \end{bmatrix}
    V_s, \quad
  B(\theta)  = \begin{bmatrix}
        -0.002 \\
        -0.264
    \end{bmatrix} 
    +
    \begin{bmatrix}
        0.001 \\
        -0.241
    \end{bmatrix}
    \bar{q}_s.\hfill \label{eq:f16-lpv-model-concrete}
\end{align}
\endgroup 
{We set $\abs{\dot \theta_1}\leq 0.02$ and $\abs{\dot \theta_2}\leq 0.05$, and $\norm{x(0)}\leq 0.3$. It is obvious that \cref{assum:theta-thetadot-bounded} holds with sets and constants that can be easily computed}. {For instance, the sets and constants in \cref{assum:theta-thetadot-bounded} can be determined as $\Theta =[-1,1]\times [-1,1]$, $\Theta_d =[-0.02,0.02]\times[-0.05,0.05]$, $b_{\dot \theta} = 0.054.$ }


{The adaptive flight control design often proceeds with a control augmentation architecture (CAA)\cite{wise2006adaptive-flight}, which includes a baseline control law 
in the form of \cref{eq:control-baseline-general},
designed for the nominal plant to achieve desired 
closed-loop performance 
in the absence of uncertainties, and an adaptive control law $u(t)$ to handle the uncertainties. Within this CAA, $u_\textup{total}(t) = u_\bl(t)+ u(t)$}. 
As noted in \cite{biannic1997parameter}, 
it is desirable that the response of AoA to a longitudinal stick input approximates a second-order transfer function. Furthermore,
pilots like higher (lower) bandwidth at higher (lower) speeds during up-and-away flight. 
We therefore design $K_x(\theta)$ such that 
the natural frequency of the nominal closed-loop system $\dot{x} ={A_m(\theta)}x$ with $A_m(\theta)= A(\theta)+B(\theta)K_x(\theta)$
ranges from 2 rad/s at $\bar q_{\min}$ to 6 rad/s at $\bar q_{\max}$, and the damping ratio is lower bounded by 0.7. Here we used the method based on PD Lyapunov functions proposed \cite{Wu96Induced} for designing the LPV gain $K_x(\theta)$, where pole placement constraints are enforced following \cite{chilali1996pole} to satisfy the requirements mentioned above for natural frequency and damping ratio. {Note that a matrix $P(\theta)$ was automatically generated from this design process to satisfy \cref{assump:desired-dynamics-stable-Lyapunov} (see also \cref{rem:desired-dynamics-stable-Lyapunov}). }
{The open-loop and closed-loop poles when $\barqs$ changes from $-1$ to $1$ 
with $V_s=1$, 
are \red{shown in Fig.~\ref{fig:OL_CL_poles},  where the dashed lines denote a damping ratio equal to 0.7. }.} It can be verified that \cref{assump:B-lipschitz-bounded} holds with the constants  $L_{B}$, {$L_{B^\dag}$} and {$L_{K_x}$} that can be computed numerically. 

\begin{figure}[h]
    \centering
    \includegraphics[width=0.7\textwidth]{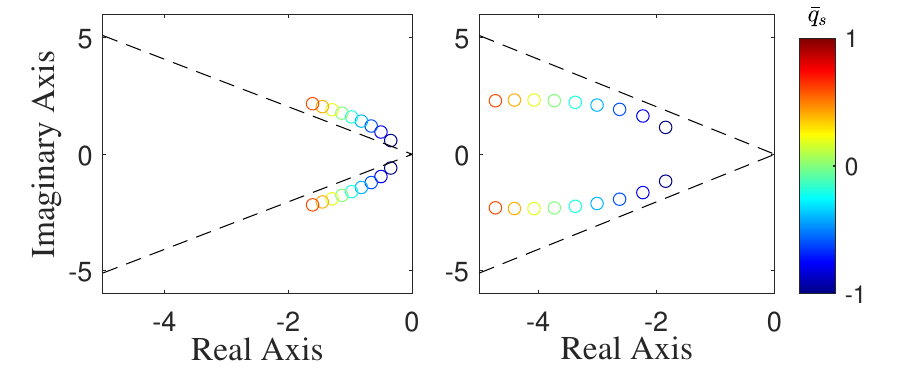}
    \caption{{\red{Open-loop (left) and closed-loop (right)  poles when $\bar{q}_s$ changes from $-1$ 
    to $1$ 
    with an increment of $0.2$}
    }}
    \label{fig:OL_CL_poles}
\end{figure}
\vspace{-3mm}
{\subsection{Uncertainties and adaptive controller design}}\label{sec:sub-f16-l1-design}

We set the uncertainties in \eqref{eq:plant-dynamics} as
\begin{equation}\label{eq:f16-uncertainties-lpv}
    \omega= 0.7,\quad f(t,x) = [ 0.02 \sin(20\pi x_1(t))+0.01\sin(\pi t), 5x_1(t)x_2(t)+0.01\cos(2\pi t)]^{\top\!},
\end{equation}
and $\Omega=[0.5,1.5]$. It is obvious that \cref{assump:unknown-input-gain} holds, and furthermore, $f(t,x)$ satisfies the inequalities in \cref{assump:ft0-bound} with a constant $b_f^0=0.01\sqrt{2}$.  
{Using the Jacobian of $f(t,x)$ with respect to $x$ and the relation between the 2-norm and Frobenius norm of a matrix, one can verify that \cref{assump:ftx-semiglobal_lipschitz} holds with $L_f^\delta= 0.16\pi^2+25\delta^2$}. Additionally,  \cref{assump:ftx-rate-bounded} holds with the constant $l_f = 0.0702$. 
{So far we have verified that  \cref{assump:B-lipschitz-bounded,assump:unknown-input-gain,assump:ft0-bound,assump:ftx-semiglobal_lipschitz,assump:ftx-rate-bounded,assump:desired-dynamics-stable-Lyapunov} all hold.} 

For the robust adaptive controller design, we chose 
$
    \ T = 0.001\ s,\ a =10,\ K=30, 
$
where 
$K$ specifies a nominal bandwidth of $30$ rad/s for the low-pass filter (see Remark~\ref{rem:filter-in-control-law}). 
We used the gridding technique \cite{Apk98} to convert the PD-LMIs involved in the baseline and robust adaptive controller design into a finite number of LMIs, which were solved by Yalmip  \cite{YALMIP} and the Mosek solver \cite{andersen2000mosek}.
{For verification of the stability condition in \eqref{eq:l1-stability-condition}, we need to compute the PPG bounds of $\mcG_{xm}(\theta)$, $\mcG_{xum}(\theta)$, and $\mcH_{xm}(\theta)\mcC K_r(\theta)$, using \cref{lem:ppg-bound-computaion-lmi}, and $\rho_\textup{in}$, using \cref{lem:xin_bound}. 
Using an affine parameterization for the decision matrix $\mrP(\theta)$, we were able to obtain $\lonenormbar{\mcG_{xm}(\theta)} =0.0509$ (with $\mu = 3.7$), and   $\lonenormbar{\mcH_{xm}(\theta)\mcC K_r(\theta)} = 6.35$ (with $\mu=3.2$).  As mentioned in \cref{rem:barH-tradeoff},  the weighting function, $\mcW(\theta)$, can be used to tune the trade-off between the performance of unmatched uncertainty attenuation (UUA) (represented by~$\lonenormbar{\mcG_{um}(\theta)}$)  and the margin left for satisfying the stability condition in \eqref{eq:l1-stability-condition} (represented by $\lonenormbar{\mcG_{xum}(\theta)}$). To focus on UUA, we selected $\mcW(\theta) = \textup{diag}(1,0)$, and obtained $\lonenormbar{\mcG_{xum}(\theta)} = 1.874$ (with $\mu=3.6$). Also, given the initial state bound $\norm{x_0}\leq \rho_0= 0.3$, we computed $\rho_\textup{in} = 3.56$ using \cref{lem:xin_bound}. With these bounds and the properties of $f(t,x)$ computed earlier, the stability condition \eqref{eq:l1-stability-condition} can be verified. As mentioned in \cref{rem:stability-condition-conservative}, the condition \cref{eq:l1-stability-condition} could be rather conservative. Therefore, to better illustrate the capability of the proposed techniques to deal with large uncertainties, we did not enforce satisfaction of condition \eqref{eq:l1-stability-condition} (and \cref{assump:F-theta-stability}) for the subsequent simulations.} 

\vspace{-.4cm}
\subsection{Simulation results using the LPV plant}\label{sec:sub-f16-lpv-simu}
We first simulate the control system using the fitted LPV model introduced in \cref{sec:sub-lpv-modeling-f16}. 
To verify the performance of the proposed controller under nonzero initialization error, i.e., $\tilx(0) \neq 0$, we  intentionally set $\hat x(0) = [\pi/180, -0.1]^{\top\!}$. 
The trajectories of the scheduling parameters were selected to be $\theta(t)=\sin(\frac{2\pi}{5}t)\times[0.5,1]^{\top\!}$, as shown in \cref{fig:lpv_theta}. We tested the AOA tracking performance using three step reference signals of different magnitudes, denoted by $r_1(t),\ r_2(t)$ and $r_3(t)$, respectively. {Considering the low bandwidth of the filter $\mcC(s)$, we chose not to filter the feedforward signal (see Remark~\ref{rem:not-filer-r})}. The state and control input trajectories are illustrated in Fig.~\ref{fig:lpv_x12} and Fig.~\ref{fig:lpv_ubl}, respectively. In Fig.~\ref{fig:lpv_x12}, the ideal response given by the ideal dynamics in \eqref{eq:ideal-dynamics} is also included. One can see that the actual  trajectories of $x_1$ (i.e., $\tilde{\alpha}$) were very close to their ideal counterparts and displayed scaled response under step references of different magnitude, despite the presence of the uncertainties in \eqref{eq:f16-uncertainties-lpv}. 
Fig.~\ref{fig:lpv_x12} also shows a relatively large discrepancy between the actual trajectory of $x_2$ (i.e., $\tilde{q}$) and its ideal counterpart, as expected. This is because we considered mitigating the effect of unmatched uncertainty on $x_1$ only when designing the mapping $\barH(\theta)$. 
\begin{figure}[h]
    \centering
\includegraphics[width=0.6\linewidth]{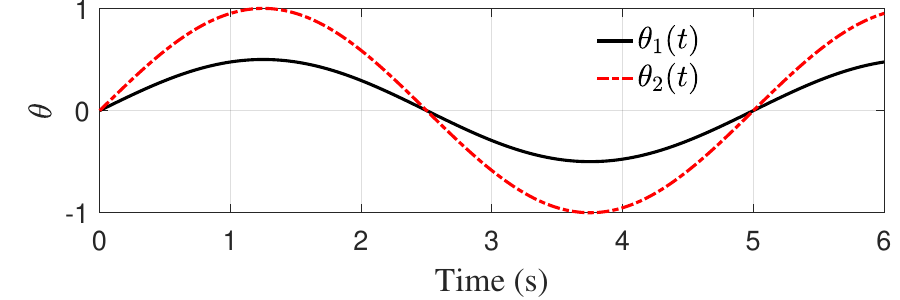}
    \caption{Trajectories of scheduling parameters}
    \label{fig:lpv_theta}
    \vspace{-3mm}
\end{figure}
\begin{figure}[h]
    \centering
    \includegraphics[width=0.6\linewidth]{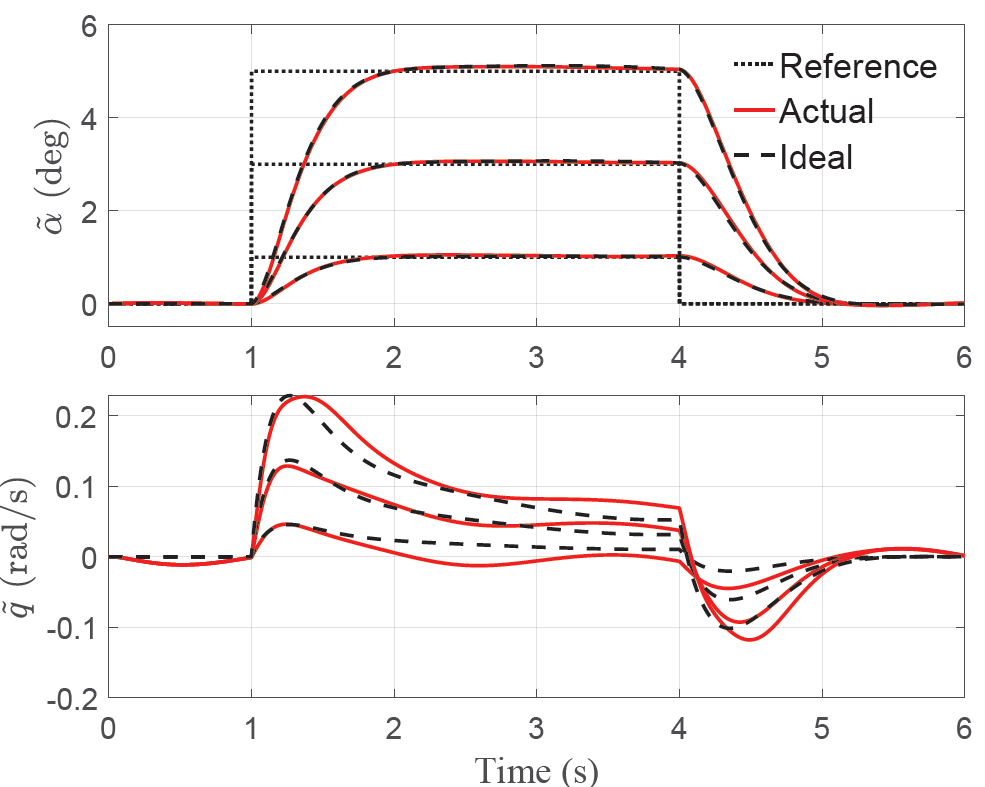} 
    \caption{\red{Trajectories of $x_1$ (top) and $x_2$ (bottom) under different references}} 
    \label{fig:lpv_x12}
    \vspace{-3mm}
\end{figure}
\begin{figure}[h]
    \centering
    \includegraphics[width=0.6\linewidth]{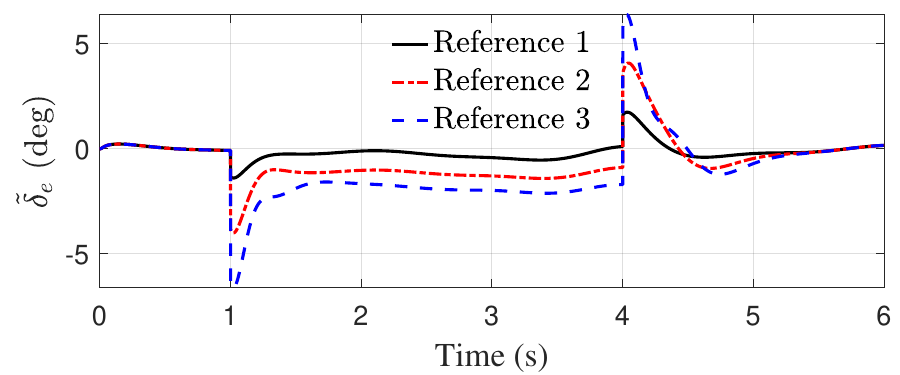}
  \caption{{Trajectories of total control inputs under different references}}
    \label{fig:lpv_ubl}
\end{figure}
Finally, Fig.~\ref{fig:lpv_sigmahat} shows the actual lumped uncertainty and its estimation using the adaptive law in \eqref{eq:adaptive_law}. We can see that due to the nonzero initialization error, i.e., $\tilx(0)\neq 0$, the estimated uncertainty $\hsigma(t)$ deviated from its true value $\sigma(t,x,u)$ (defined in \eqref{eq:sigma-txu-defn}) significantly at the first sampling interval $[0,T)$, but became very close to the true value afterward. This is consistent with Lemma~\ref{lem:sigmatilde-xtilde-bound}. For comparison, Figure~\ref{fig:lpv_x1_noL1} depicts the performance degradation when we only applied the baseline input or compensated for the matched uncertainty only (by dropping the $\hat{\sigma}_{um}$  term in the control law defined by \eqref{eq:l1-control-law}.
\begin{figure}[h!]
    \centering
    \includegraphics[width=0.6\linewidth]{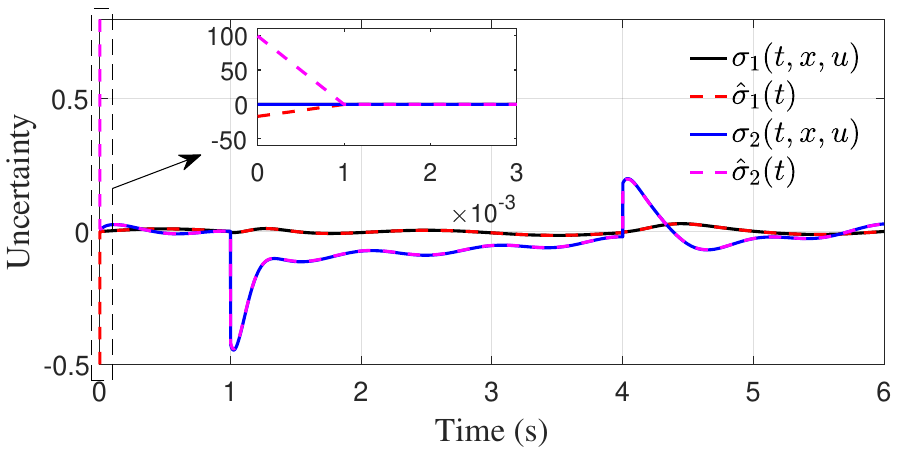}  
    \caption{{Actual and estimated (lumped) uncertainties  under the step reference of two degrees. $\hsigma_i(t)$ and $\sigma_i(t,x,u)$ denote the $i$-th element of $\hsigma(t)$ and $\sigma(t,x,u)$, respectively. }}
    \label{fig:lpv_sigmahat}
    \vspace{-3mm}
\end{figure}
\begin{figure}[h!]
    \centering
\includegraphics[width=0.6\linewidth]{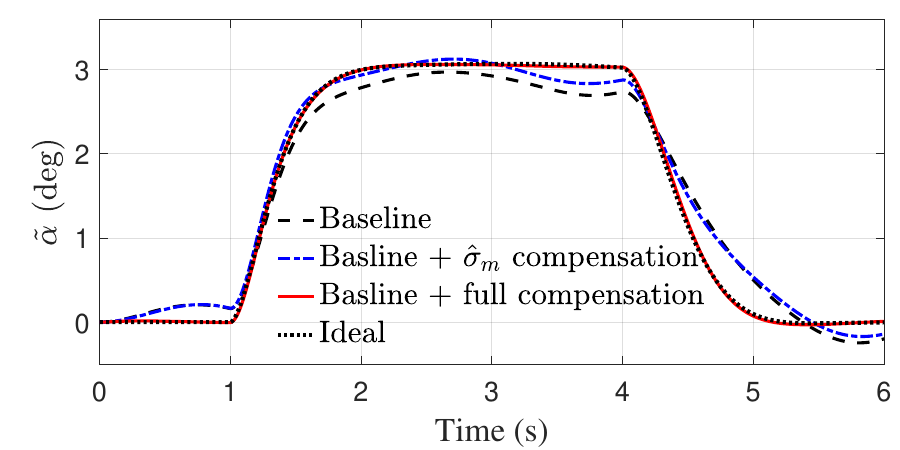}
  \caption{{Trajectories of $x_1$ under $r(t) = 3$ and different uncertainty compensation schemes}}
    \label{fig:lpv_x1_noL1}
    \vspace{-6mm}
\end{figure}

\subsection{Simulation results using a full nonlinear model}
We also used the {\it full nonlinear model} consisting of 13 states in \cite{huo2006F16}  to perform more realistic simulations.  We selected three trim points to represent the low, medium and high dynamic pressures:
\begin{itemize} 
 \item {\makebox[15mm]{Low:\hfill}} 
 $h= 40,000 \ \textup{ft}$, $V= 400$ ft/s, $\bar q =48.5\ \textup{psf}$,
    \item {\makebox[1.5cm]{Medium:\hfill}} $h= 20,000 \ \textup{ft}$,  $V= 600$ ft/s, $\bar q =228.6\ \textup{psf}$,
    \item {\makebox[1.5cm]{High:\hfill}} $h= 5,000 \ \textup{ft},\ $ $V= 900$ ft/s, $\bar q =830.4\ \textup{psf}$.
\end{itemize}
We injected the same uncertainty in \eqref{eq:f16-uncertainties-lpv}, and tested the tracking performance with a step reference signal. The AoA tracking performance is illustrated in Fig.~\ref{fig:nlin_x1}. We can see that the tracking response is different at different trim points, consistent with the design specifications. The actual trajectory of $x_1$ is still very close to its ideal counterpart, despite the presence of uncertainties and use of the simplified LPV model for controller design. 
\begin{figure}[h]
    \centering
     \vspace{-2mm}  
    \includegraphics[width=0.6\linewidth]{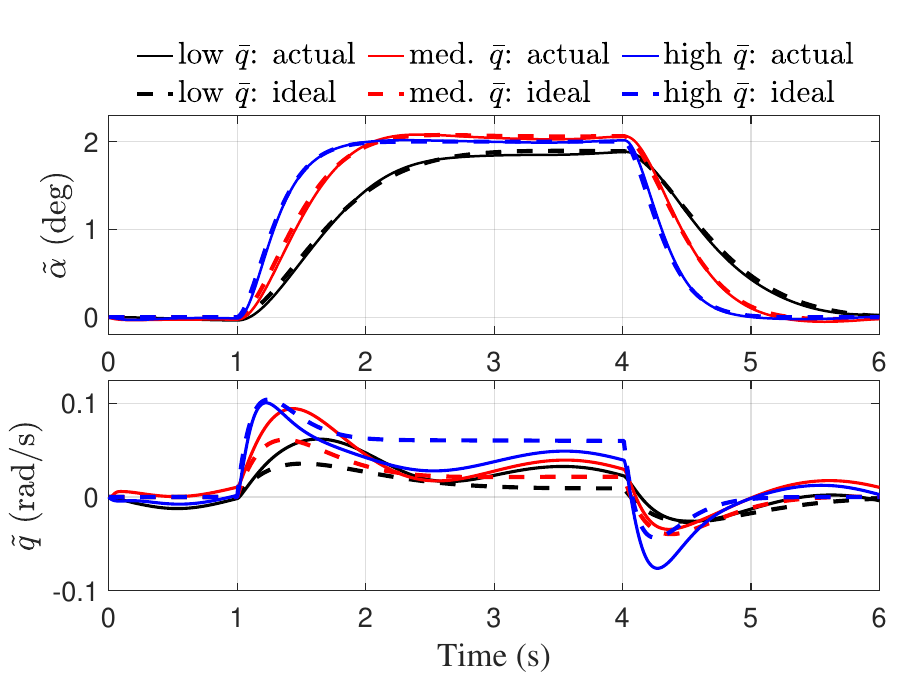}  \hspace{-5mm}  
    \caption{{Trajectories of $x_1$ (top) and $x_2$ (right) under different trim points}}
    \label{fig:nlin_x1}
\end{figure}
\begin{figure}[h]
    \centering
    \includegraphics[width=0.59\linewidth]{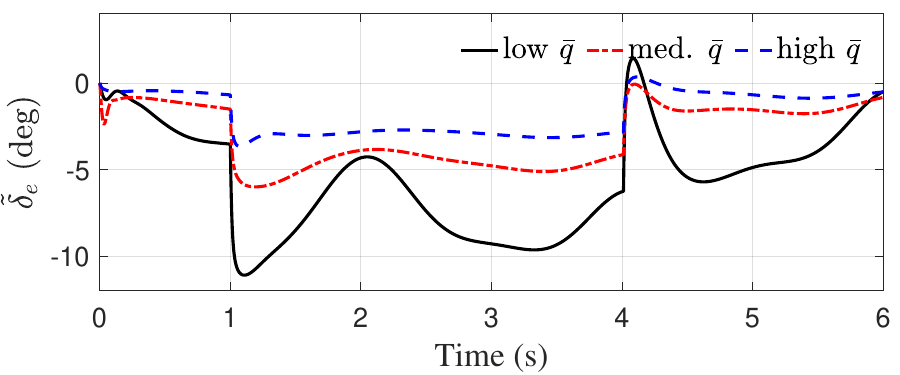}  \\
      \includegraphics[width=0.6\linewidth]{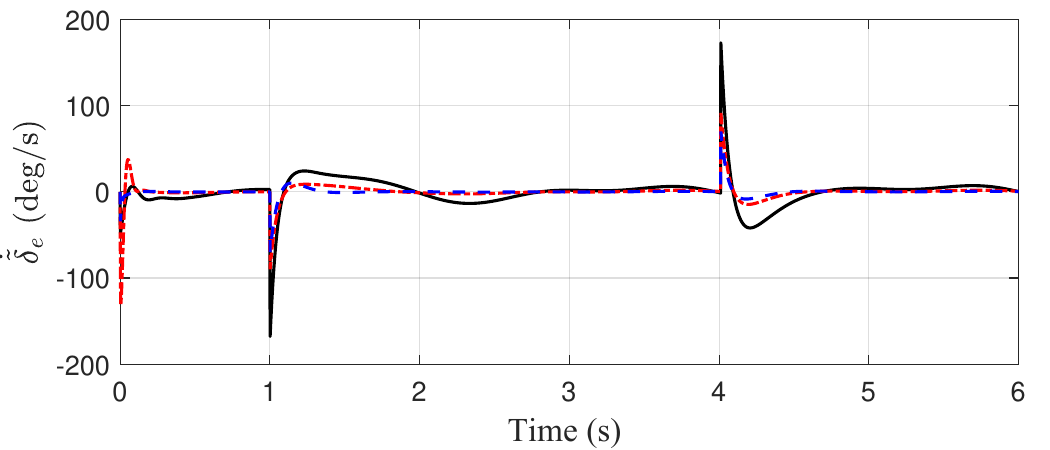}
    \caption{{Trajectories of control input and input rate under different trim points}}
    \label{fig:nlin_uBL_L1}
    \vspace{-4mm}
\end{figure}

\red{
The trajectory of the total control input and input rate is illustrated in \cref{fig:nlin_uBL_L1}.
 Note that the big jumps around 1 and 4 seconds are primarily due to the feedforward term $K_gr$ in the control signal and the step reference used in the simulation. 
Such big jumps can be mitigated by reducing the bandwidth of the filter embedded in the control law (i.e., reducing the gain $K$ in \cref{eq:l1-control-law}) and restricting the rate of the reference signal, e.g., via a low-pass filter. For instance, \cref{fig:control-rate-filtered} shows the trajectories of $x_1$ (top) and control rate (bottom) under $K=20$ (instead of $K=30$ used in other simulations) and the original reference signal filtered by a low-pass filter $\frac{10}{s+10}$. We can see that the magnitude of the control rate is reduced significantly and is roughly within the rate limits ($\pm$60 deg/s) of a real F-16 aircraft \cite[pp.210]{nguyen1979simulator-F-16}, while the compromise of tracking performance due to the reduced $K$ value is almost invisible.
 \begin{figure}[H]
     \centering
       \includegraphics[width=0.6\linewidth]{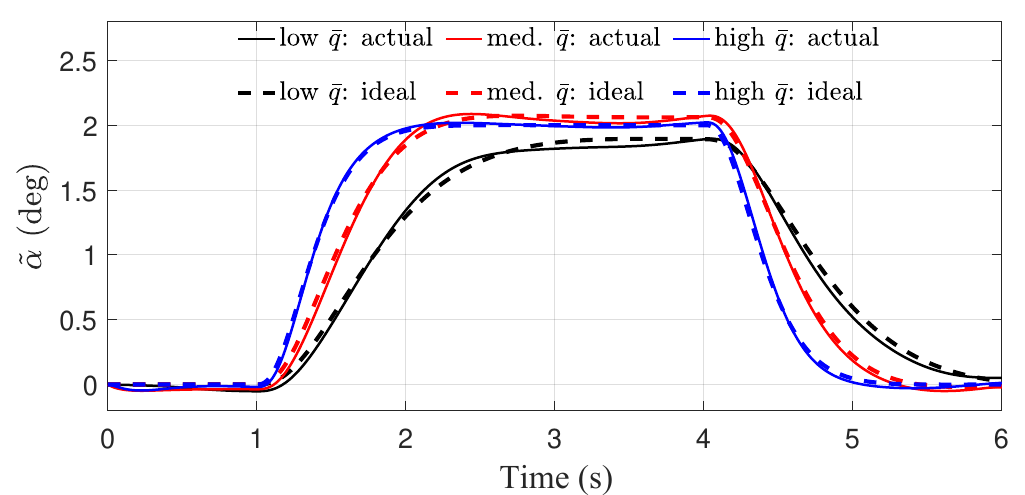}
     \includegraphics[width=0.6\linewidth]{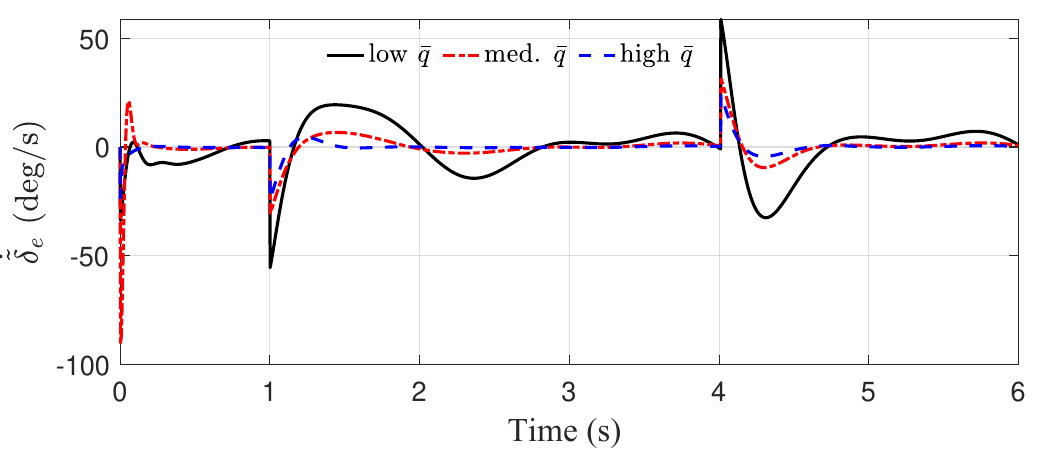}
     \caption{\red{Trajectories of $x_1$ (top) and control input rate (bottom) in nonlinear simulations under $K=20$ and the reference signal filtered by a low-pass filter $\frac{10}{s+10}$}}
     \label{fig:control-rate-filtered}
 \end{figure}}

For comparison, Fig.~\ref{fig:nlin_x1_no_comp} depicts the significant degradation when there was no compensation for the uncertainties or compensation for the matched uncertainty only.
\begin{figure}[H]
    \centering
    \includegraphics[width=0.6\linewidth]{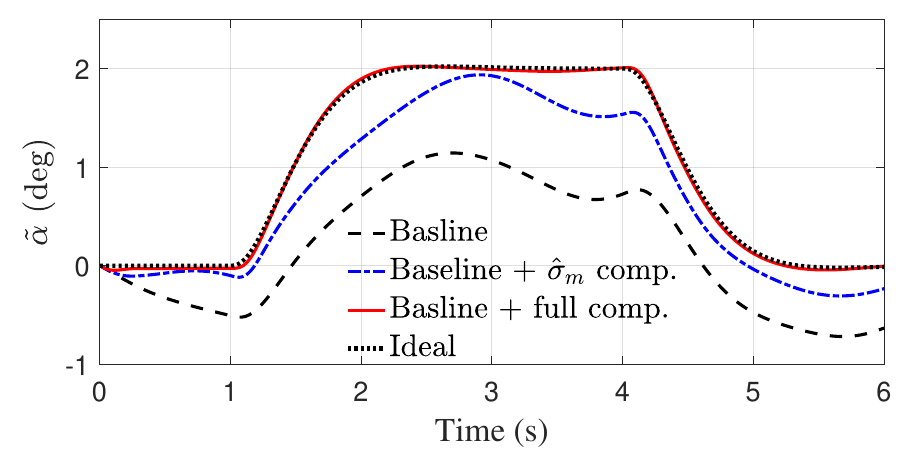}
     \vspace{-1mm}
    \caption{{Trajectories of $x_1$ under  the medium pressure scenario and different uncertainty compensation schemes}}
    \label{fig:nlin_x1_no_comp}
\end{figure}

\section{Conclusion}\label{sec:conclusion}
This paper presented a {robust adaptive control} architecture for LPV systems subject to unknown input gain and unmatched nonlinear uncertainties. Specifically, a new approach is introduced to mitigate the effect of the unmatched uncertainty on  variables of interest, e.g., system outputs, which relies on computing a feedforward LPV mapping to minimize the peak-to-peak gain from the unmatched uncertainty to those  variables of interest using linear matrix inequality (LMI) techniques. We derived the transient and steady-state performance bounds for the closed-loop system in terms of its input and output signals as compared to the same signals of a nominal system.  Simulation results on the short-period dynamics of an F-16 aircraft validate theoretical development.  Our future work includes an extension of the proposed approach to the output-feedback case. 

\section*{Acknowledgments}
This work is funded in part by the Air Force Office of Scientific Research
through grant FA9550-21-1-0411, in part by the NASA
Langley Research Center through grant 80NSSC20M0229,
and NASA University Leadership Initiative program through grant 80NSSC22M0070.

\bibliographystyle{unsrt}
\bibliography{bib/refs-pan}

\appendix
\section{Proofs}
\subsection{Proof of Lemma~\ref{lem:f-f1-f2-constants-relation}}\label{sec:prof-lem:f-f1-f2-constants-relation}
\begin{proof}
For notation brevity, we define  $\tilde \omega \trieq \omega-\mbI_m$. Equation \cref{eq:fbar0-bound} directly follows from \cref{assump:ft0-bound} and \cref{eq:f1-f2-f-relation}. 
\cref{assump:ftx-semiglobal_lipschitz} implies $\norm{f(t,x)-f(t,0)}\leq L_f^\delta\norm{x}\leq  L_f^\delta \delta $, for any $\norm{x}\leq \delta$ and any $t\geq 0$, which, together with \cref{assump:ft0-bound}, implies 
\begin{equation}\label{eq:ftx-bnd-w-delta}
    \norm{f(t,x)}\leq  L_f^\delta \delta + b_f^0, \quad \forall  \norm{x}\leq \delta, \forall t\geq 0.
\end{equation} 
With the bound $b_{\dot \theta}$  on $\norm{\dot \theta(t)}$ in \cref{assum:theta-thetadot-bounded}, we have 
\begin{equation}\label{eq:theta-t1-t2-bound}
   \norm{\theta(t_1)-\theta(t_2)}=\left\| {\int_{{t_1}}^{{t_2}} {\dot \theta } (\tau )d\tau } \right\| \leq b_{\dot \theta}\abs{t_1-t_2}.
\end{equation}
Note that 
\begin{subequations}
  \begin{align*}
    & \quad \left\| {B(\theta ({t_1}))K_x(\theta ({t_1})){x_1}\!-\!B(\theta ({t_2}))K_x(\theta ({t_2})){x_2}} \right\| \nonumber \\
   & \leq\! \left\| {\left( {B(\theta ({t_1}))\!-\!B(\theta ({t_2}))} \right)K_x(\theta ({t_1})){x_1}} \right\| \!+\! \left\| {B(\theta ({t_2}))\left( {K_x(\theta ({t_1}))\!-\!K_x(\theta ({t_2}))} \right){x_1}} \right\| \!+\! \left\| {B(\theta ({t_2}))K_x(\theta ({t_2}))\left( {{x_1}\!-\!{x_2}} \right)} \right\| \\
   &\leq\! L_B \left\| {\theta ({t_1}) \!-\! \theta ({t_2})} \right\| \!  \left\| {{x_1}} \right\|\left\| {K_x(\theta ({t_1}))} \right\| \!+\!\left\| {B(\theta ({t_2}))} \right\|{L_{K_x}}\left\| {\theta ({t_1})\!-\!\theta ({t_2})} \right\|\left\| {{x_1}} \right\| +\left\| {B(\theta ({t_2}))K_x(\theta ({t_2}))} \right\|\left\| {{x_1}\!-\!{x_2}} \right\| \\
   & \leq\! \left( {L_B \left\| {K_x(\theta ({t_1}))} \right\| \!+\!\left\| {B(\theta ({t_2}))} \right\|{L_{K_x}}} \right) \!\left\| {{x_1}} \right\|  b_{\dot \theta}\abs{t_1-t_2} +    \mathop {\max }\limits_{\theta  \in \Theta } \left\| {B(\theta )K_x(\theta )} \right\|
   \left\| {{x_1}\!-\!{x_2}} \right\| \\
   & \leq\!  \left( {L_B  b_{K_x} \!+\!b_{B}{L_{K_x}}} \right)\delta {b_{\dot \theta }}\left| {{t_1}\!-\!{t_2}} \right|  \!+\! 
  \mathop {\max }\limits_{\theta  \in \Theta } \left\| {B(\theta )K_x(\theta )} \right\|\left\| {{x_1}\!-\!{x_2}} \right\|,  
  \end{align*}
\end{subequations}
   where the second inequality is due to \cref{assump:B-lipschitz-bounded} and \cref{eq:theta-t1-t2-bound}, and the last inequality is due to the definitions in \cref{eq:Am-B-Bu-Kx-Kr-bounds} and the fact that $\norm{x_1}\leq \delta$. From the preceding inequality and the definition of ${{\bar f}}(t,x)$ in \cref{eq:f1-f2-f-relation}, we have 
\begin{align*}
    \left\| {{{\bar f}}({t_1},{x_1}) - {{\bar f}}({t_2},{x_2})} \right\| &\leq \left\| {f({t_1},{x_1}) - f({t_2},{x_2})} \right\| + \left\| {\omega  - {\mathbb{I}_m}} \right\| \left\| {B(\theta ({t_1}))K_x(\theta ({t_1})){x_1} - B(\theta ({t_2}))K_x(\theta ({t_2})){x_2}} \right\| \\
  & \leq {l_f}\left| t_1\!-\! t_2 \right| + {L_f^\delta }\left\| {{x_1} - {x_2}} \right\| + \left\| {\omega  - {\mathbb{I}_m}} \right\| \left( {L_B  b_{K_x} + b_{B}{L_{K_x}}} \right)  \delta{b_{\dot \theta }}\left| t_1\!-\! t_2 \right|\\
  & \quad + \left\| {\omega  - {\mathbb{I}_m}} \right\|\mathop {\max }\limits_{\theta  \in \Theta } \left\| {B(\theta )K_x(\theta )} \right\|\left\| {{x_1} - {x_2}} \right\| \\
  & \leq l_{{\bar f}}\left| t_1\!-\! t_2 \right| + L_{{\bar f}}^\delta\left\| {{x_1} - {x_2}} \right\|. 
\end{align*}
 Equation \cref{eq:f1-f2-f-relation} implies
\begin{subequations}\label{eq:f1f2-expression-by-f}
\begin{align}
   \hspace{-3mm} f_1(t,x) &\!=\! B^\dagger\!(\theta){\bar f}(t,x) 
  \! = \!B^\dagger\!(\theta)f(t,x) \!+\! (\omega\!-\!\mbI_m)K_x(\theta)x, \hspace{-1mm} \label{eq:f1-expression-by-f}\\
  \hspace{-3mm}  f_2(t,x) &\!=\! \Bu^\dagger(\theta)f(t,x), \label{eq:f2-expression-by-f}
\end{align}
\end{subequations}
where $B^\dagger(\theta)$ and $\Bu^\dagger(\theta)$ are the  pseudo-inverse of $B(\theta)$ and $\Bu(\theta)$, respectively. 
From \cref{assump:ft0-bound} and \cref{eq:f1f2-expression-by-f}, we obtain \cref{eq:fit0-bound}. On the other hand, from \cref{eq:f1-expression-by-f}, 
we have
\begin{align*}
    \left\| {{f_1}({t_1},{x_1}) \!-\! \!{f_1}({t_1},{x_2})} \right\| 
 &\leq  \!\left\| {{B^\dag }(\theta ({t_1}))} \right\|\!\!\left\| {f({t_1},\!{x_1}) \!-\!\! f({t_1},\!{x_2})} \right\| + \left\| {(\omega  \!-\! {\mathbb{I}_m})K_x(\theta ({t_1}))} \right\| \!\left\| {{x_1} \!\!-\! {x_2}} \right\| \\
 & \leq  \!\left( {{L_f^\delta }\left\| {{B^\dag }(\theta ({t_1}))} \right\| + \left\| {(\omega  \!-\! {\mathbb{I}_m})K_x(\theta ({t_1}))} \right\|} \right) \!\! \left\| {{x_1} 
 \!-\! {x_2}} \right\|,
\end{align*}
where the last inequality is due to \cref{assump:ftx-semiglobal_lipschitz}. Considering the definition of $L_{f_1}^\delta$ in \cref{eq:L1delta-defn}, we have  
\begin{equation}\label{eq:f1-t1-x1-x2-final}
   \left\| {{f_1}({t_1},{x_1}) \!-\! \!{f_1}({t_1},{x_2})} \right\| \leq   L_{f_1}^\delta \!\left\| {{x_1} \!- \!{x_2}} \right\|.
\end{equation}
Furthermore,
\begin{subequations}
\begin{align}
  & \left\| {{B^\dag }(\theta ({t_1}))f({t_1},{x_2}) \!- \! {B^\dag }(\theta ({t_2}))f({t_2},{x_2})} \right\| \nonumber  \\ 
  \hspace{-4mm} \leq & \left\| {{B^\dag }(\theta ({t_1}))f({t_1},{x_2}) \!- \! {B^\dag }(\theta ({t_1}))f({t_2},{x_2})} \right\|+ \left\| {{B^\dag }(\theta ({t_1}))f({t_2},{x_2}) \!- \! {B^\dag }(\theta ({t_2}))f({t_2},{x_2})} \right\| \nonumber\\ 
    \hspace{-4mm} \leq & \left\| {{B^\dag }(\theta ({t_1}))} \right\|{l_f}\left| {{t_1} \!- \! {t_2}} \right| + L_{B^\dag }b_{\dot \theta}\left| {{t_1} \!- \!{t_2}} \right|\!\left\| {f({t_2},{x_2})} \right\| \label{eq:B-pinv-f-t1t2-x2-intermed}\\ 
    \hspace{-4mm} \leq &\left( {{l_f}{\max_{\theta  \in \Theta }}\left\| {{B^\dag }(\theta )} \right\| +  L_{B^\dag }b_{\dot \theta}({b_f^0} + {L_f^\delta }\delta )} \right)\left| t_1\!-\! t_2 \right|,\label{eq:B-pinv-f-t1t2-x2-final}
\end{align}
\end{subequations}
where \cref{eq:B-pinv-f-t1t2-x2-intermed} is due to \cref{assump:ftx-rate-bounded}, \cref{eq:Bdagger-lipschitz-const} in \cref{assump:B-lipschitz-bounded} and \cref{eq:theta-t1-t2-bound}, and \cref{eq:B-pinv-f-t1t2-x2-final} results from \cref{eq:ftx-bnd-w-delta}. As a result,
\begin{align*}
    \left\| {{f_1}({t_1},{x_2}) - {f_1}({t_2},{x_2})} \right\| 
  & =\left\| {B^\dag }(\theta ({t_1}))f({t_1},{x_2}) \right.  - {B^\dag }(\theta ({t_2}))f({t_2},{x_2})  + \left. (\omega  - {\mathbb{I}_m})\left( {K_x(\theta ({t_1})) - K_x(\theta ({t_2}))} \right){x_2} \right\| 
   \\ 
   &\leq \left\| {{B^\dag }(\theta ({t_1}))f({t_1},{x_2}) - {B^\dag }(\theta ({t_2}))f({t_2},{x_2})} \right\| 
 + \left\| {(\omega  - {\mathbb{I}_m})} \right\| {L_{K_x}} b_{\dot \theta}\left| t_1\!-\! t_2 \right|\left\| {{x_2}} \right\|,
\end{align*}
 where the inequality is due to  \cref{eq:K-lipschitz-const} in \cref{assump:B-lipschitz-bounded} and \cref{eq:theta-t1-t2-bound}. Further considering \cref{eq:B-pinv-f-t1t2-x2-final} and \cref{eq:l1t-defn}, we have
 \begin{align}
 \left\| {{f_1}({t_1},{x_2}) - {f_1}({t_2},{x_2})} \right\| \leq  ~l_{f_1}\left| t_1\!-\! t_2 \right|.\label{eq:f1-t1t2-x2-final}
\end{align}
From \cref{eq:f1-t1-x1-x2-final} and \cref{eq:f1-t1t2-x2-final}, we obtain 
\begin{align}
 & \left\| {{f_1}({t_1},{x_1}) - {f_1}({t_2},{x_2})} \right\|\leq  \left\| {{f_1}({t_1},{x_1}) - {f_1}({t_1},{x_2})} \right\| + \left\| {{f_1}({t_1},{x_2}) - {f_1}({t_2},{x_2})} \right\| \leq L_{f_1}^\delta\left\| {{x_1} - {x_2}} \right\| + l_{f_1}
 \left| t_1\!-\! t_2 \right|,  \label{eq:f1-t1t2-x1x2}
\end{align}
which proves \cref{eq:fi-lipschitz-cond} for $i=1$.
Similarly, \begin{align}
  \hspace{-3mm}\left\| {{f_2}({t_1},{x_1}) \!-\! \!{f_2}({t_1},{x_2})} \right\| &= \! \left\| {\Bu^\dag }(\theta ({t_1}))\!\left( {f({t_1},{x_1})\!\! -\!\! f({t_1},{x_2})} \right)\right\| \nonumber \\ 
 &\leq   \!\left\| {{\Bu^\dag }(\theta ({t_1}))} \right\|\!\left\| {f({t_1},{x_1}) \!-\! f({t_1},{x_2})} \right\| \!  
  \leq  \left\| {{\Bu^\dag }(\theta ({t_1}))} \right\| {L_f^\delta } \left\| {{x_1} \!-\! {x_2}} \right\|\leq 
  \! {L_{f_2}^\delta} \!\left\| {{x_1} \!-\! {x_2}} \right\|,  \label{eq:f2-t1-x1-x2-final}
\end{align}
where the second last inequality is due to \cref{assump:ftx-semiglobal_lipschitz}. 
Furthermore, following a  procedure similar to that for deriving \cref{eq:B-pinv-f-t1t2-x2-intermed}, we can obtain
\begin{align*}
    \left\| {{f_2}({t_1},{x_2}) - {f_2}({t_2},{x_2})} \right\| 
 & =  \left\| {{\Bu^\dag }(\theta ({t_1}))f({t_1},{x_2}) - {\Bu^\dag }(\theta ({t_2}))f({t_2},{x_2})} \right\|   \\
  &\leq \left\| {{\Bu^\dag }(\theta ({t_1}))} \right\|{l_f}\left| t_1\!-\! t_2 \right| + L_{\Bu^\dag }b_{\dot \theta}\left| t_1\!-\! t_2 \right|\left\| {f({t_2},{x_2})} \right\| \label{eq:f2-t1t2-x2-intermed}\\ 
 &  \leq ( {{l_f}\left\| {{\Bu^\dag }(\theta ({t_1}))} \right\| + L_{\Bu^\dag }b_{\dot \theta}\left\| {f({t_2},{x_2})} \right\|})\left| {{t_1} \!-\! {t_2}} \right|.
\end{align*}

Further considering the definition of $l_{f_2}$ in \cref{eq:l2t-defn}, we have 
\begin{equation}\label{eq:f2-t1t2-x2-final}
    \left\| {{f_2}({t_1},{x_2}) - {f_2}({t_2},{x_2})} \right\| \leq  l_{f_2}\left| t_1\!-\! t_2 \right|,
\end{equation}
which, together with \cref{eq:f2-t1-x1-x2-final}, leads to \cref{eq:fi-lipschitz-cond} for $i=2$.  \qedclosed
\end{proof}
\subsection{A Preliminary Lemma for Proof of Lemma~\ref{lem:p2p-synthesis-final}}\label{sec:p2p-synthesis}
\vspace{-1mm}
\begin{lemma}\label{lem:p2p-synthesis}
Consider the LPV generalized plant in \eqref{eq:lpv-generalized-plant}, with parameter trajectories constrained by \eqref{eq:theta-thetadot-constraints}. There exists an LPV controller \eqref{eq:barH-ss-form} 
enforcing internal stability and a bound $\gamma$ on the PPG of the closed-loop system, whenever there exist PD symmetric matrices $\tilde{X}(\theta)$ and $\tilde{Y}(\theta)$, a PD quadruple of state-space data ($\tilde{A}(\theta), \tilde{B}(\theta), \tilde C(\theta), \tilde D(\theta)$) and positive constants $\mu$ and $\upsilon$ such that the conditions \eqref{eq:p2p-lim-1} and \eqref{eq:p2p-lim-2} hold. In such case, the controller in \eqref{eq:barH-ss-form} is given by
\begin{equation}\label{eq:lpv-controller-construction}
  \left\{ \begin{array}{rl}
 A_{\bar {\mathcal H}}(\theta)  =& N^{-1} (\tilde A- \tilde Y{{\bar B}_2}{C_{\bar {\mathcal H}}}{M^ \top } + \tilde Y\bar A\tilde X  + \dot{\tilde Y}\tilde X + \dot N{M^ \top })M^{-\top}\\
{ B_{\bar {\mathcal H}}}(\theta)= &N^{-1} (\tilde B - \tilde Y{{\bar B}_2}{D_{\bar {\mathcal H}}})\\
 C_{\bar {\mathcal H}}(\theta)=& \tilde C M^ {-\top }\\
{D_{\bar {\mathcal H}}}(\theta) =&\tilde  D, 
\end{array}\right.
\end{equation}
where $M$ and $N$ are computed from the factorization
$
    N(\theta)M(\theta)^{\top\!} = \mbI_{\bar n}-\tilde X(\theta) \tilde Y(\theta): 
$
{\setlength{\arraycolsep}{2.2pt}
\begin{equation}\label{eq:p2p-lim-1}
\hspace{-3mm}
\left[ {\begin{array}{*{20}{c}}
{\! - \dot{\tilde X} \!+\! \langle {\bar A\tilde X \!+\! {{\bar B}_2}\tilde C} \rangle  \!+\! \mu \tilde X}& \star & \star \\
{\tilde A \!+\! {\bar A^ \top } \!+\! \mu {\mbI_{\bar n}}}&  { - \dot{\tilde Y} \!+\! \langle {\bar A\tilde Y} \rangle  \!+\! \mu \tilde Y} & \star \\
{{{({{\bar B}_1} \!+\! {{\bar B}_2}\tilde D{{\bar D}_{21}})}^ \top }}& \! {\bar B}_1^ \top  \!\!+\! \bar D_{21}^{\top\!}({{\tilde B}^ \top } \!\!+\! \tilde D{\bar B}_2^{\top\!} ) &{ - \upsilon {\mbI_{n - m}}}
\end{array}} \right]  < 0, 
\end{equation}
\begin{align}
\vspace{-2mm}
\left[ {\begin{array}{*{20}{c}}
{\mu \tilde X}& \star & \star & \star \\
{\mu {\mbI_{\bar n}}}&{\mu \tilde Y}& \star & \star \\
0&0&{(\gamma  - \upsilon ){\mbI_{n - m}}}& \star \\
{{{\bar C}_1}\tilde X \!+\! {{\bar D}_{12}}\tilde C}& {{{\bar C}_1}}&{{{\bar D}_{11}} \!+\! {{\bar D}_{12}}\tilde D{{\bar D}_{21}}}&{\gamma {\mbI_n}}
\end{array}} \right] > 0.  \label{eq:p2p-lim-2} 
\end{align}}Note that in \cref{eq:lpv-controller-construction,eq:p2p-lim-1,eq:p2p-lim-2}, the dependence of the  decision matrices on $\theta$ is omitted for brevity.  
\end{lemma} 
The proof can be readily obtained by extending the conditions for bounding the PPG of LTI systems \cite[Section~IV.E]{Scherer97Multi} to LPV systems, and considering the fact that $\bar C_2 =0$ and $\bar D_{11} =0$ according to \eqref{eq:lpv-generalized-plant}. The extension is achieved through the use of a PD Lyapunov function that depends on $\tilde X$ and $\tilde Y$. \qedclosed
\begin{remark}\label{rem:practical-validity}
Due to the derivative terms, $\dot{\tilde Y}(\theta)$ and $\dot{\tilde N}(\theta)$, in \eqref{eq:lpv-controller-construction}, implementation of the controller needs the derivative of scheduling parameters, $\dot \theta$, which is often not available in practice. One way to remove the dependence on $\dot \theta$ is to impose structures on the Lyapunov matrices , i.e., making either $\tilde X$ or $\tilde Y$ parameter-independent \cite[Section~III]{Apk98}. However, enforcing this structure will induce conservatism. Motivated by the idea in \cite{sato2008inverse},  we  derive new conditions equivalent to \eqref{eq:p2p-lim-1} and \eqref{eq:p2p-lim-2} while removing the use of $\dot \theta$ in controller implementation without resorting to structured Lyapunov functions. The new conditions are given in Lemma~\ref{lem:p2p-synthesis-final}.
\end{remark}

\section{State-Space Realizations of LPV Mappings}\label{sec:ss-realization}
We now show the state-space realization of mappings $\mcG_{xm}(\theta)$ and $\mcH_{xm}(\theta)\mcF K_r(\theta)$ involved in the stability condition in \eqref{eq:l1-stability-condition}. 
The mapping $\mcH_{xm}(\theta)$ defined for \cref{eq:lpv-system-no-nonlinearity} can be rewritten using its state-space realization as {$\ssreal{A_m(\theta)}{B(\theta)}{\mbI_n}{0_{n\times m} }$}. The mappings $\mcF$ for the filter in \eqref{eq:filter-defn} and $\mbI_m-\mcF$ can be written as {$\ssreal{-\omega K}{\mbI_m}{\omega K}{0_{m\times m}}$} and {$\ssreal{-\omega K}{\mbI_m}{-\omega K}{\mbI_m}$}, respectively. Thus, the definition in \eqref{eq:Gxm-defn} implies that $\mcG_{xm}(\theta)$ can be represented as 
\begin{equation}\label{eq:Gtheta-ss-realization}
  \mcG_{xm}(\theta) =   \ssreal{ \begin{bmatrix}
    A_m(\theta) & -B(\theta)\omega K \\
    0_{m \times n} & -\omega K 
\end{bmatrix}}{
\begin{bmatrix}
    B(\theta)\\
    \mbI_m
\end{bmatrix}
}{\begin{bmatrix}
    \mbI_n & 0_{n\times m}
\end{bmatrix}
}{0_{n\times m}}.
\end{equation}
Similarly, the mapping $\mcH_{xm}(\theta)\mcF K_r(\theta)$ can be written as 
\begin{equation}\label{eq:HxCkg-ss-realization}
    \mcH_{xm}(\theta)\mcF K_r(\theta) \! = \! \ssreal{
    \begin{bmatrix}
        A_m(\theta)& B(\theta)\omega K \\
        0_{m\times n} &-\omega K
    \end{bmatrix}
    }{
    \begin{bmatrix}
        0_{n\times m} \\
       K_r(\theta)
    \end{bmatrix}
    }{
    \begin{bmatrix}
      \mbI_n& 0_{n \times m}  
    \end{bmatrix}
     }{0_{m\times m}}.
\end{equation}Note that the state space matrices in \cref{eq:Gtheta-ss-realization} and \cref{eq:HxCkg-ss-realization} are affine with respect to $\omega$.  Similarly, the  state-space realization of the mapping $\mcG_{xum}(\theta)$ can be derived and shown to have an affine relation with respect to $\omega$. With these state space realizations, the PPG bounds of $\mcG_{xm}(\theta)$, $\mcH_{xm}(\theta)\mcF K_r(\theta)$, and $\mcG_{xum}(\theta)$ can be computed using Lemma~\ref{lem:ppg-bound-computaion-lmi} in order to verify the stability condition in \eqref{eq:l1-stability-condition}. See Section~\ref{sec:sub-computation-perspective} for an explanation of how to deal with the uncertain $\omega$ in computing the PPG bounds.  

\end{document}